\documentclass[runningheads]{llncs}
\usepackage[utf8]{inputenc}

\usepackage{makecell}
\usepackage{authblk}
\usepackage{color}
\usepackage{xspace}
\usepackage{comment}
\usepackage{graphicx} 
\usepackage{amsmath, amssymb, amsfonts}
\usepackage{breqn}
\usepackage{float}
\usepackage{tikz}
\usepackage{bm}
\usepackage{ascmac}
\usepackage{multicol}
\usepackage{algorithm}
\usepackage{algpseudocode}
\usepackage{enumitem}
\usepackage[resetlabels,labeled]{multibib}
\usepackage{hyperref}
\makeatletter
\def\senbun#1(#2)#3({\@senbun(#2)(}
\def\@senbun(#1,#2)(#3,#4){%
   \@tempdima#1\p@ \advance\@tempdima#3\p@
   \divide\@tempdima\tw@
   \@tempdimb#2\p@ \advance\@tempdimb#4\p@
   \divide\@tempdimb\tw@
   \edef\@senbun@temp{\noexpand\qbezier(#1,#2)%
      (\strip@pt\@tempdima,\strip@pt\@tempdimb)(#3,#4)}%
   \@senbun@temp}
\makeatother

\newcommand{\ASY}{{\sc Asynch}}
\newcommand{\FSY}{{\sc Fsynch}}
\newcommand{\SSY}{{\sc Ssynch}}
\newcommand{\RSY}{{\sc Rsynch}}

\newcommand{\RR}{\textsc{RoundRobin}\xspace}
\newcommand{\clight}{\textit{Light}}

\newcommand{\Look}{\mathit{Look}\xspace}
\newcommand{\Compute}{\mathit{Compute}\xspace}
\newcommand{\Move}{\mathit{Move}\xspace}
\newcommand{\LCM}{\mathit{LCM}\xspace}
\newcommand{\LU}{{\mathcal{LUMI}}} 
\newcommand{\LUM}{{\mathcal{LUMI}}} 
\newcommand{\FS}{{\mathcal{FST\!A}}} 
\newcommand{\FC}{{\mathcal{FCOM}}} 
\newcommand{\OB}{{\mathcal{OBLOT}}} 
\newcommand{\LUMI}{{\mathcal{L}}} 
\newcommand{\FSTA}{{\mathcal{FS}}} 
\newcommand{\FCOM}{{\mathcal{FC}}} 
\newcommand{\OBLOT}{{\mathcal{O}}}
\newcommand{\N}{{\rm I\kern-.22em N}} 
\newcommand{\Z}{{\sf Z\kern-.42em Z}} 
\newcommand{\R}{{\rm I\kern-.22em R}}

\newcommand{\LK}{{\mathit{Look}}}
\newcommand{\CP}{{\mathit{Comp}}}
\newcommand{\M}{{\mathit{Move}}}

\newcommand{\LC}{{\mathit{LC}}}

\newcommand{\CM}{{\mathit{CM}}}

\newcommand{\T}{t}

\newcommand{\calC}{\mathcal{C}}

\newcommand{\calR}{\mathcal{R}}
\newcommand{\calA}{\mathcal{A}}
\newcommand{\calS}{\mathcal{S}}

\newcounter{Codeline}

\renewcommand{\tablename}{Table}
\renewcommand{\figurename}{Fig.}

\newtheorem{Definition}{Definition}

\title{Beyond Pairwise Comparisons: Unveiling Structural Landscape of Mobile Robot Models\thanks{This work was supported in part by JSPS KAKENHI Grant Number~25K03079.
}}
\titlerunning{Beyond Pairwise Comparisons}

\author{Shota Naito\orcidID{0009-0007-5834-4294} 
\and
Tsukasa Ninomiya\orcidID{0009-0003-3962-745X}
 \and
Koichi Wada\orcidID{0000-0002-5351-1459}
}

\authorrunning{S. Naito et al.}

\institute{Hosei University, Tokyo, Japan \\
\email{shota.naito.5x@stu.hosei.ac.jp, \\
tsukasa.ninomiya.4w@stu.hosei.ac.jp, wada@hosei.ac.jp}
}

\begin{document}

\maketitle
\begin{abstract}
Understanding the computational power of mobile robot systems is a fundamental challenge in distributed computing. While prior work has focused on pairwise separations between models, we explore how robot capabilities, light observability, and scheduler synchrony interact in more complex ways.

We first show that the \textbf{Exponential Times Expansion} (\textit{ETE}) problem is solvable only in the strongest model—\textit{fully-synchronous} robots with full mutual lights ($\LU^{F}$). We then introduce the \textbf{Hexagonal Edge Traversal} (\textit{HET}) and \textit{TAR($d$)*} problems to demonstrate how internal memory and lights interact with synchrony: under weak synchrony, internal memory alone is insufficient, while full synchrony can substitute for both lights and memory.

In the \textit{asynchronous} setting, we classify problems such as \textit{LP--MLCv}, \textit{VEC}, and \textit{ZCC} to show fine-grained separations between $\FS$ and $\FC$ robots. We also analyze \textbf{Vertex Traversal Rendezvous} (\textit{VTR}) and \textbf{Leave Place Convergence} (\textit{LP--Cv}), illustrating the limitations of internal memory in symmetric settings.

These results extend the known separation map of 14 canonical robot models, revealing structural phenomena only visible through higher-order comparisons. Our work provides new impossibility criteria and deepens the understanding of how observability, memory, and synchrony collectively shape the computational power of mobile robots.
\end{abstract}


\section{Introduction}
\subsection{Background and Motivation}

The computational power of autonomous mobile robots operating through $\Look$-$\Compute$-$\Move$ ($\LCM$) cycles has been a central topic in distributed computing. Robots are modeled as anonymous, uniform, disoriented points in the Euclidean plane. In each cycle, a robot observes the configuration ($\Look$), computes a destination ($\Compute$), and moves accordingly ($\Move$). These agents can solve distributed tasks collectively, despite their simplicity.

The weakest standard model, $\OB$, assumes robots are oblivious (no persistent memory) and silent (no explicit communication). Since its introduction~\cite{SY}, this model has been extensively studied for basic coordination tasks such as Gathering~\cite{AP,AOSY,BDT,CDN,CFPS,CP,FPSW05,ISKIDWY,SY}, Pattern Formation~\cite{FPSW08,FYOKY,SY,YS,YUKY}, and Flocking~\cite{CG,GP,SIW}. The limitations imposed by obliviousness and silence have led to extensive investigation of stronger models.
A widely studied enhancement is the luminous robot model, $\LU$~\cite{DFPSY}, where robots are equipped with a visible, persistent light (i.e., a constant-size memory) used for communication and state retention. The lights can be updated during $\Compute$ and are visible to others, enabling both memory and communication in each cycle.

To understand the individual roles of memory and communication, two submodels of $\LU$ have been introduced: $\FS$ (persistent memory without communication), and $\FC$ (communication without memory)~\cite{apdcm,BFKPSW-IandC,OWD}. Studying these four models—$\OB$, $\FS$, $\FC$, and $\LU$—clarifies the contribution of internal robot capabilities to problem solvability.

External assumptions such as synchrony and activation schedulers also play a crucial role. The \emph{semi-synchronous} (\SSY) model~\cite{SY} activates an arbitrary non-empty subset of robots in each round, who perform one atomic $\LCM$ cycle. Special cases include the \emph{fully-synchronous} (\FSY) model, where all robots are activated every round, and the \emph{restricted synchronous} (\RSY) model~\cite{BFKPSW-IandC}, where activations in consecutive rounds are disjoint.

In contrast, the \emph{asynchronous} (\ASY) model~\cite{FPSW99} introduces arbitrary but finite delays between phases of each robot's cycle, under a fair adversarial scheduler. This model captures high uncertainty and has proven computationally weaker.

Weaker variants of \ASY\ have been studied by assuming partial atomicity: for example, \textbf{$LC$-atomic} and \textbf{$CM$-atomic} models~\cite{OWD}. These restrict delays to only parts of the cycle and allow finer-grained modeling.

\begin{figure}[H]
    \centering
    \includegraphics[width=0.75\linewidth]{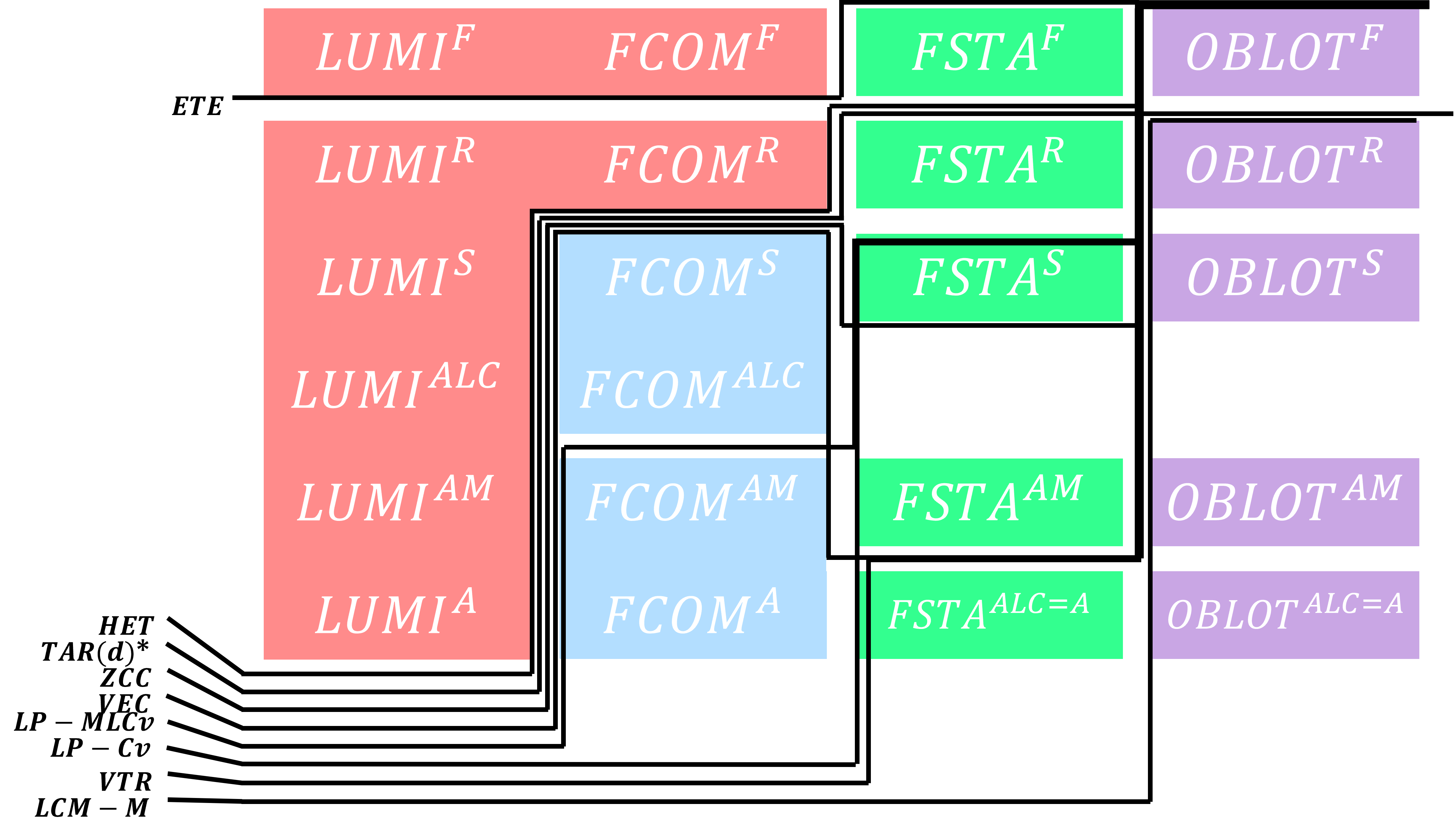}
    \caption{Hierarchy of 14 robot models combining robot capabilities and scheduler classes. Also this figure captures newly established separations.
    }
    \label{fig:boundary-line}
\end{figure}
\vspace{-0.5cm}
These internal models and external schedulers define a two-dimensional space of 14 canonical robot models. Previous work has explored their pairwise computational relationships by proving \emph{separation} (i.e., a problem solvable in one model but not the other) and identifying \emph{equivalence} or \emph{orthogonality} between some models~\cite{DFPSY,BFKPSW-IandC,FSSW23}.

\medskip
\subsection{Our Contributions}

This paper goes beyond pairwise comparisons and explores \emph{triadic relationships}, revealing new structural properties in the model hierarchy.

\begin{itemize}
    \item In Section~\ref{Sep-Synch-Scheduler}, we show that some problems are solvable only in the strongest model $\LU^F$ (equivalent to $\FC^F$), but not in either $\FS^F$ or $\LU^R$, highlighting the role of synchrony in enabling luminous computation.
    We also introduce the problems $\mathit{HET}$ and $\mathit{TAR}(d)^*$, which separate intermediate models such as $\FC^S$ and $\FS^R$ from weaker models like $\OB^F$.
    
    \item In Section~\ref{sep-asynch-scheduler}, we show that some problems require both memory and communication under \textit{asynchronous} or \textit{semi-synchronous} settings, while stronger synchrony can render these capabilities unnecessary.
\end{itemize}

These results reveal asymmetries and complex trade-offs between memory, communication, and synchrony that are not captured by pairwise analysis. They underscore the need for a global landscape perspective on distributed robot models.

\section{Models and Preliminaries}\label{sec:model}
\subsection{Robots}

We consider a set $R_n = \{ r_0, \ldots, r_{n-1} \}$ of $n > 1$ mobile computational entities, called {\em robots}, which move and operate in the Euclidean plane $\mathbb{R}^2$, where they are viewed as points.  
We simply write $R$ instead of $R_n$, except when it is necessary to emphasize the number of robots.
The robots are {\em autonomous} (i.e., they operate without central control or external supervision),  
{\em identical} (i.e., they are indistinguishable by appearance and do not have unique identifiers), and  
{\em homogeneous} (i.e., they execute the same program).  
Throughout this paper, we denote the position of a robot $r \in R$ at time $t \in \mathbb{R}_{\ge 0}$ by $r(t) \in \mathbb{R}^2$.

Each robot is equipped with a local coordinate system (in which it is always at its origin), and it is able
to observe the positions of the other robots in its local coordinate system.
The robots are {\em disoriented} (i.e., that is, there might not be consistency between the coordinate systems of different robots at the same time, or of the same robot at different times \footnote{The disorientation is said to be {\em static} if it is assumed that each local coordinate system always remains the same.}).
The robots however have {\em chirality}; that is, they agree on the same circular orientation of the plane (e.g., ``clockwise'' direction). 

A robot is endowed with motorial and computational capabilities. When active, a robot performs a $\LK$-$\Compute$-$\M{}$ ($\LCM$) cycle:
\begin{enumerate}
\item $\LK$: The robot obtains an instantaneous snapshot of the positions occupied by the robots, expressed in its local coordinate system. It cannot detect whether multiple robots are at the same location, nor how many (i.e., there is no {\em multiplicity detection}).
\item $\Compute$: The robot executes its built-in deterministic algorithm, identical for all robots, using the snapshot as input. The result of the computation is a destination point.
\item $\M$: The robot moves continuously in a straight line until it reaches the computed destination (i.e., movements are {\em rigid}).\footnote{
Movements are said to be \emph{non-rigid} if they may be interrupted by the adversary. 
}
If the destination is the current location, the robot stays still.
\end{enumerate}

In the standard model, $\OB$, the robots are also {\em silent}: 
they have no means of direct
communication of information to other robots; 
furthermore, they are {\em oblivious}: at the start of a cycle, a robot has no
memory of observations and computations performed in previous cycles.

In the other common model, $\LUM$,
each robot $r$ is equipped with a persistent register $\clight[r]$, called a {\em light},
whose value—referred to as its {\em color}—is drawn from a constant-sized set $C$ and is visible to all robots.
The color can be set by $r$ at the $\Compute$ operation, and it is not automatically reset at the end of a cycle.
In $\LUM$, the $\LK$ operation returns a colored snapshot; that is, it returns the set of distinct pairs $(\textit{position}, \textit{color})$
representing the states of the other robots at that time.
Note that if $|C|=1$, this case corresponds to the $\OB$ model.



 Two submodels of $\LU$ have been defined and investigated, $\FS$ and $\FC$, each offering only one
 of its two capabilities, persistent memory and 
direct means of communication, respectively.  
In $\FS$, a robot can only see the color of its own lights; thus, the color merely encodes an internal state. Therefore, robots are {\em silent}, as in $\OB$, but they are {\em finite-state}. 
In $\FC$, a robot can only see the colors of the lights of the other robots;
thus, it can communicate the color of its own lights to others,
but it cannot remember its own state (i.e., its own $k$ colors).
Hence, robots are enabled with {\em  finite-communication}  but are {\em oblivious}. 


\subsection{Schedulers,  Events}

In this subsection, we define several schedulers---\FSY, \RSY, \SSY, \ASY, 
and several subclasses of \ASY
---in terms of the activation schedule of the robots and the duration of their $\LCM$ cycles.  
These schedulers impose different constraints on the adversary, as we shall see,  
but they all share a common constraint called \emph{fairness}:  
every robot must perform the $\LK$, $\Compute$, and $\M$ operations infinitely often.

The time it takes to complete a cycle is assumed to be finite and the operations $\LK$
 and  $\Compute$  are  assumed to be instantaneous. 

We assume that each of the $\LK$, $\Compute$, and $\M$ operations completes in finite time.  
Moreover, the $\LK$ and $\Compute$ operations are assumed to be instantaneous.  
In the literature, the $\Compute$ operation has also been modeled as having some nonzero duration,  
at the end of which a robot changes its color (in the $\LU$ and $\FC$ models).  
However, we can assume without loss of generality that the $\Compute$ operation is instantaneous:  
since the robot’s color always changes exactly at the end of the $\Compute$ operation,  
its duration has no effect on the subsequent execution.

For any robot $r \in R$ and $i \in \mathbb{Z}_{> 0}$, let  
$t_L(r,i)$, $t_C(r,i)$, $t_B(r,i)$, and $t_E(r,i) \in \mathbb{R}_{\ge 0}$  
denote the times at which $r$ performs its $i$-th Look, Compute, Move-begin, and Move-end operations, respectively.  
These satisfy:  
$t_L(r,i) < t_C(r,i) < t_B(r,i) < t_E(r,i) < t_L(r,i+1)$.  
If $r$ decides to move from $p_B$ to $p_E$ at $t_C(r,i)$, it moves continuously along $[p_B, p_E]$ during $[t_B(r,i), t_E(r,i)]$, with $r(t_B(r,i)) = p_B$, $r(t_E(r,i)) = p_E$, and variable speed: for any $t_1 < t_2$ in $[t_B(r,i), t_E(r,i)]$,  
$|p_B - r(t_1)| < |p_E - r(t_2)|$ if $p_B \ne p_E$.

In the $\LU$ and $\FC$ models, if $r$ changes its light from $c_1$ to $c_2$ at time $t$ and another robot $s$ performs Compute at $t$, then $s$ observes $c_2$.





In the the {\em synchronous} setting (\SSY), also called {\em semi-synchronous} and first studied in \cite{SY},  time is divided into discrete
intervals, called {\em rounds}; in each round, a non-empty set of robots is activated and they simultaneously perform
a single  $\LK$-$\Compute$-$\M$ cycle in perfect synchronization. 
The particular  synchronous setting, where every robot is activated in every round
is called {\em fully-synchronous} (\FSY).

We define \RSY\ as the semi-synchronous scheduler in which,
after an optional finite prefix of fully synchronous rounds,
the remaining rounds activate non-empty subsets of robots 
with the constraint that any two consecutive subsets are disjoint.


In the {\em asynchronous} setting (\ASY), first studied in \cite{FPSW99}, we do not have any assumption---except for the fairness mentioned before---on the timing of each $\LK$, $\Compute$, and $\M$ operation and the duration of each $\M$ operation. 

$\T_i < \T_{i+1}, i\in \N$.  

In the rest of this subsection, we introduce subclasses of \ASY,  
classified according to the level of atomicity of the $\LK$, $\Compute$, and $\M$ operations.

\begin{itemize}
\item \textbf{$M$-atomic-\ASY}:  
Under this scheduler, each robot performs its $\M$ operation atomically; that is, no robot obtains a snapshot while any other robot is performing its $\M$ operation~\cite{DKKOW19,OWD}.  
Formally, this scheduler satisfies the following condition:  
$$ \forall r, s \in R,\ \forall i, j \in \mathbb{Z}_{> 0}: \ t_L(r,i) \notin [t_B(s,j), t_E(s,j)]. $$

\item $\CM$-{\bf atomic}-\ASY:  
Under this scheduler, each robot performs its consecutive $\Compute$ and $\M$ operations atomically;  
that is, no robot obtains a snapshot during the interval between the $\Compute$ and $\M$ operations of any other robot. 
Formally, this scheduler satisfies the following condition:  
$$ \forall r, s \in R,\ \forall i, j \in \mathbb{Z}_{> 0}: \ t_L(r,i) \notin [t_C(s,j), t_E(s,j)]. $$

\item $\LC$-{\bf atomic}-\ASY:  
Under this scheduler, each robot performs its consecutive $\LK$ and $\Compute$ operations atomically;  
that is, no robot obtains a snapshot during the interval between the $\LK$ and $\Compute$ operations of any other robot~\cite{DKKOW19,OWD}.  
Formally, this scheduler satisfies the following condition:  
$$ \forall r, s \in R,\ \forall i, j \in \mathbb{Z}_{> 0}: \ t_L(r,i) \notin (t_L(s,j), t_C(s,j)]. $$
Note that the interval $(t_L(s,j), t_C(s,j)]$ is left-open; hence, this scheduler may allow two or more robots to obtain snapshots simultaneously.

\item $\LCM$-{\bf atomic}-\ASY:
Under this scheduler, each robot performs its consecutive $\LK$, $\Compute$, and $\M$ operations atomically; that is, no robot obtains a snapshot during the interval between the $\LK$, $\Compute$, $\M$ operations of any other robot.  
\end{itemize}



It can be proved that \SSY\ and $\LCM$-{\bf atomic}-\ASY\ are equivalent\cite{FSSW23}.




 In the following, for simplicity of notation, we shall use the symbols 
 ${F}$, ${S}$, ${A}$, ${A_M}$, ${A_{CM}}$, and ${A_{LC}}$ to 
 denote the schedulers \FSY, 
\SSY, \ASY, $M$-{\bf atomic}-\ASY, $\CM$-{\bf atomic}-\ASY, and $\LC$-{\bf atomic}-\ASY, respectively.

Throughout the paper, we assume variable disorientation, chirality and rigidity, unless explicitly stated otherwise.

\subsection{Problems and Computational Relationships}
 Let ${\cal M} = \{\LU, \FC,\FS,\OB \}$  be the set of  models under investigation and 
${\cal S}= \{ F, S, A, A_{LC}, \allowbreak A_{M}, A_{CM}\}$ be the set of  schedulers 
under consideration.

 A {\em configuration} ${\cal C}(\T)$ at time $\T$ is the multiset of the $n$ pairs 
 $(r_i(\T), c_i(\T))$, where $(r_i(\T))$ is the location of robot $r_i$ at time $\T$,
 expressed in a fixed global coordinate system (unknown to the robots), and
 $c_i(\T)$ is the  color of its light at time $\T$. 

A \emph{configuration} $\calC$ of $n$ robots is a function $\gamma: \calR_n \to \mathbb{R}^2 \times C$ that specifies the location and color of each robot.  
Each location is given with respect to a fixed global coordinate system (which is unknown to the robots).  
Recall that the color set $C$ is a singleton (i.e., $|C| = 1$) in the $\OB$ model; in this case, the colors carry no information.  
A configuration that specifies only the locations of the robots, omitting their colors, is called a \emph{geometric configuration} or \emph{placement}; it is defined as a function $\gamma: \calR_n \to \mathbb{R}^2$.

A problem to be solved (or task to be performed) is described by a set of 
{\em temporal geometric predicates} 
 which implicitly define the {\em valid}  initial, intermediate, and (if existing) 
terminal. 
A terminal configuration is one in which, once reached, the robots no longer move. 
placements, 
as well as restrictions (if any) on the size  $n$ of the set $R$ of robots. 

Given a model $M \in {\cal M}$ and  a scheduler $K\in  {\cal S}$, we denote by
$M(K)$,
 the set of problems solvable 
by robots in $M$ 
under adversarial scheduler $K$.
Let $M_1, M_2\in{\cal M}$ and $K_1, K_2\in{\cal S}$.
\begin{itemize} 
\item We say that  model $M_1$  under scheduler  $K_1$
is {\em computationally not less powerful than} 
model $M_2$ under $K_2$, denoted by
$M_{1}^{K_1} \geq M_{2}^{K_2}$,
 if $M_{1}(K_1) \supseteq M_{2}(K_2)$.

\item We say that 
 $M_1$  under  $K_1$
is {\em computationally more powerful than} 
$M_2$ under $K_2$,  
denoted by
$M_{1}^{K_1} >  M_{2}^{K_2}$,
if 
$M_{1}^{K_1} \geq M_{2}^{K_2}$ 
and
$(M_{1}(K_1) \setminus M_{2}(K_2))  \neq \emptyset$.

\item We say that  $M_1$  under  $K_1$
and  $M_2$  under  $K_2$, are {\em computationally equivalent }, denoted by  
$M_{1}^{K_1} \equiv  M_{2}^{K_2}$,
if $M_{1}^{K_1} \geq M_{2}^{K_2}$ and $M_{2}^{K_2} \geq M_{1}^{K_1}$.

\item Finally, we say that 
$K_1$
$K_2$, are {\em computationally orthogonal} (or {\em incomparable}), denoted by  
$M_{1}^{K_1} \bot  M_{2}^{K_2}$, 
if 
 $(M_{1}(K_1) \setminus M_{2}(K_2))  \neq \emptyset$
 and 
 $(M_{2}(K_2) \setminus M_{1}(K_1))  \neq \emptyset$.

 \end{itemize}

\section{Known Separation Landscape}
 Let ${\cal M} = \{\LU, \FC,\FS,\OB \}$  be the set of  models under investigation and 
${\cal S}= \{ F, S, A, A_{LC}, \allowbreak A_{M}, A_{CM}\}$ be the set of  schedulers 
under consideration.

 From the definitions, 
 
\begin{theorem}
    For any $M\in{\cal M}$ and any  $K\in{\cal S}$:
    \begin{description}
        \item[(1)] $M^{F} \geq M^{S} \geq M^{A_{LC}} \geq M^{A}$  
        \item[(2)] $M^{F} \geq M^{S} \geq M^{A_{CM}} \geq M^{A_{M}} \geq M^{A}$
        \item[(3)] $\LUM^{K} \geq \FS^{K} \geq \OB^{K}$
        \item[(4)] $\LUM^{K} \geq \FC^{K} \geq \OB^{K}$
    \end{description}
\end{theorem}



As for $X \in \{\FS,\OB \}$, since  robots cannot observe the colors of the other robots, we have $X^{A_{CM}} \equiv X^{A_{M}}$ 
and  
$X^{A_{LC}} \equiv X^{A}$. 

Let us also recall the following equivalences: 
\begin{theorem} \label{th:equality}
    \begin{description}
        \item
        \item[(1)] $\LU^{F} \equiv \FC^{F}$~\cite{FSW19},
        \item[(2)] $\LU^{R} \equiv \LU^{S} \equiv \LU^{A} \equiv \FC^{R}$~\cite{DFPSY,BFKPSW-IandC},\label{th:eq-2} 
        \item[(3)] $\FC^{S} \equiv \FC^{A_{LC}}$~\cite{FSSW23}, and
        \item[(4)] $\FC^{A_{CM}} \equiv \FC^{A_{M}} \equiv \FC^{A}$~\cite{FSSW23}.
    \end{description}
\end{theorem}


Considering the equivalences above, we treat equivalent configurations as a single model represented by the strongest variant whenever possible.
However, when there is a widely used conventional name for a weaker equivalent, we follow the standard notation for clarity.
For example, although $A_{CM}$ and $A_{M}$ are equivalent for all robot models, we use the more familiar $A_{M}$ to denote both.
Similarly, when models are equal, only the strongest variant is counted in our classification, which reduces the total distinct configurations to 14, that is,
$\LU^F$, $\FS^F$, $\OB^F$,
$\LU^R$, $\FS^R$, $\OB^R$,
$\FC^S$, $\FS^S$, $\OB^S$,
$\FC^{A_M}$, $\FS^{A_M}$, $\OB^{A_M}$,
$\FS^A$, and $\OB^A$.

{\bf Table}~\ref{table1} summarizes all known results on the pairwise relations between models~\cite{FSW19,BFKPSW-IandC,FSSW23},
clearly separated into the cases for \SSY\ and stronger schedulers (\textit{synchronous}) and \SSY\ and weaker schedulers (\textit{asynchronous}).
Following these tables, we present a unified inclusion and separation diagram (Separation Map) that visually represents these results (Fig.~\ref{fig:separation-map}) .

\begin{table}[h!]
\caption{
Each cell indicates the known relation between two configurations:
\textbf{$>$} means that the left configuration strictly separates the right one,
with the concrete problem given in parentheses.
\textbf{$\bot$} indicates orthogonality, again with the separation problem noted.
\textbf{$\gets$} means the relation trivially follows from an inclusion by definition.
The top row and leftmost column 
uses abbreviations:
$\LU = \LUMI$, $\FC = \FCOM$, $\FS = \FSTA$, $\OB = \OBLOT$.
Problems used for separation are described in Sections~\ref{sec:start} to~\ref{sec:end}.
The first table shows pairs for \SSY\ and stronger schedulers (\textit{synchronous} part),
and the second table covers \SSY\ and weaker schedulers (\textit{asynchronous} part).
}

    \centering\tiny{
    \renewcommand{\arraystretch}{1.2}
    \begin{tabular}{|c|c|c|c|c|c|c|c|c|c|}
        \hline
         & $\LUMI^F$ & $\FSTA^F$ & $\OBLOT^F$ & $\LUMI^R$ & $\FSTA^R$ & $\OBLOT^R$ & $\FCOM^S$ & $\FSTA^S$ & $\OBLOT^S$\\
        \hline
        $\LUMI^F$ & - & \makecell{$>$ \\ (\textit{CYC})} & \makecell{$>$ \\ (\textit{CGE})} & \makecell{$>$ \\ ($\gets$)} & \makecell{$>$ \\ ($\gets$)} & \makecell{$>$ \\ ($\gets$)} & \makecell{$>$ \\ ($\gets$)} & \makecell{$>$ \\ ($\gets$)} & \makecell{$>$ \\ ($\gets$)}\\
        \hline
        $\FSTA^F$ & - & - & \makecell{$>$ \\ (\textit{CGE})} & \makecell{$\bot$ \\ (\textit{CGE, CYC})} & \makecell{$>$ \\ (\textit{CGE})} & \makecell{$>$ \\ ($\gets$)} & \makecell{$\bot$ \\ (\textit{CGE, CYC})} & \makecell{$>$ \\ (\textit{OSP})} & \makecell{$>$ \\ ($\gets$)}\\
        \hline
        $\OBLOT^F$ & - & - & - & \makecell{$\bot$ \\ (\textit{CGE*, CYC})} & \makecell{$\bot$ \\ (\textit{CGE*, OSP})} & \makecell{$>$ \\ (\textit{CGE*})} & \makecell{$\bot$ \\ (\textit{CGE*, CYC})} & \makecell{$\bot$ \\ (\textit{CGE*, TAR(d)})} & \makecell{$>$ \\ (\textit{CGE*})}\\
        \hline
        $\LUMI^R$ & - & - & - & - & \makecell{$>$ \\ (\textit{CYC})} & \makecell{$>$ \\ ($\gets$)} & \makecell{$>$ \\ (\textit{OSP})} & \makecell{$>$ \\ (\textit{CYC})} & \makecell{$>$ \\ ($\gets$)}\\
        \hline
        $\FSTA^R$ & - & - & - & - & - & \makecell{$>$ \\ (\textit{OSP})} & \makecell{$\bot$ \\ (\textit{OSP, CYC})} & \makecell{$>$ \\ (\textit{OSP})} & \makecell{$>$ \\ ($\gets$)}\\
        \hline
        $\OBLOT^R$ & - & - & - & - & - & - & \makecell{$\bot$ \\ (\textit{SRO, CYC})} & \makecell{$\bot$ \\ (\textit{SRO, TAR(d)})} & \makecell{$>$ \\ (\textit{SRO})}\\
        \hline
        $\FCOM^S$ & - & - & - & - & - & - & - & \makecell{$\bot$ \\ (\textit{CYC, TAR(d)})} & \makecell{$>$ \\ (\textit{CYC})}\\
        \hline
        $\FSTA^S$ & - & - & - & - & - & - & - & - & \makecell{$>$ \\ (\textit{TAR(d)})}\\
        \hline
        $\OBLOT^S$ & - & - & - & - & - & - & - & - & -\\
        \hline
    \end{tabular}}
    \label{table1}
\end{table}

\begin{table}[h!]
    \centering\tiny{
    \renewcommand{\arraystretch}{1.2}
    \begin{tabular}{|c|c|c|c|c|c|c|c|c|c|}
        \hline
         & $\LUMI^S$ & $\FCOM^S$ & $\FSTA^S$ & $\OBLOT^S$ & $\FCOM^{A_M}$ & $\FSTA^{A_M}$ & $\OBLOT^{A_M}$ & $\FSTA^A$ & $\OBLOT^A$\\
        \hline
        $\LUMI^S$ & - & \makecell{$>$ \\ (\textit{OSP})} & \makecell{$>$ \\ (\textit{OSP})} & \makecell{$>$ \\ ($\gets$)} & \makecell{$>$ \\ ($\gets$)} & \makecell{$>$ \\ ($\gets$)} & \makecell{$>$ \\ ($\gets$)} & \makecell{$>$ \\ ($\gets$)} & \makecell{$>$ \\ ($\gets$)}\\
        \hline
        $\FCOM^S$ & - & - & \makecell{$\bot$ \\ (\textit{TAR(d), CYC})} & \makecell{$>$ \\ (\textit{CYC})} & \makecell{$>$ \\ (\textit{RDV})} & \makecell{$\bot$ \\ (\textit{TAR(d), CYC})} & \makecell{$>$ \\ (\textit{CYC})} & \makecell{$\bot$ \\ (\textit{SM, CYC})} & \makecell{$>$ \\ ($CYC$)}\\
        \hline
        $\FSTA^S$ & - & - & - & \makecell{$>$ \\ (\textit{TAR(d)})} & \makecell{$\bot$ \\ (\textit{TAR(d), CYC})} & \makecell{$>$ \\ (\textit{MLCv})} & \makecell{$>$ \\ ($\gets$)} & \makecell{$>$ \\ (\textit{TAR(d)})} & \makecell{$>$ \\ ($\gets$)}\\
        \hline
        $\OBLOT^S$ & - & - & - & - & \makecell{$\bot$ \\ (\textit{MLCv, CYC})} & \makecell{$\bot$ \\ (\textit{TAR(d), MLCv})} & \makecell{$>$ \\ (\textit{MLCv})} & \makecell{$\bot$ \\ (\textit{SM, MLCv})} & \makecell{$>$ \\ (\textit{MLCv})}\\
        \hline
        $\FCOM^{A_M}$ & - & - & - & - & - & \makecell{$\bot$ \\ (\textit{TAR(d), CYC})} & \makecell{$>$ \\ (\textit{CYC})} & \makecell{$\bot$ \\ (\textit{SM, CYC})} & \makecell{$>$ \\ (\textit{CYC})}\\
        \hline
        $\FSTA^{A_M}$ & - & - & - & - & - & - & \makecell{$>$ \\ (\textit{TAR(d)})} & \makecell{$>$ \\ (\textit{TAR(d)})} & \makecell{$>$ \\ ($\gets$)}\\
        \hline
        $\OBLOT^{A_M}$ & - & - & - & - & - & - & - & \makecell{$\bot$ \\ (\textit{SM, TF})} & \makecell{$>$ \\ (\textit{TF})}\\
        \hline
        $\FSTA^A$ & - & - & - & - & - & - & - & - & \makecell{$>$ \\ (\textit{SM})}\\
        \hline
        $\OBLOT^A$ & - & - & - & - & - & - & - & - & -\\
        \hline
    \end{tabular}}
\end{table}
\clearpage
\begin{figure}[H]
        \centering \vspace{-3cm}
        \includegraphics[width=120mm]{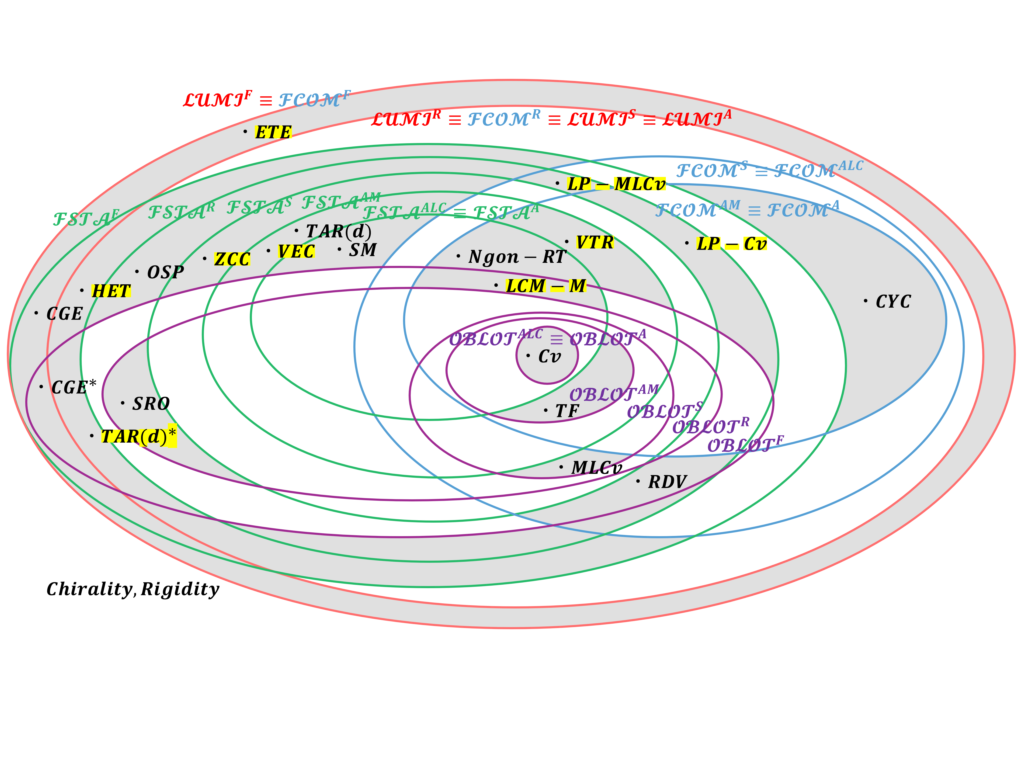}
\caption{Classification of computational problems based on their solvability across 14 canonical robot models.
    Each nested ellipse corresponds to a model class, with outer ellipses representing more powerful models.
    A dot represents a problem, and its placement indicates the set of models in which it is solvable.
    Problems highlighted in \textbf{yellow with bold labels} (e.g., \textit{ETE}, \textit{HET}, \textit{TAR(d)*}, etc.) are newly introduced in this paper.
    The shaded gray region represents model classes for which this paper identifies problems that separate them from strictly weaker models.
    That is, for each shaded region, at least one problem has been found that is solvable within that region but unsolvable in all strictly inner models.
    The remaining white areas correspond to model combinations for which the existence of such separating problems is still unknown, thus constituting open questions.}\label{fig:separation-map}
\end{figure}

\subsection{\textit{RDV} (Rendezvous)\cite{FPS}}\label{sec:start}
    Two robots are assembled at a single point that is not predetermined from an arbitrary initial configuration.

    \begin{lemma}
        $\textit{RDV} \notin \mathcal{FSTA}^{S}$.
        \label{lem:RDV_FSTA}
    \end{lemma}
    
    \begin{proof}
        The proof is by contradiction. Assume that a deterministic algorithm $\mathcal{A}$ exists that solves the rendezvous problem in finite time under the given conditions.
        
        To achieve rendezvous, the two robots, $r$ and $q$, must break symmetry and take different (asymmetric) actions at some point. However, an adversary can always prevent this using a three-pronged strategy.
        
        \begin{description}
            \item[Case 1: The destination is the other robot's position.]
            The adversary activates both robots r and q simultaneously. Using its power of Variable Disorientation, the adversary ensures that robot r moves to the initial position of q, while q moves to the initial position of r. This forces the robots to simply swap positions. The adversary can repeat this strategy indefinitely, causing the robots to perpetually swap locations without ever meeting.
            
            \item[Case 2: The destination is the midpoint.]
            The adversary activates the robots one at a time. For instance, it activates only $r$. It uses Variable Disorientation to force $r$ to move precisely to the midpoint of their current positions. In the subsequent round, the adversary activates only $q$, which is then forced to move to the new midpoint. Since the FSTA model implies that a robot's action depends on its internal state, activating a single robot ensures its state changes, compelling it to move. This process can be repeated, and while the distance between the robots decreases, it never becomes zero as they perpetually move towards a constantly shifting target.
            
            \item[Case 3: The destination is neither the other robot's position nor the midpoint.]
            The adversary can activate both robots simultaneously. In either case, it uses Variable Disorientation to enforce symmetric movements relative to the system's center. For any move computed by $r$, the adversary forces $q$ to execute a perfectly corresponding symmetric move. This preserves the global symmetry of the configuration, ensuring the robots never meet.
        \end{description}
            
        In all possible cases generated by any deterministic algorithm $\mathcal{A}$, the adversary has a guaranteed strategy to prevent the robots from rendezvous. This contradicts the initial assumption that algorithm $\mathcal{A}$ can solve the rendezvous problem. Consequently, the system's configuration will either continue as Case 3, or transition to Case 1 or 2, thereby proving that a solution to this problem is impossible.
        
        Hence, $\textit{RDV} \notin \mathcal{FSTA}^{S}$ holds.
    \end{proof}

\subsection{\textit{CYC} (Cyclic Circles)\cite{BFKPSW-IandC}}
    Let $n \geq 3$, $k = 2^{n-1}$, and $d : \mathbb{N} \to \mathbb{R}$ is a non-invertible function. The problem is to form a cyclic sequence of patterns $C, C_0, C, C_1, C, C_2, \ldots, C, C_{k-1}$ where $C$ is a pattern of $n$ robots forming a “circle” occupied by $n-1$ robots with one in the center, and $C_i$ (for $0 \leq i \leq k-1$) is a configuration where the $n-1$ circle robots are in the exact same position, but the center robot occupies a point at distance $d(i)$ from the center on the radius connecting the center to the missing robot position on the circle. In other words, the central robot moves to the designed position at distance $d(i)$ and comes back to the center, and the process repeats after the $2^{n-1}$ configurations $C_i$ have been formed.
    
\subsection{\textit{CGE} (Center of Gravity Expansion)\cite{BFKPSW-IandC}}
    Let $R$ be a set of robots with $|R| \geq 2$. Let $P = \{(x_1, y_1), (x_2, y_2), \ldots, (x_n, y_n)\}$ be the set of their positions and let $c = (c_x, c_y)$ be the coordinates of the \textit{CoG} (Center of Gravity) of $P$ at time $t = 0$. The \textit{CGE} problem requires each robot $r_i \in R$ to move from its initial position $(x_i, y_i)$ directly to $(f(x_i, c_x), f(y_i, c_y))$ where $f(a, b) = \lfloor 2a - b \rfloor$, away from the Center of Gravity of the initial configuration so that each robot doubles its initial distance from it and no longer moves.
    
\subsection{\textit{TAR($d$)} (Triangle Rotation)\cite{FSW19}}
    Let $a, b, c$ be three robots forming a triangle $ABC$, let $C$ be the circumscribed circle, and let d be a value known to the three robots. The \textit{TAR($d$)} problem requires the robots to move so to form a new triangle $A'B'C'$ with circumscribed circle C, and where $dis(A, A') = dis(B, B') = dis(C, C') = d$.
    
\subsection{\textit{OSP} (Oscillating Points)\cite{FSW19}}
    Two robots, $a$ and $b$, initially in distinct locations, alternately come closer and move further from each other. More precisely, let $d(t)$ denote the distance of the two robots at time $t$. The \textit{OSP} problem requires the two robots, starting from an arbitrary distance $d(t_0) > 0$ at time $t_0$, to move so that there exists a monotonically increasing infinite sequence time instant $t_0, t_1, t_2, \ldots$ such that:
    \begin{enumerate}
        \item $d(t_{2i+1}) < d(t_{2i}), \quad \text{and} \quad \forall h', h'' \in [t_{2i}, t_{2i+1}], \; h' < h'' \implies d(h'') \leq d(h')$,
        \item $d(t_{2i}) > d(t_{2i-1}), \quad \text{and} \quad \forall h', h'' \in [t_{2i-1}, t_{2i}], \; h' < h'' \implies d(h'') \geq d(h')$.
    \end{enumerate}

\subsection{\textit{Cv} (Convergence)\cite{FSSW23-arxiv}}
    Require robots to approach an arbitrary distance without reaching a common location.
    
\subsection{\textit{MLCv} (Monotone Line Convergence)\cite{FSSW23-arxiv}}
    The two robots, $r$ and $q$ must solve the Collisionless Line Convergence problem without ever increasing the distance between them. In other words, an algorithm solves \textit{MLCv} iff it satisfies the following predicate:
    \begin{align*}
        \textit{MLCv} \equiv 
        \big[ &\{\exists \ell \in \mathbb{R}^2, \, \forall \epsilon \geq 0, \, \exists T \geq 0, \, \forall t \geq T : 
            |r(t) - \ell| + |q(t) - \ell| \leq \epsilon \}, \\
        &\text{and } \{\forall t \geq 0 : r(t), q(t) \in \overline{r(0) q(0)} \}, \\
        &\text{and } \{\forall t \geq 0 : 
            \text{dis}(r(0), r(t)) \leq \text{dis}(r(0), q(t)), \, 
            \text{dis}(q(0), q(t)) \leq \text{dis}(q(0), r(t)) \}, \\
        &\text{and } \{\forall t' \geq t : |r(t') - q(t')| \leq |r(t) - q(t)| \} \big].
    \end{align*}
    
\subsection{\textit{TF} (Trapezoid Formation)\cite{FSSW23-arxiv}}
    Consider a set of four robots, $R={a,b,c,d}$ whose initial configuration forms a convex quadrilateral $Q=(ABCD)=(a(0)b(0)c(0)d(0))$ with one side, say $\overline{CD}$, longer than all others. The task is to transform $Q$ into a trapezoid $T$, subject to the following conditions:\\
    \noindent
    (1) If $Q$ is a trapezoid, the configuration must stay unchanged (Fig.\ref{fig:tf}(1)); i.e.,
    \begin{align*}
        \textit{TF1} \equiv 
        \big[ &\text{Trapezoid}(ABCD) \implies \{ \forall t > 0, r \in \{a, b, c, d\} : r(t) = r(0) \} \big].
    \end{align*}
    \noindent
    (2) Otherwise, without loss of generality, let $A$ be farther than $B$ from $CD$. Let $Y(A)$ (resp., $Y(B)$) denote the perpendicular lines from $A'$ (resp., $B'$) to $CD$ meeting $CD$ in $A'$ (resp. $B'$), and let $\alpha$ be the smallest angle between $\angle BAA'$ and $\angle ABB'$.\\
    \noindent
    (2.1) If $\alpha \geq \pi / 4$ then the robots must form the trapezoid shown in Fig.\ref{fig:tf}(2), where the location of $a$ is a translation of its initial one on the line $Y(A)$, and that of all other robots is unchanged; specifically,
    \begin{align*}
        \textit{TF2.1} \equiv 
        \big[ &(\alpha \geq \pi / 4) \implies \big\{ \forall t \geq 0, r \in \{b, c, d\} : r(t) = r(0), a(t) \in Y(A) \big\} \\
        &\text{and } \big\{ \exists t > 0 : \forall t' \geq t, \{ \overline{a(t') b(t')} \parallel \overline{CD} \} \text{ and } \{a(t') = a(t)\} \big\} \big].
    \end{align*}
    \noindent
    (2.2) If instead $\alpha < \pi / 4$ then the robots must form the trapezoid shown in Fig.\ref{fig:tf}(2), where the location of all robots but $b$ is unchanged, and that of $b$ is a translation of its initial one on the line $Y(B)$;
    \begin{align*}
        \textit{TF2.2} \equiv 
        \big[ &(\alpha \geq \pi / 4) \implies \big\{ \forall t \geq 0, r \in \{a, c, d\} : r(t) = r(0), b(t) \in Y(B) \big\} \\
        &\text{and } \big\{ \exists t > 0 : \forall t' \geq t, \{ \overline{a(t') b(t')} \parallel \overline{CD} \} \text{ and } \{b(t') = b(t)\} \big\} \big].
    \end{align*}
    \begin{figure}[H]
        \centering
        \includegraphics[width=7.5cm]{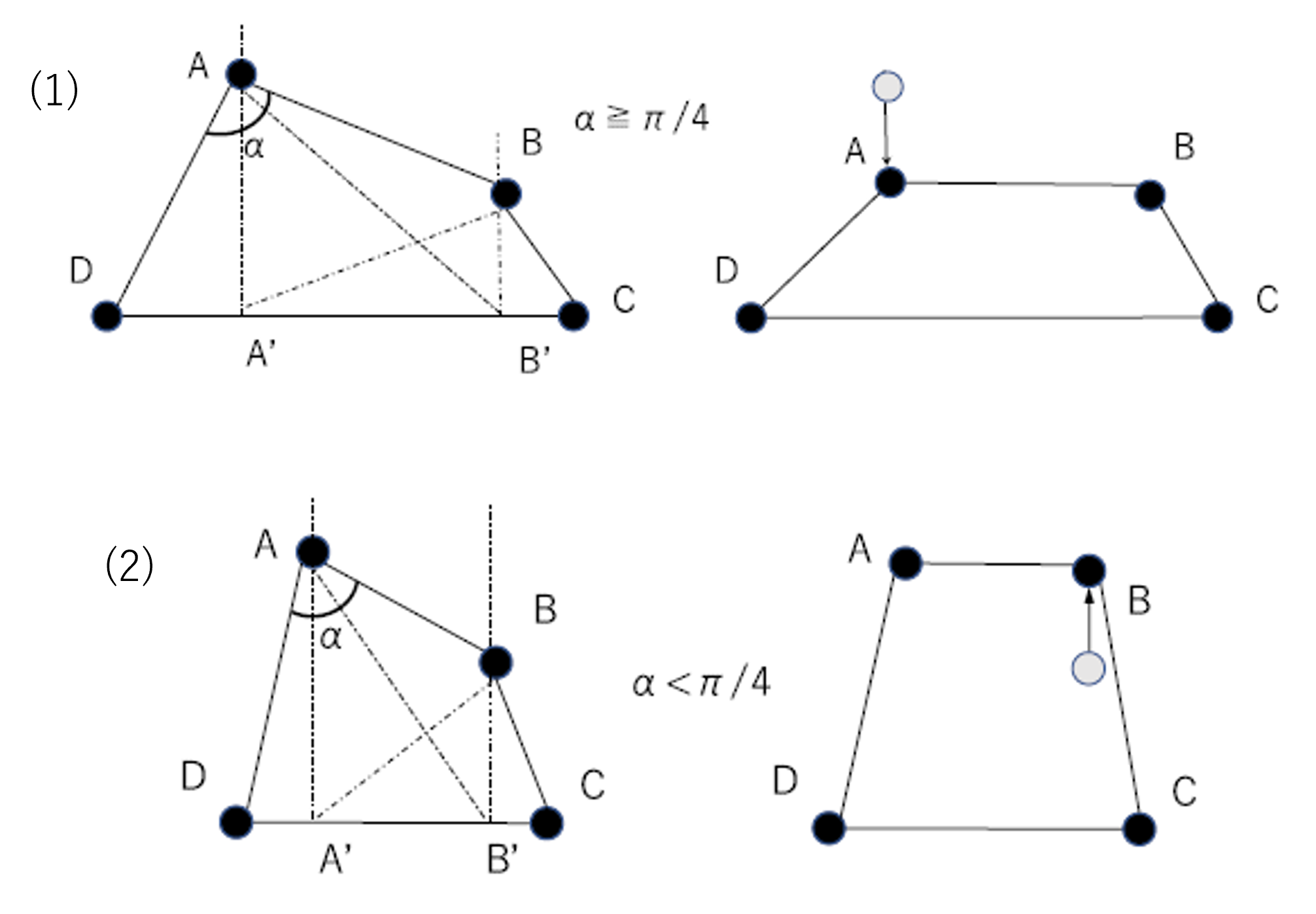}
        \caption{Movement of the \textit{TF} problem~\cite{FSSW23}}
        \label{fig:tf}
    \end{figure}

\subsection{\textit{SM} (Single Move)\cite{FSSW23-arxiv}}\label{sec:end}
    Let $a, b$ be two robot on distinct locations $a(0), b(0)$ where $r(t)$ denotes the position of $r \in \{a, b\}$ at time $t \geq 0$. The problem \textit{SM} requires each of the two robots to make exactly one non-zero movement. In other words, an algorithm solves \textit{SM} iff it satisfies the following predicate:
    
    \begin{align*}
        \textit{SM} \equiv 
        \forall r \in \{a, b\}, \exists t_1, t_2 : \left( 
        \begin{array}{l}
        r(t_1) \neq r(t_2) \land \left( 0 \leq \forall t \leq t_1 : r(t) = r(t_1) \right)
        \land \left( \forall t \geq t_2 : r(t_2) = r(t) \right) \\
        \land \left( \forall t, t' : t_1 \leq t < t' \leq t_2 \rightarrow r(t) \neq r(t') \right)
        \end{array}
        \right)
    \end{align*}

\section{Separation under Synchronous Schedulers}\label{Sep-Synch-Scheduler}
This section presents the key separation results for \SSY\ and stronger schedulers.
In particular, we introduce three new problems—\textit{ETE} (Exponential Times Expansion), 
\textit{HET} (Hexagonal Edge Traversal), and \textit{TAR($d$)*} (Infinite Triangle  Rotation  with parameter $d$)—each designed to reveal non-trivial gaps that arise when varying the degree of synchrony or capability.


Fig.~\ref{fig:boundary-line} illustrates the known separation landscape for \SSY\ and stronger schedulers,
highlighting where the \textit{ETE}, \textit{HET}, and \textit{TAR($d$)*} problems contribute to revealing strict separations 
or orthogonal relations that cannot be derived by trivial inclusion alone.

The figure visualizes which pairs of models are distinguished by each problem,
and where trivial inclusions hold due to capability or scheduler hierarchy.
New separation results introduced in this paper are marked explicitly in the diagram. 
In the figure, each line connects two models such that the upper model can solve the given problem while the lower model cannot. 
Thus, the vertical placement encodes solvability: higher nodes represent stronger models.
  \subsection{The Most Difficult Problem (ETE)}
  \label{arc:ETE}
\begin{Definition}
    \textit{ETE} (Exponential Times Expansion)
    \\Let $n \geq 3$, $k = 2^{n-1}$, and $d: \mathbb{N} \to \mathbb{R}$ is a non-invertible, meaning that for a given $i$, $d(i)$ is uniquely determined, but the inverse is not.The problem is to construct a sequence of patterns $C_0, C_1, \ldots, C_{k}$, each formed at respective times $t_0 < t_1 < \cdots < t_{k}$. Pattern $C_0$ is a circle formed by $n - 1$ robots, with one robot located at the center (see Figure~\ref{fig:ete}). Let $R $ be the set of robots excluding the central one $r_0$, and let $g_i = (g_x, g_y)$ be the coordinates of the center of gravity at time $t_i$. For each $i \neq 0$, pattern $C_i$ is formed by moving each robot $r \in R $ from position $(x, y)$ to directly position $(f(x, g_x), f(y, g_y))$, where the function $f(a, b)$ is defined as follows:
    \[
    f(a, b) = \left\lfloor b + d(i) \cdot (a - b) \right\rfloor
    \]
    In other words, each robot is scaled by a factor of $d(i)$ from the center of gravity.
\end{Definition}

\begin{figure}[H]
    \centering
    \includegraphics[width=120mm]{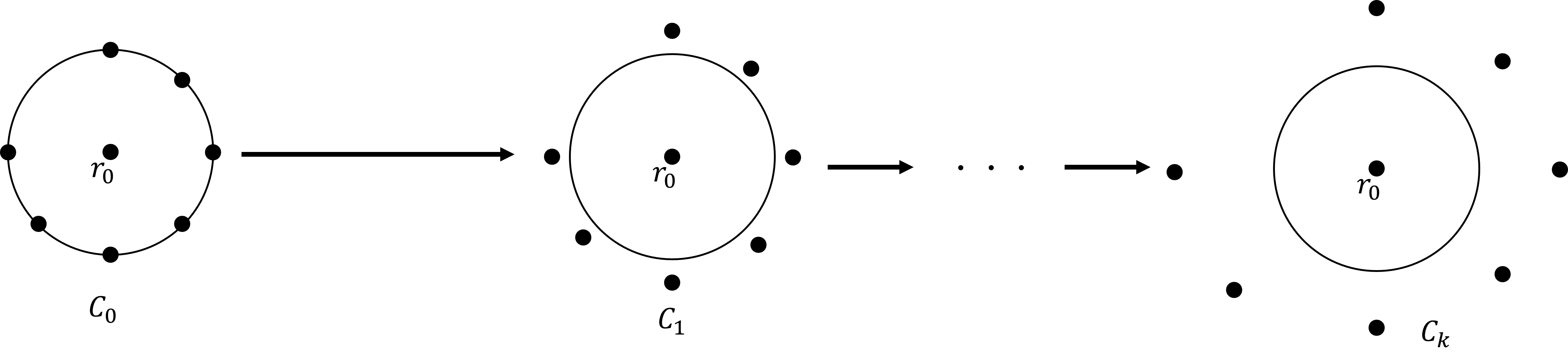}
    \caption{Movement of the \textit{ETE} problem.}
    \label{fig:ete}
\end{figure}

\begin{lemma}
    For a set of points $P = \{(x_i, y_i) \mid 1 \leq i \leq n\}$ there exist infinite sets  of  points $ Q_i \quad (1 \leq i \leq n)$ such that $(x_i, y_i) \in Q_i \quad (1 \leq i \leq n)$ and $ \forall (x_i', y_i') \in Q_i \quad $ : $f(x_i, g_x) = f(x', g_x')$, and $\quad f(y_i, g_y) = f(y', g_y')$, where $(c_x', c_y')$ is the coordinates of the \textit{ETE} of $\{(x_i', y_i') \mid 1 \leq i \leq n\}$.
    \label{lem:3}
\end{lemma}

\begin{proof}
    Given $(x_i, y_i) \in P$ 
    \[A_i = f(x_i, g_x) = \lfloor g_x + d(i)\cdot(x_i - g_x) \rfloor\]
    and
    \[B_i = f(y_i, g_y) = \lfloor g_y + d(i)\cdot(y_i - g_y) \rfloor\]
    that is :
    \[d(i)\cdot(x_i - g_x) = A_i + 1 - \delta \quad (0 < \delta \leq 1)\]
    \[d(i)\cdot(y_i - g_y) = B_i + 1 - \delta' \quad (0 < \delta' \leq 1)\]
    Consider the points $(x_i', y_i') \quad (1 \leq i \leq n)$ with $x_i' = x_i + \varepsilon \quad (0 < \varepsilon < \delta)$ and, $y_i' = y_i + \varepsilon' \quad (0 < \varepsilon' < \delta')$. We now show that $f(x_i, g_x) = f(x', g_x')$ and, $f(y_i, g_y) = f(y', g_y')$.
    
    \noindent Let us first compute $g_x'$ :
    \[ g_x' = \frac{1}{n} \sum_{i=1}^n (g_x + \varepsilon) = g_x + \varepsilon\] 
    Then
    {\small
    \[f(x_i', g_x') = \left\lfloor (g_x + \varepsilon) + d(i)\cdot \big((x_i + \varepsilon) - (g_x + \varepsilon)\big) \right\rfloor = \left\lfloor d(i)\cdot(x_i - g_x) + (g_x +\varepsilon)\right\rfloor\]
    }
    Since, by definition, $d(i)\cdot(x_i - g_x) + g_x = A_i + 1 - \delta$ and $\varepsilon < \delta$, 
    $d(i)\cdot(x_i - g_x) + (g_x + \varepsilon) < A_i + 1$.
    
    \noindent It follows that:\[ A_i \leq d(i)\cdot(x_i - g_x) + (g_x + \varepsilon) < A_i + 1\]
    and thus\[f(x_i', g_x') = A_i = f(x_i, g_x)\]
    With a similar reasoning, it follows that \[f(y_i', g_y') = B_i = f(y_i, g_y)\]
    Since this holds for an infinite number of values of $\varepsilon$ and $\varepsilon'$ the lemma follows.
\end{proof}

\begin{lemma}
    Let $n \geq 3$. $\textit{ETE} \notin \LU^{F'}$, where $F'$is any synchronous scheduler such that the first activation does not contain all robots.
    \label{lem:ETEnotinluF}
\end{lemma}

\begin{proof}
    Let $k$ ($1 \leq k < n$) be the number of activated robots in the first round. By contradiction, suppose that the problem is solvable from an arbitrary initial configuration. Consider two scenarios with $R = \{r_i \mid 1 \leq i \leq n\}$.

    \noindent In scenario A, the global coordinates of the robots are given by $r_i = (c_i, 0)$ for $1 \leq i \leq n$, where $c_i$ ($1 \leq i \leq n-1$) are distinct negative integers, $c_n$ is a natural number, and
    \[
    \sum_{i=1}^{n} c_i = 0,
    \]
    thus the center of gravity is at $(0, 0)$.
    
    \noindent In scenario B, the global coordinates of the robots are $r_i = (c_i + \varepsilon, 0)$ for $1 \leq i \leq n-1$, and $r_n = (c_n, 0)$, where $c_i$ ($1 \leq i \leq n-1$) are distinct negative integers, $c_n$ is a natural number, and
    \[
    \sum_{i=1}^{n} c_i = 0, \quad \text{with } 1 > \varepsilon > 0,
    \]
    so the center of gravity is at $\left( \frac{(n-1)\varepsilon}{n}, 0 \right)$.
    
    \noindent Consider an execution $\varepsilon$ where the scheduler activates $r_1, r_2, \dots, r_k$ at time $t_0$, with $1 \leq k \leq n-1$. In this execution, in scenario A, each $r_i$ ($1 \leq i \leq k$) moves to
    \[
    (f(x_i, g_x), f(y_i, g_y)) = (f(c_i, 0), f(0, 0)) = (d(i) \cdot c_i, 0),
    \]
    possibly changing color. The new configuration at time $t_0 + 1$ is
    \[
    r_i = (d(i) \cdot c_i, 0) \quad (1 \leq i \leq k), \quad r_i = (c_i, 0) \quad (k+1 \leq i \leq n-1), \quad r_n = (c_n, 0).
    \]
    
    \noindent Observe that the same configuration would result from the same execution in scenario B, where the center of gravity is $\left( \frac{(n-1)\varepsilon}{n}, 0 \right)$. Each $r_i$ ($1 \leq i \leq k$) moves to
    \[
    \left(f\left(c_i + \varepsilon, \frac{(n-1)\varepsilon}{n}\right), f(0, 0)\right)
    = \left( \left\lfloor \frac{(n-1)\varepsilon}{n} + d(i) \cdot \left(c_i + \left(1 - \frac{n-1}{n}\right)\varepsilon \right) \right\rfloor, 0 \right),
    \]
    assuming $\frac{(n-1)\varepsilon}{n} < 1$.
    
    \noindent Thus, the new configuration at time $t_0 + 1$ is:
    \[
    r_i = \left( \left\lfloor \frac{(n-1)\varepsilon}{n} + d(i) \cdot \left(c_i + \left(1 - \frac{n-1}{n}\right)\varepsilon \right) \right\rfloor, 0 \right) \quad (1 \leq i \leq k),
    \]
    \[
    r_i = (c_i + \varepsilon, 0) \quad (k+1 \leq i \leq n-1), \quad r_n = (c_n, 0).
    \]
    
    \noindent At time $t_0 + 1$, let the scheduler activate $R_{t_0+1}$ (possibly all of $R$), and let $r \in R$. Robot $r_n$ may be able to determine that $r_i$ ($1 \leq i \leq k$) have reached their destinations by observing their colors; i.e., it may know their positions were computed using the target function $f$. However, by \textbf{Lemma~\ref{lem:3}}, the inverse $f^{-1}$ corresponds to an infinite set of possible origins, which includes both $(c_i, 0)$ and $(c_i + \varepsilon, 0)$ for $1 \leq i \leq k$.
    
    \noindent In scenario A, the correct destination for $r_n$ would be:
    \[
    (f(c_n, 0), f(0, 0)) = (d(i) \cdot c_n, 0).
    \]
    
    \noindent In scenario B, to solve the problem, $r_n$ should instead move to:
    \[
    \left(f\left(c_n + \varepsilon, \frac{(n-1)\varepsilon}{n}\right), f(0, 0)\right)
    = \left( \left\lfloor \frac{(n-1)\varepsilon}{n} + d(i) \cdot \left(c_n + \left(1 - \frac{n-1}{n}\right)\varepsilon \right) \right\rfloor, 0 \right).
    \]
    
    \noindent Since robot $r_n$ observes the same configuration at time $t_0 + 1$ in both scenarios, it cannot distinguish between them. This leads to a contradiction.
\end{proof}

\begin{lemma}
    $\textit{ETE} \notin \FS^{F}$.
    \label{lem:2}
\end{lemma}
\begin{proof}
    To reach a contradiction, assume that there exists an algorithm $\calA$ that solves the \textit{ETE} problem. The robots following this algorithm $\calA$ must form the sequence of configuration patterns $C_0, C_1, \dots, C_{k}$. Assuming \textit{Variable Disorientation}, the adversary can force the observed distances in the snapshot taken during each robot’s $\LK$ operation to always appear as 1. In this case, there will exist configurations $C_i = C_j$ for some $i \neq j$. Thus, there will be at least two different ways to proceed to the next configuration, but the robots have no means to distinguish between them. Therefore, a number of colors equal to the number of configuration patterns would be required. However, to distinguish other robots and recognize $d(i)$, a logarithmic amount of memory is necessary. There is no way to make a decision using only a constant amount of memory. Since robots can only use a constant number of colors, no robot $r$ can recognize $d(i)$ using only its own light. As a result, the robots cannot move to the positions $(f(x, g_x), f(y, g_y))$, defined by $f(a, b) = \left\lfloor b + d(i) \cdot (a - b) \right\rfloor$, and cannot form the configuration patterns $C_0, C_1, \dots, C_{k}$. This contradicts our initial assumption.
\end{proof}

\begin{lemma}
    For $n \geq 3$, $\textit{ETE} \in \FC^{F}$.
    \label{lem:4}
\end{lemma}
\begin{proof}
    Each robot $r_i$ has 4 lights, $r_i.\text{status}$, $r_i.b$, $r_i.c$, $r_i.suc_b$. Let
    \[Status(t) = \left( r_0.status(t), \dots, r_{n-1}.status(t) \right)\]
    \[B_{value}(t) = \left( r_0.b(t), r_1.b(t), \dots, r_{n-1}.b(t) \right)\]
    \[B_{carry}(t) = \left( r_0.c(t), r_1.c(t), \dots, r_{n-1}.c(t) \right)\]
    \[B_{suc}(t) = \left( r_0.suc_b(t), r_1.suc_b(t), \dots, r_{n-1}.suc_b(t) \right)\]
    The configuration $K(t)$ at time $t$ is defined by
    \[K(t) = \left( X(t), Status(t), B_{value}(t), B_{carry}(t), B_{suc}(t) \right)\]
    Here, we define the following elements:
    \begin{itemize}
        \item $X(t)$: the set of positions of the $n$ robots at time $t$.
        \item $Status(t)$: distinguishes between $C$, $C_i$, and $C_{i+1}$.
        \item $B_{value}(t)$, $B_{carry}(t)$: encode the binary representation of $i$.
        \item $B_{suc}(t)$: a copy of the light $b$ from the succeeding robot on the circle.
    \end{itemize}
    Each robot can derive its own value of $b$ by observing the light $suc_b$ of its predecessor.
    At the initial state, the system is defined as follows:
    \[Status(0) = (\textit{initial}, \dots, \textit{initial})\]
    \[B_{value}(0) = (0, 0, \dots, 0)\]
    \[B_{carry}(0) = (0, 0, \dots, 0)\]
    \[B_{suc}(0) = (0, 0, \dots, 0)\]
    \textbf{Algorithm 1} operates as follows:
    \begin{itemize}
        \item In the case of robot $r_0$:
        \begin{itemize}
            \item If all other robots and the neighboring robot's $suc$ are in the $final$ state:\\
            Set $r_0.status$ to $initial$.
            
            \item Otherwise:\\
            $r_0$ prepares the next configuration. Specifically, it sets $(r_0.b, r_0.c, r_0.suc_b)$ to $(0, 1, r_1.b)$. Here, setting $r_0.c = 1$ prepares the increment of the binary value stored in $B_{value}(t)[1 \dots n-1]$. Finally, set $r_0.status$ to $final$.
        \end{itemize}
    
        \item In the case of robot $r_i \ (i \neq 0)$:
        \begin{itemize}
            \item If all other robots and the neighboring robot's $suc$ are in the $final$ state:\\
            Move to the "final position" calculated from $B_{value}(t)[1 \dots n-1]$. Since rigidity is not assumed, $r_i$ continues moving until it reaches its final position. Once the destination is reached, the pattern $C_p$ is formed, where $p$ is the integer represented by $B_{value}(t)[1 \dots n-1]$. Then, set $r_i.status$ to $initial$.
    
            \item If all preceding robots $r_0$ through $r_{i-1}$ are in the $initial$ state, and all succeeding robots $r_{i+1}$ through $r_{n-1}$ are in the $final$ state:\\
            A binary increment is performed sequentially, similar to a full-adder operation. For example, $r_1$ computes its least significant bit $r_1.b$ and the carry to the next digit $r_1.c$ using $r_0.c$ and $r_0.suc_b$ (i.e., $r_1.b$), and then sets $r_1.status$ to $final$. The status configuration becomes $Status(0) = (r_0.status=final, r_1.status=final, r_2.status=initial, \dots, r_{n-1}.status=initial)$, and as the computation progresses, each $r_i.status$ is sequentially updated.
        \end{itemize}
    \end{itemize}
    It is easy to verify that this algorithm works correctly in the \FSY\ scheduler.
\end{proof}

\begin{algorithm}
    \caption{Alg \textit{ETE} for robot $r$}
    \label{alg:ETE_FCOM_F}
    \textbf{Assumptions:} $\FC$, \FSY \\
    \textbf{Light:} $r.status \in \{initial, final\}$, initially $initial$ \\
    
    $Phase Look:$ Each robot $r_i$ observes the positions and lights of other robots.  Note that due to $\FC$, $r_i$ cannot observe its own light.\\
    
    $Phase Compute:$
    \begin{algorithmic}[1]
        \State $r.suc.status \gets suc(r).status$
        \If {$i = 0$}
            \If {$\forall \rho \neq r, \rho.status = initial$ and $\rho.suc.status = initial$}
                \State $r.status \gets final$
                \State$r.des \gets$ point obtained by scaling the distance from center of gravity by $d((r_{n-1}.b), (r_{n-2}.b), \dots, (r_1.b))$ and taking the floor.
            \Else
                \State $r_0.status \gets initial$
                \State $r_0.b \gets 0$
                \State $r_0.c \gets 1$
                \State $r_0.suc_b \gets r_1.b$
            \EndIf
        \Else
        \If {$\forall \rho \neq r, \rho.status = initial$ and $\rho.suc.status = initial$}
            \State $r.status \gets final$
            \State$r.des \gets$ point obtained by scaling the distance from center of gravity by $d((r_{n-1}.b), (r_{n-2}.b), \dots, (r_1.b))$ and taking the floor.
            \ElsIf {$\forall (0 \leq j \leq i-1) : (r_j.status = initial)$ and $\forall (i+1 \leq j \leq n-1) : (r_j.status = final)$}
                \State $r_i.b \gets pred(r_i).c \oplus pred(r_i).suc_b$
                \State $r_i.c \gets pred(r_i).c \cdot pred(r_i).suc_b$
                \State $r_i.suc_b \gets pred(r_i).b$
                \State $r_i.status \gets initial$
            \EndIf
        \EndIf
    \end{algorithmic}

    \vspace{1em}
    $Phase Move:$\\
    Move to $r_i.des$.
\end{algorithm}

\clearpage
\begin{figure}[H]
    \centering
    \includegraphics[width=12cm]{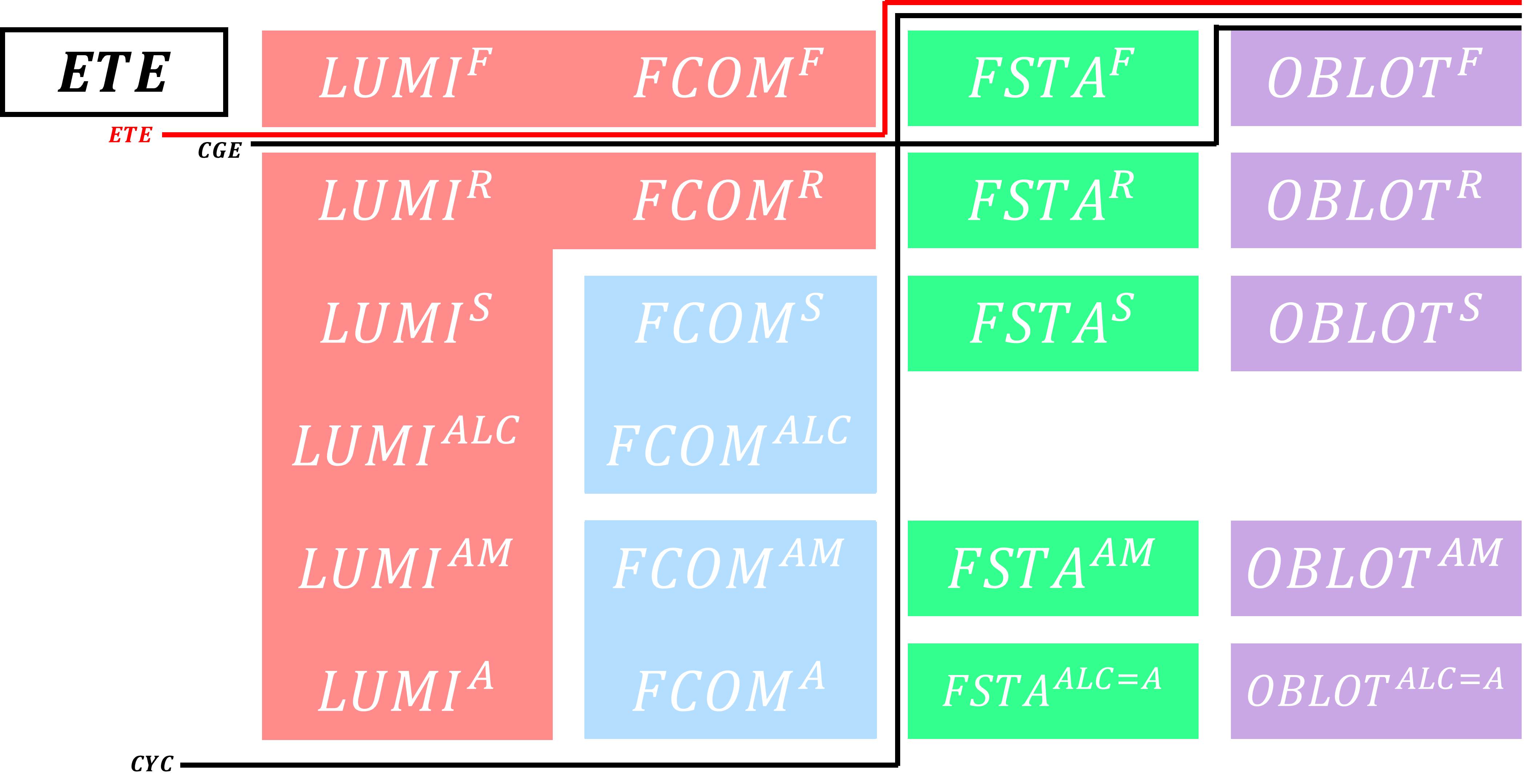}
    \caption{Boundary of ETE}
    \label{fig:ETE_boundary}
\end{figure}
\begin{theorem}
    \textit{ETE} is solved in $\LU^F$, but not in any of $\FS^F$ or $\LU^R$.
    \label{th:etet}
\end{theorem}
  
  \subsection{Scheduler-Capability Interplay (HET)}
  \label{arc:HET}
\begin{Definition}
    \textit{HET} (Hexagonal Edge Traversal)
        \\Two robots, $r$ and $q$, are initially placed on opposite vertices of a regular hexagon. 
        Each robot, when activated, must traverse exactly once to a vertex along an adjacent edge,
        reach a designated diagonal vertex, and then return to its initial position.
        
        Let the vertices of the hexagon be denoted by $V = \{v_0, v_1, v_2, v_3, v_4, v_5\}$.
        Suppose robot $r$ starts at vertex $v_i$ and robot $q$ starts at vertex $v_{(i+3) \bmod 6}$.
        
        There exist two strictly increasing infinite time sequences 
        \[
        0 \leq t_0 < t_1 < t_2 < \dots,\quad 0 \leq T_0 < T_1 < T_2 < \dots
        \]
        such that for all $j \geq 0$, the time indices satisfy the following cyclic condition:
        \[
        t_0 = t_{13j \bmod 13},\quad t_1 = t_{13j+1 \bmod 13},\quad \ldots, \quad
        T_0 = T_{13j \bmod 13},\quad T_1 = T_{13j+1 \bmod 13}, \quad \ldots
        \]
        
        The traversal operation must guarantee that each robot visits its designated diagonal vertex exactly once per cycle and then returns to its original vertex,
        without deviating from the predefined path.
        \clearpage
        \begin{align*}
            \textit{HET} \equiv 
            \forall r : \quad & (t_0 \leq \forall t \leq t_1 : r(t) = v_i) \land (t_2 \leq \forall t \leq t_3 : r(t) = v_{i-1 \bmod 6}) \\
            & \land (\exists t,t' : t_1 \leq t < t'\leq t_2 \rightarrow r(t),r(t') \in \overline{v_i v_{i-1 \bmod 6}}, r(t) \ne r(t'), \\
            & \text{dis}(r(t), v_i) < \text{dis}(r(t'), v_i), \text{dis}(r(t), v_{i-1 \bmod 6}) > \text{dis}(r(t'), v_{i-1 \bmod 6}))\\
            & \land (t_4 \leq \forall t \leq t_5 : r(t) = v_i) \\
            & \land (\exists t,t' : t_3 \leq t < t'\leq t_4 \rightarrow r(t),r(t') \in \overline{v_{i-1 \bmod 6} v_i}, r(t) \ne r(t')), \\
            & \text{dis}(r(t), v_i) > \text{dis}(r(t'), v_i), \text{dis}(r(t), v_{i-1 \bmod 6}) < \text{dis}(r(t'), v_{i-1 \bmod 6}))\\
            & \land (t_6 \leq \forall t \leq t_7 : r(t) = v_{i+1 \bmod 6}) \\
            & \land (\exists t,t' : t_5 \leq t < t'\leq t_6 \rightarrow r(t),r(t') \in \overline{v_i v_{i+1 \bmod 6}}, r(t) \ne r(t')), \\
            & \text{dis}(r(t), v_i) < \text{dis}(r(t'), v_i), \text{dis}(r(t), v_{i+1 \bmod 6}) > \text{dis}(r(t'), v_{i+1 \bmod 6}))\\
            & \land (t_8 \leq \forall t \leq t_9 : r(t) = v_i) \\
            & \land (\exists t,t' : t_7 \leq t < t'\leq t_8 \rightarrow r(t),r(t') \in \overline{v_{i+1 \bmod 6} v_i}, r(t) \ne r(t')),\\
            & \text{dis}(r(t), v_i) > \text{dis}(r(t'), v_i), \text{dis}(r(t), v_{i+1 \bmod 6}) < \text{dis}(r(t'), v_{i+1 \bmod 6}))\\
        \end{align*}
        \begin{align*}
            \textit{HET} \equiv 
            \forall q : \quad & (T_0 \leq \forall T \leq T_1 : q(T) = v_{i+3 \bmod 6}) \land (T_2 \leq \forall T \leq T_3 : q(T) = v_{i+2 \bmod 6}) \\
            & \land (\exists T,T' : T_1 \leq T < T'\leq T_2 \rightarrow q(T),q(T') \in \overline{v_{i+3 \bmod 6} v_{i+2 \bmod 6}}, q(T) \ne q(T'), \\
            & \text{dis}(q(T), v_{i+3 \bmod 6}) < \text{dis}(q(T'), v_{i+3 \bmod 6}), \text{dis}(q(T), v_{i+2 \bmod 6}) > \text{dis}(q(T'), v_{i+2 \bmod 6}))\\
            & \land (T_4 \leq \forall T \leq T_5 : q(T) = v_{i+3 \bmod 6}) \\
            & \land (\exists T,T' : T_3 \leq T < T'\leq T_4 \rightarrow q(T),q(T') \in \overline{v_{i+2 \bmod 6} v_{i+3 \bmod 6}}, q(T) \ne q(T')), \\
            & \text{dis}(q(T), v_{i+3 \bmod 6}) > \text{dis}(q(T'), v_{i+3 \bmod 6}), \text{dis}(q(T), v_{i+2 \bmod 6}) < \text{dis}(q(T'), v_{i+2 \bmod 6}))\\
            & \land (T_6 \leq \forall T \leq T_7 : q(T) = v_{i+4 \bmod 6}) \\
            & \land (\exists T,T' : T_5 \leq T < T'\leq T_6 \rightarrow q(T),q(T') \in \overline{v_{i+3 \bmod 6} v_{i+4 \bmod 6}}, q(T) \ne q(T')), \\
            & \text{dis}(q(T), v_{i+3 \bmod 6}) < \text{dis}(q(T'), v_{i+3 \bmod 6}), \text{dis}(q(T), v_{i+4 \bmod 6}) > \text{dis}(q(T'), v_{i+4 \bmod 6}))\\
            & \land (T_8 \leq \forall T \leq T_9 : q(T) = v_{i+3 \bmod 6}) \\
            & \land (\exists T,T' : T_7 \leq T < T'\leq T_8 \rightarrow q(T),q(T') \in \overline{v_{i+4 \bmod 6} v_{i+3 \bmod 6}}, q(T) \ne q(T')),\\
            & \text{dis}(q(T), v_{i+3 \bmod 6}) > \text{dis}(q(T'), v_{i+3 \bmod 6}), \text{dis}(q(T), v_{i+4 \bmod 6}) < \text{dis}(q(T'), v_{i+4 \bmod 6}))\\
        \end{align*}
\end{Definition}

\clearpage
\begin{figure}[H]
    \centering
    \includegraphics[width=100mm]{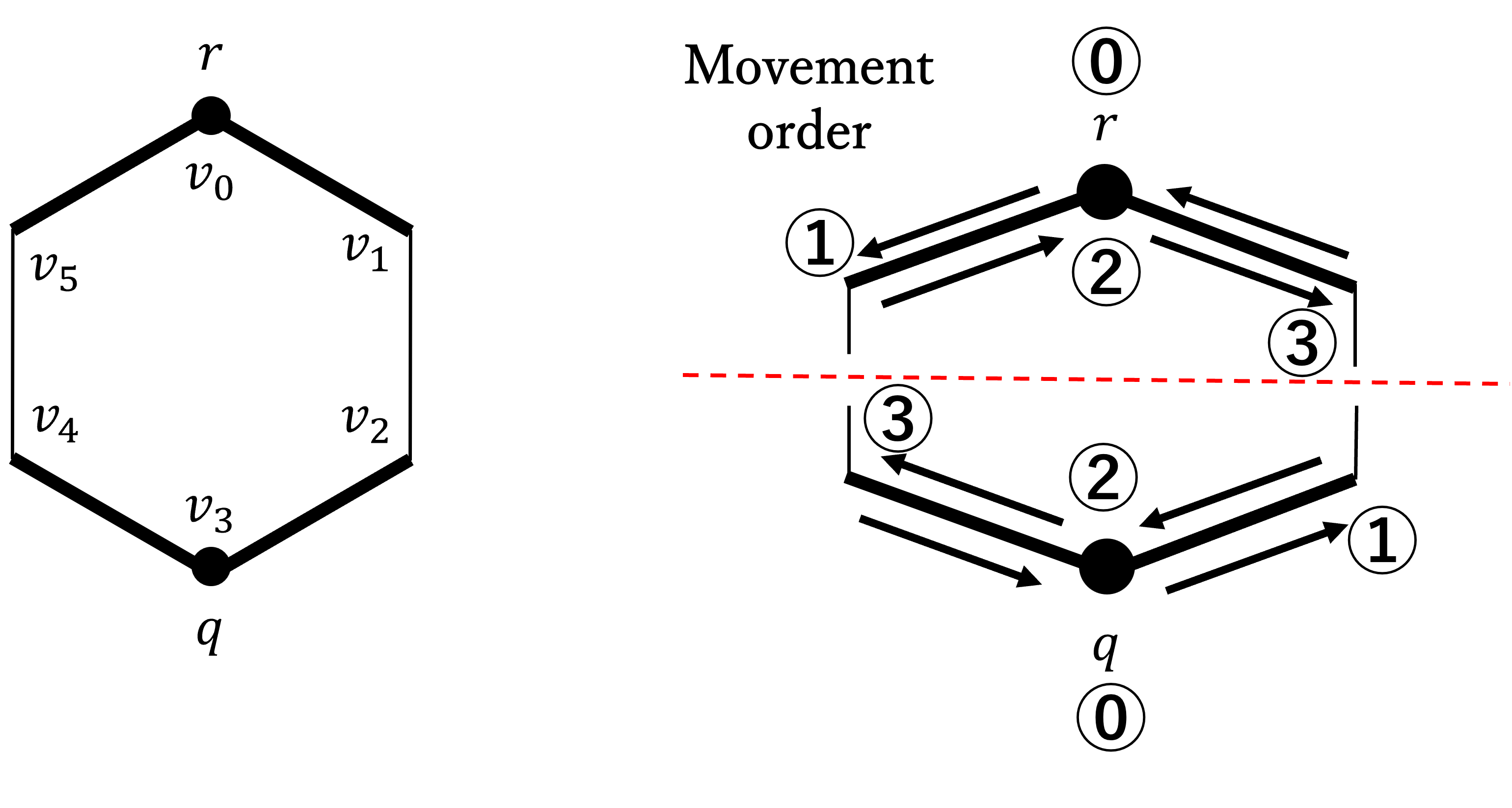}
    \caption{Movement of the \textit{HET} problem.}
\end{figure}

\begin{lemma}
    $\textit{HET} \in \FS^{F}$.
    \label{lem:5}
\end{lemma}
\begin{proof}
    The problem is solved by \textbf{Algorithm~\ref{alg:HET_FSTA_F}}. Initially, robots $r$ and $q$ are located at configuration $v_0$ and $v_3$, positioned on opposite vertices along a diagonal of a regular hexagon. At this time, both robots $r$ and $q$ have $r.status = 0$. When $r.status = 0$, it is updated to $r.status = 1$. Then, robot $r$ moves along the edge to the $v_5$, while robot $q$ moves along the edge to the $v_2$. Next, when $r.status = 1$, it is updated to $r.status = 2$, and both robots $r$ and $q$ return along the path they came from to the same positions. Then, when $r.status = 2$, it is updated to $r.status = 3$. In this state, robot $r$ moves along the edge to the $v_1$, and robot $q$ moves along the edge to the $v_4$. Finally, when $r.status = 3$, it is updated to $r.status = 0$, and both robots $r$ and $q$ return along the path they came from to the same positions. This cycle is repeated. Each robot can infer the other’s position from the color of its light and also maintains a local record of its own light. Therefore, it always knows which vertex it currently occupies. As a result, regardless of when robots $r$ and $q$ are activated, they can recognize that the other is on the opposite vertex. Even as they move along two adjacent edges from the initial configuration, they can precisely determine each other’s positions and continue to recognize the structure of the regular hexagon.
\end{proof}

\begin{algorithm}
    \caption{Algorithm for \textit{HET} Problem}
    \label{alg:HET_FSTA_F}
    \textbf{Assumptions:} $\FS$, \FSY \\
    \textbf{Lights:} $r.status \in \{0, 1, 2, 3\},$ initial value is $0$. \\
    
    $Phase Look:$ Robot $r$ observes its own position and light, as well as the positions of other robots. However, due to $\FS$, it cannot see the lights of other robots.\\

    $Phase Compute:$
    \begin{algorithmic}[1]
        \If {$r.status = 0$}
            \State $r.status \gets 1$
            \State $r.des \gets 1$
        \ElsIf {$r.status = 1$}
            \State $r.status \gets 2$
            \State $r.des \gets 2$
        \ElsIf {$r.status = 2$}
            \State $r.status \gets 3$
            \State $r.des \gets 3$
        \ElsIf {$r.status = 3$}
            \State $r.status \gets 0$
                \State $r.des \gets 0$
        \EndIf
    \end{algorithmic}
    
    \vspace{1em}
    $Phase Move:$\\
    Move to $r_i.des$.
\end{algorithm}
    
\begin{lemma}
    $\textit{HET} \notin \FS^{R}$.
    \label{lem:6}
\end{lemma}
\begin{proof}
    Assume that there exists an algorithm $\calA$ that solves the \textit{HET} problem under an arbitrary \RSY\ schedule $\calS$. We then construct another \RSY\ schedule $\calS'$ under which algorithm $\calA$ fails to solve \textit{HET}. In order for $\calA$ to solve \textit{HET} under schedule $\calS$, robots $r$ and $q$ must be activated in full synchrony up to some round $i$, and the problem must be solved based on \textbf{Lemma~\ref{lem:5}}. Next, we define schedule $\calS'$ such that it behaves identically to $\calS$ until round $i$, and after that, for a finite number of rounds, it activates only one of the robots, either $r$ or $q$. In such a case, when robots begin to be activated in an alternating (asynchronous) fashion, robot $r$ cannot determine its current position or how it should move. Suppose that during its first activation, robot $r$ moves from the initial position $v_0$ to $v_5$. In the next activation, robot $q$, located at $v_3$, observes robot $r$ at $v_5$ via a snapshot. However, assuming \textit{Variable Disorientation}, an adversary can force the perceived distance in each snapshot during the $\LK$ operation to always appear as $1$. As a result, the configuration in which $r$ is at $v_5$ and $q$ is at $v_3$ becomes indistinguishable from the initial configuration where the robots are placed on opposite vertices of the regular hexagon. Thus, in schedule $\calS'$, if either $r$ or $q$ is activated for a finite number of rounds, the configuration at time $t_L$ becomes indistinguishable from the initial one, making it difficult to correctly reconstruct the hexagonal structure. This leads to an invalid configuration. Since \textit{HET} cannot be solved by algorithm $\calA$ under the valid \RSY\ schedule $\calS'$, this contradicts the assumption that $\calA$ solves \textit{HET} under any \RSY\ schedule. Therefore, we conclude that no valid algorithm exists that solves \textit{HET} under $\FS$ with a \RSY\ scheduler.
\end{proof}

\begin{lemma}
    $\textit{HET} \notin \OB^{F}$.
    \label{lem:7}
\end{lemma}
\begin{proof}
    Consider two oblivious robots $r$ and $q$. The impossibility arises from the assumption of \textit{Variable Disorientation}, where, regardless of full synchrony, an adversary can force the observed distance in each snapshot taken during a $\LK$ operation to always appear as 1. In such a case, at time $t_L$, neither robot can determine which vertex they or the other occupies. As a result, they are unable to accurately perceive each other’s position or move along the edges while recognizing the structure of the regular hexagon.
\end{proof}

\begin{lemma}
    $\textit{HET} \notin \FC^{S}$.
    \label{lem:8}
\end{lemma}
\begin{proof}
    Assume that there exists an algorithm $\calA$ that solves the \textit{HET} problem under an arbitrary \SSY\ schedule $\calS$. Next, we construct another \SSY\ schedule $\calS'$ under which the algorithm $\calA$ fails to solve \textit{HET}. For $\calA$ to solve \textit{HET} under schedule $\calS$, robots $r$ and $q$ must be activated alternately up to some round $i$, accurately perceive each other’s positions, and move along the edges while maintaining correct recognition of the regular hexagonal structure. Now, consider a scheduler $\calS'$ that activates only robot $r$ for a finite number of rounds. In its first activation, robot $r$ moves from its initial position $v_0$ to $v_5$. On the next activation, robot $r$, now located at $v_5$, takes a snapshot and observes robot $q$ located at $v_3$. However, under the assumption of \textit{Variable Disorientation}, the adversary can force the observed distance in the snapshot taken during each $\LK$ operation to appear as $1$. As a result, the configuration in which $r$ is at $v_5$ and $q$ at $v_3$ becomes indistinguishable from the initial configuration where they are placed on opposite vertices. Thus, if either $r$ or $q$ is kept active for $x$ rounds under schedule $\calS'$, the robots cannot distinguish the configuration at time $t_L$ from the initial one, making it difficult to correctly reconstruct the hexagonal structure. This results in an invalid computation. Although $\calS'$ remains a fair scheduler and $x$ is finite (ensuring that other robots are eventually reactivated), algorithm $\calA$ fails to solve \textit{HET} under this valid \SSY\ schedule $\calS'$. This contradicts the assumption that $\calA$ solves \textit{HET} under any \SSY\ schedule. Therefore, we conclude that there exists no valid algorithm that solves \textit{HET} under $\FC^{S}$ and a \SSY\ scheduler.
\end{proof}

\begin{lemma}
    $\textit{HET} \in \LU^R$.
    \label{lem:9}
\end{lemma}
\begin{proof}
    The problem is solved by \textbf{Algorithm~\ref{alg:HET_LUMI_R}}. In the initial state, robots $r$ and $q$ are located at the initial configuration $v_0$ and $v_3$, positioned on opposite vertices (i.e., along a diagonal) of a regular hexagon. At this time, both robots have their light set to $0$. If a robot's light is $0$ and it observes that the other robot's light is $0$ or $1$, robot $r$ moves along the edge to the $v_5$, and robot $q$ moves along the edge to the $v_2$. Next, if a robot’s light is $1$ and it observes that the other robot’s light is $1$ or $2$, both robots return along the paths they came and move back to the same position as in the initial configuration. Then, if a robot’s light is $2$ and it observes that the other robot’s light is $2$ or $3$, robot $r$ moves along the edge to the $v_1$, and robot $q$ moves along the edge to the $v_4$. Finally, if a robot’s light is $3$ and it observes that the other robot’s light is $3$ or $0$, both robots return along the paths they came and move back to the same position as in the initial configuration. This cycle is then repeated. Each robot can infer the other’s position by reading its light, and since it also records the light locally, it can always determine its own vertex and the position of the other robot. Therefore, even if robots $r$ and $q$ are activated asynchronously, they can always recognize that the other is located along the diagonal. As a result, even when moving along the two adjacent edges from the initial configuration, they can accurately determine each other’s position and consistently recognize the regular hexagonal structure.This is also the case where chirality is undefined.
\end{proof}

\begin{algorithm}
    \caption{Algorithm for \textit{HET} Problem}
    \label{alg:HET_LUMI_R}
    \textbf{Assumptions:} $\LU$, \RSY \\
    \textbf{Lights:} $r.status \in \{0, 1, 2, 3\},$ initial value is $0$.\\
    
    $Phase Look:$ Robot $r$ observes its own position and light, as well as the positions of other robots. Note that due to $\LU$, it can see the lights of other robots.\\
    
    $Phase Compute:$
    \begin{algorithmic}[1]
        \If {$r.status = 0$}
            \If {$\forall \rho \neq r, \rho.status = 0$ \textbf{or} $\forall \rho \neq r, \rho.status = 1$}
                \State $r.status \gets 1$
                \State $r.des \gets 1$
            \EndIf
        \ElsIf {$r.status = 1$}
            \If {$\forall \rho \neq r, \rho.status = 1$ \textbf{or} $\forall \rho \neq r, \rho.status = 2$}
                \State $r.status \gets 2$
                \State $r.des \gets 2$
            \EndIf
        \ElsIf {$r.status = 2$}
            \If {$\forall \rho \neq r, \rho.status = 2$ \textbf{or} $\forall \rho \neq r, \rho.status = 3$}
                \State $r.status \gets 3$
                \State $r.des \gets 3$
            \EndIf
        \ElsIf {$r.status = 3$}
            \If {$\forall \rho \neq r, \rho.status = 3$ \textbf{or} $\forall \rho \neq r, \rho.status = 0$}
                \State $r.status \gets 0$
                \State $r.des \gets 0$
            \EndIf
        \EndIf
    \end{algorithmic}
    
    \vspace{1em}
    $Phase Move:$\\
    Move to $r_i.des$.
\end{algorithm}

\clearpage
\begin{figure}[H]
    \centering
    \includegraphics[width=12cm]{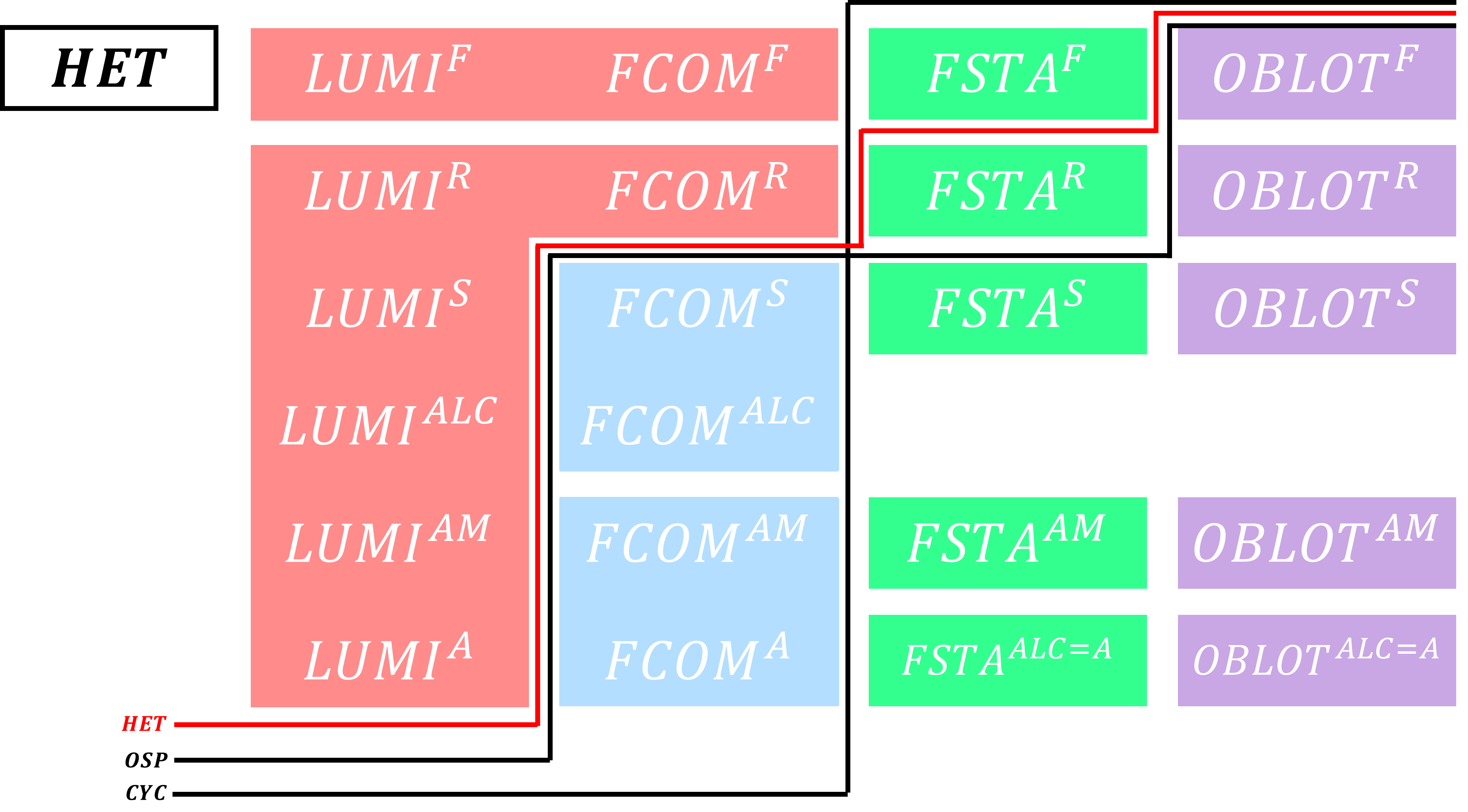}
    \caption{Boundary of HET}
    \label{fig:HET_boundary}
\end{figure}

\begin{theorem}
    \textit{HET} is solved in both of $\FS^{F}$ and $\LU^{R}$, but not in any of $\FS^{R}$, $\OB^{F}$ or $\FC^{S}$.
    \label{th:hett}
\end{theorem}
  
  \subsection{Synchrony Can Substitute Capability Only under FSYNCH (TAR(d)*)}
  \label{arc:TAR}
\begin{Definition}
    \textit{TAR($d$)*} (Infinite Triangle Rotation with parameter $d$)
    \\Let $a, b, c$ be three robots forming a triangle $ABC$, let $\calC$ be the circumscribed circle, and let $d$ be a value known to the three robots. The \textit{TAR(d)*} problem requires the robots to move so to form a new triangle $A'B'C'$ with circumscribed circle $\calC$, and where $dis(A, A') = dis(B, B') = dis(C, C') = d$. We then construct a new triangle $A''B''C''$ such that its circumcircle is denoted by $\calC$, and it satisfies the following condition : $dis(A', A'') = dis(B', B'') = dis(C', C'') = d$. This process is required to be repeated indefinitely.
\end{Definition}

\begin{figure}[H]
    \centering
    \includegraphics[width=12cm]{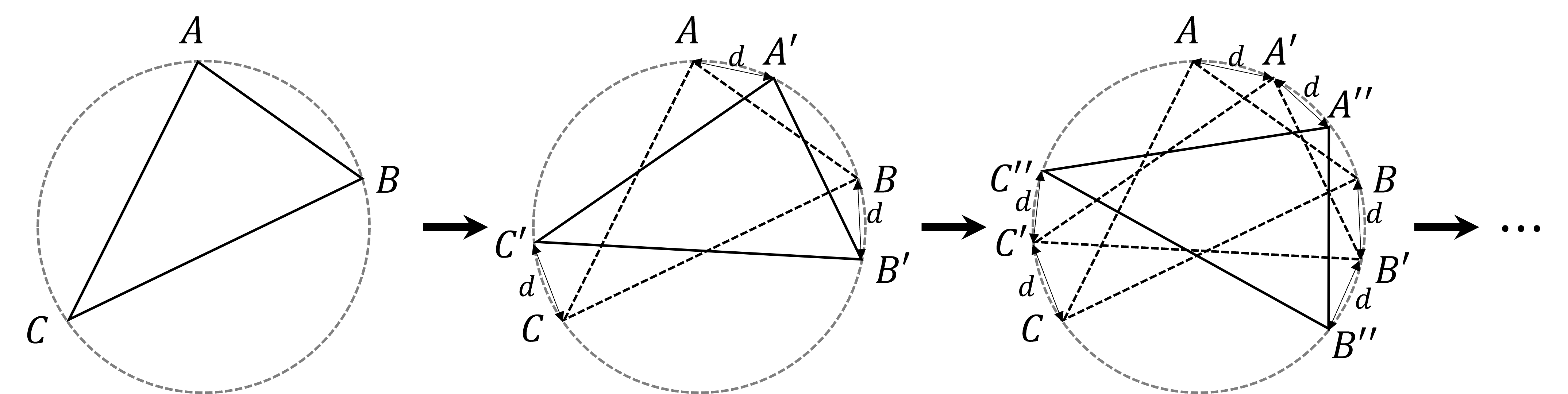}
    \caption{\textit{Movement of the TAR($d$)*} problem.}
    \label{fig:TAR}
\end{figure}
\begin{lemma}
    $\textit{TAR($d$)*} \in \OB^{F}$
    \label{lem:10}
\end{lemma}
\begin{proof}
    Even if a robot cannot see its own light or that of other robots, since they operate under full synchrony, all robots move simultaneously by a distance of $d$. Therefore, it is possible to form a new triangle whose circumscribed circle is $\calC$.
\end{proof}

\begin{lemma}
    $\textit{TAR($d$)*} \notin \FS^{R}$
    \label{lem:11}
\end{lemma}
\begin{proof}
    Assume that there exists an algorithm $\calA$ that solves \textit{TAR(d)*} under any \RSY\ schedule $\calS$. Next, we construct another \RSY\ schedule $\calS'$ under which algorithm $\calA$ fails to solve \textit{TAR(d)*}. In order for algorithm $\calA$ to solve \textit{TAR($d$)*} under $\calS$, the three robots must be activated in perfect synchrony up to some round $i$, during which the problem is solved. We define $\calS'$ to be identical to $\calS$ up to round $i$, but for a finite number of rounds thereafter, $\calS'$ activates only one of the robots $r$ or $q$. Although $r$ and $q$ can recognize that they have moved and changed their light colors, they cannot determine whether the other robot has also changed its color or state, since it may not have been activated. Moreover, since the third robot is not activated during this time, it is impossible to form a new triangle. If there were a guarantee that the three robots are activated in turn, such as under a round-robin schedule, it would be possible to eventually form a new triangle. However, since no such guarantee exists in this configuration, it becomes difficult to construct a new triangle based solely on snapshots. Therefore, under schedule $\calS'$, if either robot $r$ or $q$ is activated for a finite number of rounds, it is difficult to form a new triangle using only snapshots, resulting in an invalid configuration. In other words, \textit{TAR(d)*} cannot be solved by algorithm $\calA$ under the valid \RSY\ schedule $\calS'$. This contradicts the assumption that algorithm $\calA$ solves \textit{TAR(d)*} under any \RSY\ schedule. Thus, we conclude that no valid algorithm exists for solving \textit{TAR(d)*} under \RSY\ schedules in the $\FS$ model.
\end{proof}

\begin{lemma}
    $\textit{TAR($d$)*} \notin \FC^{S}$.
    \label{lem:12}
\end{lemma}
\begin{proof}
    We prove by contradiction. Assume that there exists a correct solution protocol $\calA$ in $\FC^{S}$. Let the initial configuration $C_0$ be such that robots $a$, $b$, and $c$ form an acute triangle $ABC$, where $AB \neq d$, $BC \neq d$, and $CA \neq d$, and all lights are initially set to $off$. Consider an execution $\varepsilon$ in which all robots are activated simultaneously in every round, starting from $C_0$ and producing a sequence of configurations $C_0, C_1, \dots$. Suppose that there exists a round in which at least one robot is not activated simultaneously with the others. Let $r$ be a robot that operates asynchronously for the first time in round $k$, after having observed the configuration $C_{k-1}$. Now consider another execution $\varepsilon'$, which proceeds identically to $\varepsilon$ for the first $k-1$ rounds, but in round $k$, only robot $r$ is activated. In this case, $r$ changes its light color and moves to a new position. Assume that the resulting configuration forms a scaled triangle. If the previous configuration was also an acute triangle, then the externally observable situation becomes indistinguishable from the previous one, and thus robot $r$ may repeat the same move again. That is, the robot continues to move without gaining any new information from its previous position, and consequently, the solution to \textit{TAR(d)*} cannot be obtained. Therefore, the \textit{TAR($d$)*} problem cannot be solved by robots operating under the $\FC^{S}$ model.
\end{proof}

\begin{lemma}
    $\textit{TAR(d)*} \in \LU^{R}$.
    \label{lem:13}
\end{lemma}
\begin{proof}
     The problem is solved by \textbf{Algorithm~\ref{alg:TARd*_LUMI_R}}. Initially, each robot has $r.status = 0$. If a robot finds that its own $r.status$ is the same as that of all other robots, or that the $r.status$ of every other robot is equal to its own status plus one, then it updates its status to $r.status + 1 \bmod 3$, moves a predetermined distance $d$, and repeats this cycle. Each robot can determine the status of the others from their lights and also stores the light states locally, which allows it to always know both its own and others’ statuses. Therefore, regardless of when they are activated, the robots can continue to form new triangles whose circumscribed circle is $\calC$ indefinitely.
\end{proof}

\begin{algorithm}
    \caption{Alg \textit{TAR($d$)*} for robot $r$}
    \label{alg:TARd*_LUMI_R}
    \textbf{Assumptions:} $\LU$,\RSY \\
    \textbf{Light:} $r.status \in \{0, 1, 2, 3\},$ initial value is $0$.\\
    
    $Phase Look:$ The robot can see its own lights and the lights of other robots.\\
    
    $Phase Compute:$
    \begin{algorithmic}[1]
        \If {$(\forall \rho \neq r,\, \rho.status = r.status)\ \lor\ (\exists \rho \neq r,\, \rho.status = r.status + 1)$}
            \State $r.status \gets (r.status + 1) \bmod 3$
            \State $r.des \gets d$
        \EndIf
    \end{algorithmic}

    \vspace{1em}
    $Phase Move:$\\
    Move to $r_i.des$.
\end{algorithm}
\clearpage
\begin{figure}[H]
    \centering
    \includegraphics[width=12cm]{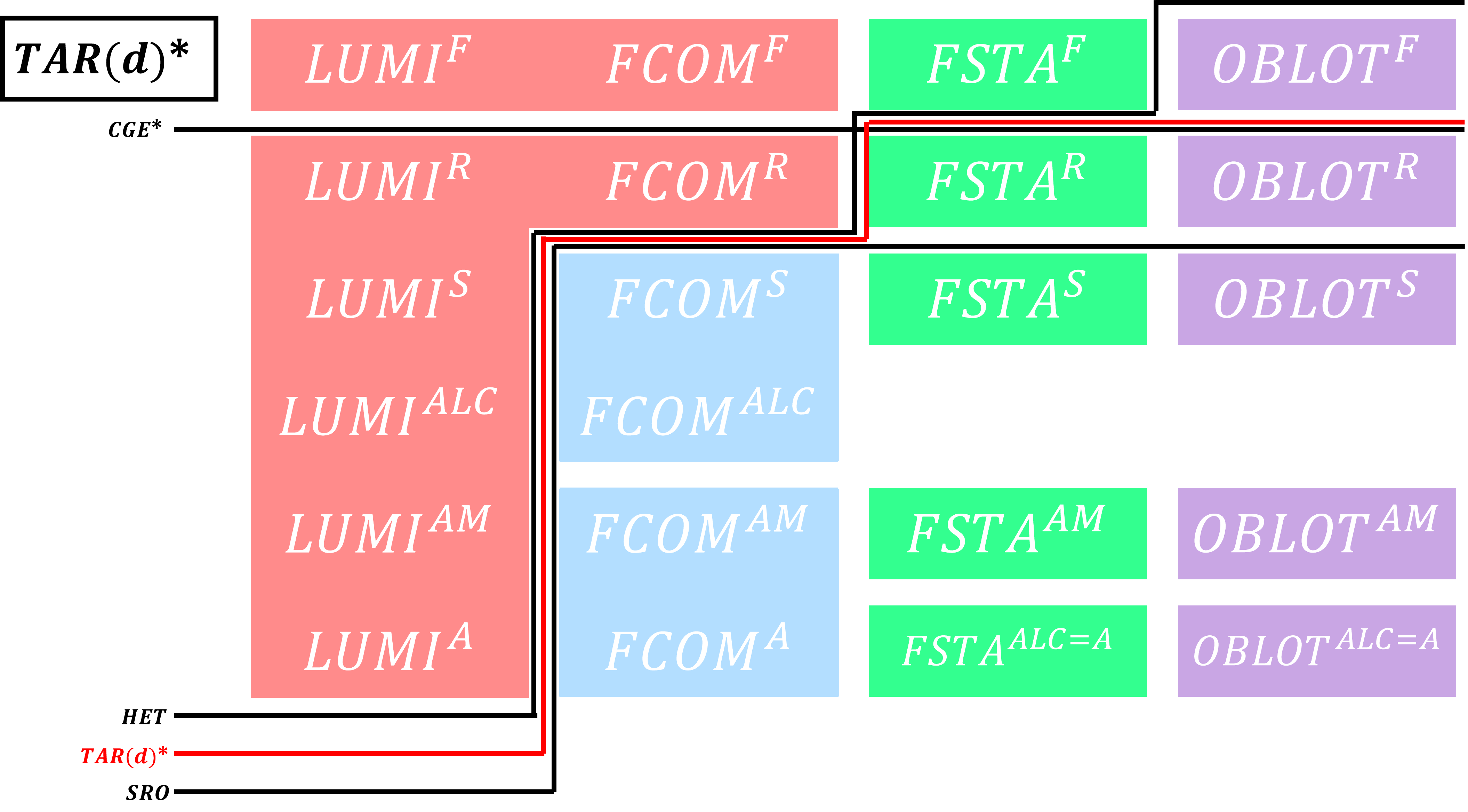}
    \caption{Boundary of \textit{TAR($d$)*}}
    \label{fig:TAR_boundary}
\end{figure}

\begin{theorem}
    \textit{TAR($d$)*} is solved in both of $\OB^{F}$ and $\LU^{R}$, but not in any of $\FS^{R}$ or $\FC^{S}$.
    \label{th:tart}
\end{theorem}
  
\section{Separation under Asynchronous Schedulers}\label{sep-asynch-scheduler}
This section addresses separation results that involve asynchronous and partially synchronous schedulers.

Specifically, Section~\ref{sec:5-1} and Section~\ref{sec:5-2} present problems that can be solved by $\LU^A$, but cannot be solved by $\OB$ under full synchrony ($\OB^{F}$).
These problems illustrate how $\FC$ and $\FS$ models depend critically on their scheduler assumptions 
and reveal non-trivial trade-offs between internal memory, external lights, and the degree of synchrony.

In contrast, Section~\ref{sec:5-3} focuses on a unique problem that is solvable by any model other than $\OB$ under the weakest scheduler (\ASY),
but remains unsolvable by $\OB$ even under \RSY.
It becomes solvable for $\OB$ only when perfect synchrony (\FSY) is assumed.

These results together clarify how the separation landscape extends beyond synchronous settings
and how asynchronous schedulers highlight the minimal capabilities required for solving coordination problems.
  \subsection{Fully Asynchronous Mutual Light Trade-Off Problems (LP-MLCv, VEC, ZCC)}\label{sec:5-1}

The diagram summarizes three representative problems—\textit{LP–MLCv}, \textit{VEC}, and \textit{ZCC}—that each demonstrate how specific trade-offs between internal memory, external lights, and scheduler strength determine solvability.

Each arrow or line in the figure shows which pairs of configurations are strictly separated by these problems,
clarifying how $\LU^{A}$ can solve them, while $\OB^{F}$ cannot,
and revealing which combinations of $\FC$ and $\FS$ depend critically on the scheduling assumptions.

  \label{arc:LP-MLCv}
\begin{Definition}
    \textit{LP--MLCv} (Leave the Place--MLCv)
    \\Two robots $r$ and $q$ are placed at arbitrary distinct points. Each robot moves away from the other along the line segment connecting $r$ and $q$ (referred to as \textit{LP}). Subsequently, $r$ and $q$ must solve the collision-free linear convergence problem (referred to as \textit{MLCv}) without ever increasing the distance between them. When performing \textit{MLCv} after completing \textit{LP}, the robots must have moved away from their initial positions. Assuming the robot located in the negative direction is $r$ and the robot in the positive direction is $q$, the \textit{LP--MLCv} problem is defined by the following logical formula.
    \begin{align*}
        \textit{LP--MLCv} \equiv 
        \big[ &\exists T \geq 0 : \big[\{ \forall t \geq 0 : r(0), q(0) \in \overline{r(t) q(t)} \}\\
            &\land \{r(T) \neq r(0), \, q(T) \neq q(0) \} \\
            &\land \{\forall t \leq T : r(t) - r(0) \leq 0, \, q(t) - q(0) \geq 0 \}\big]\\
            &\land \big[ \{\exists l \in \mathbb{R}^2, \forall \epsilon \geq 0, \exists T' \geq T, \forall t \geq T' : 
                |r(t) - l| + |q(t) - l| \leq \epsilon \}\\
            &\land \{\forall t \geq T : r(t), q(t) \in \overline{r(T) q(T)} \}\\
            &\land \{\forall t \geq T : 
                \text{dis}(r(T), r(t)) \leq \text{dis}(r(T), q(t)), \, 
                \text{dis}(q(T), q(t)) \leq \text{dis}(q(T), r(t)) \}\\
            &\land \{\forall t \geq T : q(t) - r(t) \leq q(T) - r(T)\} \big] \big]
    \end{align*}
\end{Definition}
    
\begin{figure}[H]
    \centering
    \includegraphics[width=7.5cm]{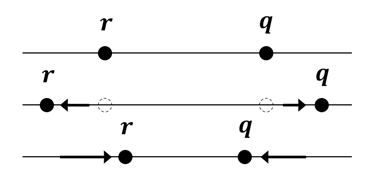}
    \caption{Movement of the \textit{LP--MLCv} problem.}
    \label{fig:lp}
\end{figure}

\begin{lemma}
    $\textit{LP--MLCv} \notin \OB^{F}$.
    \label{lem:14}
\end{lemma}
\begin{proof}
    In the $\OB$ model, robots are \textit{oblivious} and thus base their decisions solely on information from the current snapshot. This constraint makes it impossible for an activated robot to distinguish whether it or the other robot has already completed the \textit{LP} action. Consequently, the robot cannot correctly select whether to perform \textit{LP} or \textit{MLCv} as its next action.
\end{proof}

\begin{lemma}
    $\textit{LP--MLCv} \notin \FS^S$.
    \label{lem:15}
\end{lemma}
\begin{proof}
    Suppose there exists an algorithm $\calA$ that solves \textit{LP--MLCv} under the $\FS^S$ model.  
    Consider the case where, in the \textit{LP} phase, a robot $r$ that has finished moving away from a robot $q$ is activated in round $i$.  
    In round $i$, the possible actions for $r$ are either to approach $q$ in order to perform \textit{MLCv}, or to remain stationary until $q$ completes the LP phase.  
    Assume that round $j$ is the first round in which $r$ performs \textit{MLCv}. By this round, both $r$ and $q$ must have completed the LP phase. Now, consider a round $k$ before $q$ has performed its LP phase. Due to variable disorientation, the adversary can activate only $r$ from round $k$ until round $j$ to bring the system to an identical state from $r$'s perspective. In this scenario, $r$ would initiate \textit{MLCv} without $q$ having performed the LP phase, which means algorithm $\calA$ is incorrect.
    Here, the robot $r$ under the $\FS$ model cannot determine with only a finite amount of memory whether $q$ is still located at its initial position or has moved away from it.  
    Therefore, if algorithm $\calA$ dictates that $r$ chooses to approach $q$, there is a possibility that $r$ executes this action even if $q$ has remained at its initial position.  
    On the other hand, if algorithm $\calA$ instructs $r$ to stay still, $r$ will never observe that $q$ has moved away from the initial position, and thus $r$ will remain stationary forever from round $i$ onward, regardless of how many times it is activated.  
    These behaviors contradict the problem definition, and thus contradict the correctness of $\calA$.
\end{proof}

\begin{lemma}
    $\textit{LP--MLCv} \in \FS^R$.
    \label{lem:16}
\end{lemma}
\begin{proof}
    The problem is solved by \textbf{Algorithm \ref{alg:LPMLCv_FSTA_R}}.  
    By the definition of the \RSY\ scheduler, when either robot is activated for the second time, the other robot has been activated at least once.  
    In other words, under the $\FS^R$ model, an algorithm solving the \textit{LP--MLCv} problem only needs to recognize whether the activated robot has previously performed the \textit{LP} action. This is also the case where chirality and rigidity are undefined.
\end{proof}
\begin{algorithm}
    \caption{Alg \textit{LP--MLCv} for robot $r$}
    \label{alg:LPMLCv_FSTA_R}
    \textbf{Assumptions:} $\FS$, \RSY \\
    \textbf{Light:} $r.\text{color} \in \{A, B\}$, initial value is $A$\\
    
    $Phase Look:$ $other.\text{pos} \gets$ position of the other robot in $r$'s local coordinate system.\\
    
    $Phase Compute:$
    \begin{algorithmic}[1]
        \State \textbf{case} $r.\text{color}$ \textbf{of}
        \State \hspace{1em} $A$:
        \State \hspace{2em} $r.\text{color} \gets B$
        \State \hspace{2em} \parbox[t]{0.8\linewidth}{$r.\text{des} \gets$ the point opposite to $other.\text{pos}$ at distance $|\text{other.pos}|/2$ from $r$}
        \State \hspace{1em} $B$:
        \State \hspace{2em} \parbox[t]{0.8\linewidth}{$r.\text{des} \gets$ the point toward $other.\text{pos}$ at distance $|\text{other.pos}|/2$ from $r$}
    \end{algorithmic}

    \vspace{1em}
    $Phase Move:$\\
    Move to $r_i.des$.
\end{algorithm}

\begin{lemma}
    $\textit{LP--MLCv} \in \FC^S$.
    \label{lem:17}
\end{lemma}
\begin{proof}
    The problem is solved by \textbf{Algorithm \ref{alg:LPMLCv_FCOM_S}}.  
    By using three-colored lights, a robot can determine whether the other has moved away from its initial position during the $LP$ phase. Specifically, each robot makes decisions as follows:
    
    \begin{enumerate}[label=(\arabic*)]
        \item If the other robot's light is $A$, it means the other robot has never been activated, and this robot itself may have been activated once before or never activated at all.  
        Therefore, this robot changes its own light to $q$ and moves away from the other robot.
        
        \item If the other robot's light is $B$, it means the other robot has been activated at least once before, and this robot itself may have been activated once before or never activated.  
        Therefore, this robot changes its own light to $C$ and moves away from the other robot.
        
        \item If the other robot's light is $C$, it means both robots have been activated at least once before.  
        Therefore, this robot changes its own light to $C$ and performs the \textit{MLCv} movement.
    \end{enumerate}
    
    In \textit{MLCv}, the process (3) is simply repeated, and prior research has shown that $\textit{MLCv} \in \FC^S$. This is also the case where chirality and rigidity are undefined.
\end{proof}
\begin{algorithm}
    \caption{Alg \textit{LP--MLCv} for robot $r$}
    \label{alg:LPMLCv_FCOM_S}
    \textbf{Assumptions:} $\FC$, \SSY \\
    \textbf{Light:} $r.\text{color}, \text{other.color} \in \{A, B, C\}$, initial value is $A$\\
    
    $Phase Look:$ $other.\text{pos} \gets$ position of the other robot in $r$'s local coordinate system.\\
    
    $Phase Compute:$
    \begin{algorithmic}[1]
        \State \textbf{case} $\text{other.color}$ \textbf{of}
        \State \hspace{1em} $A$:
        \State \hspace{2em} $r.\text{color} \gets B$
        \State \hspace{2em} \parbox[t]{0.8\linewidth}{$r.\text{des} \gets$ the point in the opposite direction of $other.\text{pos}$ at a distance of $|\text{other.pos}|/2$ from $r$}
        \State \hspace{1em} $B$:
        \State \hspace{2em} $r.\text{color} \gets C$
        \State \hspace{2em} \parbox[t]{0.8\linewidth}{$r.\text{des} \gets$ the point in the opposite direction of $other.\text{pos}$ at a distance of $|\text{other.pos}|/2$ from $r$}
        \State \hspace{1em} $C$:
        \State \hspace{2em} $r.\text{color} \gets C$
        \State \hspace{2em} \parbox[t]{0.8\linewidth}{$r.\text{des} \gets$ the point toward $other.\text{pos}$ at a distance of $|\text{other.pos}|/2$ from $r$}
    \end{algorithmic}

    \vspace{1em}
    $Phase Move:$\\
    Move to $r_i.des$.
\end{algorithm}

\begin{lemma}
    $\textit{LP--MLCv} \notin \FC^A$.
    \label{lem:18}
\end{lemma}
\begin{proof}
    Assume, for contradiction, that there exists an algorithm $\calA$ that solves \textit{LP--MLCv} under the $\FC^A$ model.  
    Consider two robots $r$ and $q$ which share the same unit of distance, initially facing each other along the positive direction of the $X$-axis, and assume their local coordinate systems do not change during the execution of $\calA$.  
    We make three observations:
    
    \begin{enumerate}
        \item By the predicate defining \textit{LP--MLCv}, if a robot moves, it must move toward the other robot.  
        Moreover, in this specific setting, the robots must remain on the $X$-axis.
        
        \item Every time a robot is activated and executes $\calA$, it must move.  
        Conversely, if $\calA$ specifies that a robot activated at a distance $d$ from the other robot must not move, then assuming $\calA$ executes in full synchrony and both robots initially are at distance $d$, neither robot will move, and thus convergence will never occur.
        
        \item After robot $r$ observes $q$ at distance $d$, when $r$ moves toward $q$ on the $X$-axis, the computed movement length $f(d)$ is the same as the length $f(d)$ computed when $q$ observes $r$ at distance $d$.
    \end{enumerate}
    
    Now consider the following execution $\epsilon$ under $\calA$: initially, both robots are activated simultaneously and are at distance $d$ from each other.  
    Robot $r$ completes its computation and executes the move instantaneously (note that $r$ operates under $\calA$), while robot $q$ is still performing the $\LK$ and $\CP$ operations and continues to be activated, executing $\calA$ repeatedly.
    
    Each movement of $r$ obviously reduces the distance between the two robots. More precisely, by observation (3), after $k \geq 0$ moves, the distance decreases from $d$ to $d_k$, where
    
    \[
    d_0 = d, \quad d_{k > 0} = d_{k-1} - f(d_{k-1}) = d - \sum_{0 \leq i < k} f(d_i).
    \]
    
    Here, by the problem definition, $f(d_i) - f(d_{i+1})$ is always positive.
\end{proof}

\begin{lemma}
    After a finite number of movements of $r$, the distance between the two robots becomes smaller than $f(d)$.
    \label{lem:19}
\end{lemma}
\begin{proof}
    Assume for contradiction that $r$ never approaches $q$ closer than $f(d)$.  
    That is, for all $k > 0$, we have $d_k > f(d)$.  
    Consider the execution $\hat{\epsilon}$ of $\calA$ under a \RR\ synchronous scheduler, where initially the robots are at distance $d$ and are activated alternately one by one in each round.  
    Since $\calA$ is assumed to be correct under $\FC^A$, it should also be correct under \RR.  
    This implies that starting from the initial distance $d$, the two robots eventually get closer than any fixed positive distance $d'$.  
    Let $m(d')$ denote the number of rounds until this happens.

    After $m(f(d))$ rounds, the distance between the robots must become smaller than $f(d)$.  
    Furthermore, after round $i$, the distance $d_i$ decreases by $f(d_i)$. Summarizing,

    \[
    d_{m(f(d))} = d - \sum_{0 \leq i < m(f(d))} f(d_i) < f(d),
    \]

    which contradicts the assumption that $d_k > f(d)$ for all $k > 0$.  
    \qed

    Now consider the execution $\epsilon$ when the distance becomes smaller than $f(d)$.  
    Suppose robot $q$ completes its computation at that time and moves a distance $f(d)$ toward $r$.  
    This move causes a collision, contradicting the correctness of $\calA$.
\end{proof}

\begin{figure}[H]
    \centering
    \includegraphics[width=12cm]{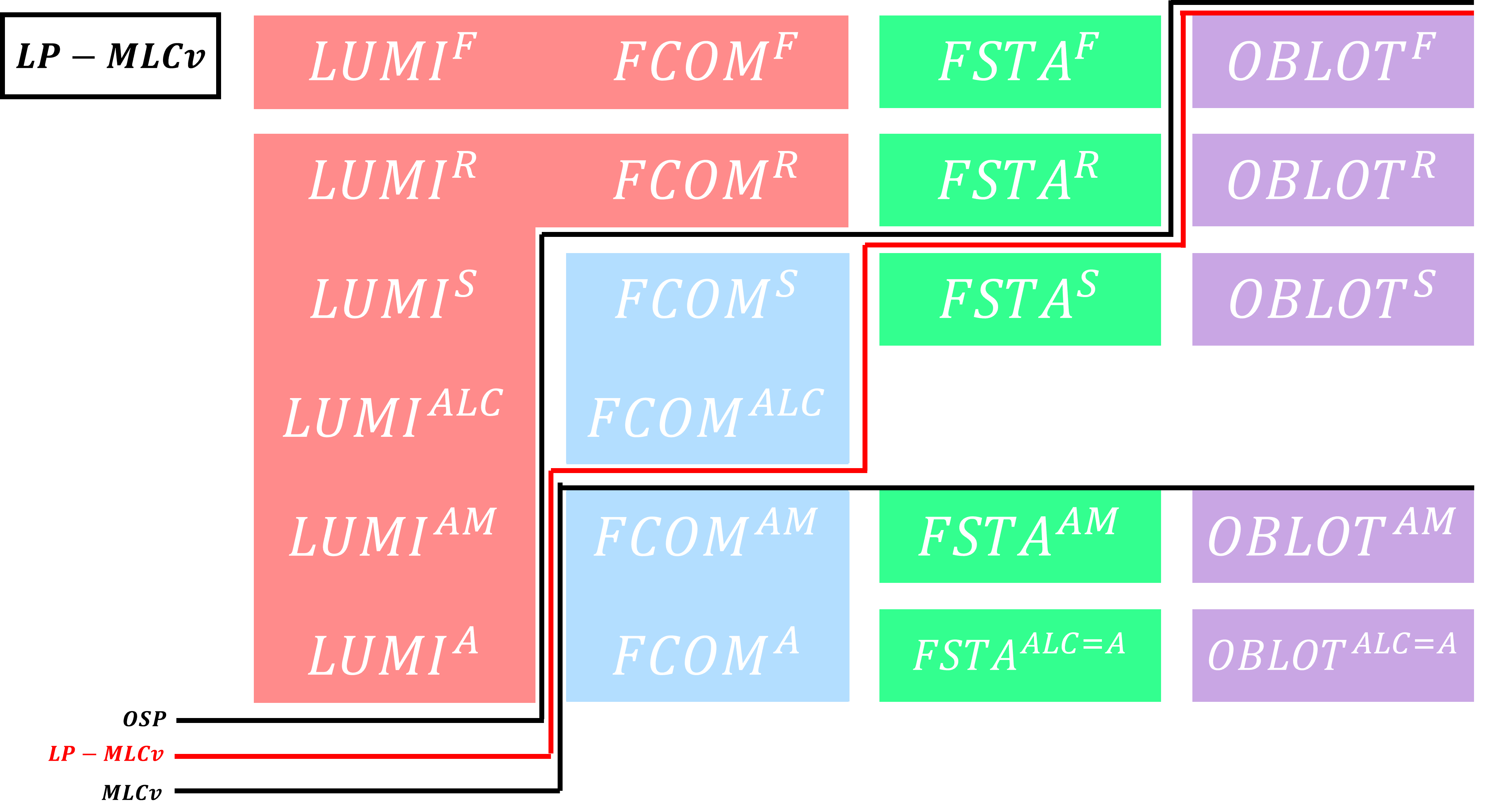}
    \caption{Boundary of LP--MLCv}
    \label{fig:LP-MLCv_boundary}
\end{figure}

\begin{theorem}
    $\textit{LP--MLCv}$ is solved in $\FS^R$ and $\FC^S$, but not in any of $\OB^F$, $\FS^S$ or $\FC^S$.
    \label{th:lpmlcvt}
\end{theorem}
  \label{arc:VEC}
\begin{Definition}
    \textit{VEC} (Vertical Configuration)
    \\Two robots, $r$ and $q$, are placed at arbitrary, distinct points (initial configuration $\calC_0$). These two robots uniquely identify a square $\mathcal{Q}_0$ whose diagonal is the segment connecting the two robots. Let $r_0$ and $q_0$ denote the initial positions of the robots, and let $d_0$ be the segment between the robots, with $\mathrm{length}(d_0)$ denoting its length. Let $r_i$ and $q_i$ be the positions of $r$ and $q$ in configuration $\calC_i$ ($i \geq 0$). The robots move from configuration $\calC_i$ to $\calC{i+1}$ so as to satisfy one of the conditions 1, 2, or 3, and in all cases, condition 4 must be satisfied. However, condition 5 must hold.
    \begin{enumerate}[label=\ \ Condition \arabic*, leftmargin=*, labelsep=1em]
        \item $d_{i+1}$ is equal to $d_i$, and $\mathrm{length}(d_{i+1}) = \mathrm{length}(d_i)$
        
        \item $d_{i+1}$ is obtained by rotating $d_i$ $90^\circ$ clockwise, and hence $\mathrm{length}(d_{i+1}) = \mathrm{length}(d_i)$
        
        \item $d_{i+1}$ is obtained by rotating $d_i$ $45^\circ$ clockwise, and $\mathrm{length}(d_{i+1}) = \frac{\mathrm{length}(d_i)}{\sqrt{2}}$
        
        \item $r_i$ and $q_i$ must be contained within the square $\mathcal{Q}_{i-1}$
        
        \item If $d_0$ and $d_i$ are orthogonal, the robots must remain stationary
    \end{enumerate}
\end{Definition}                

\begin{figure}[H]
    \centering
    \includegraphics[width=8cm]{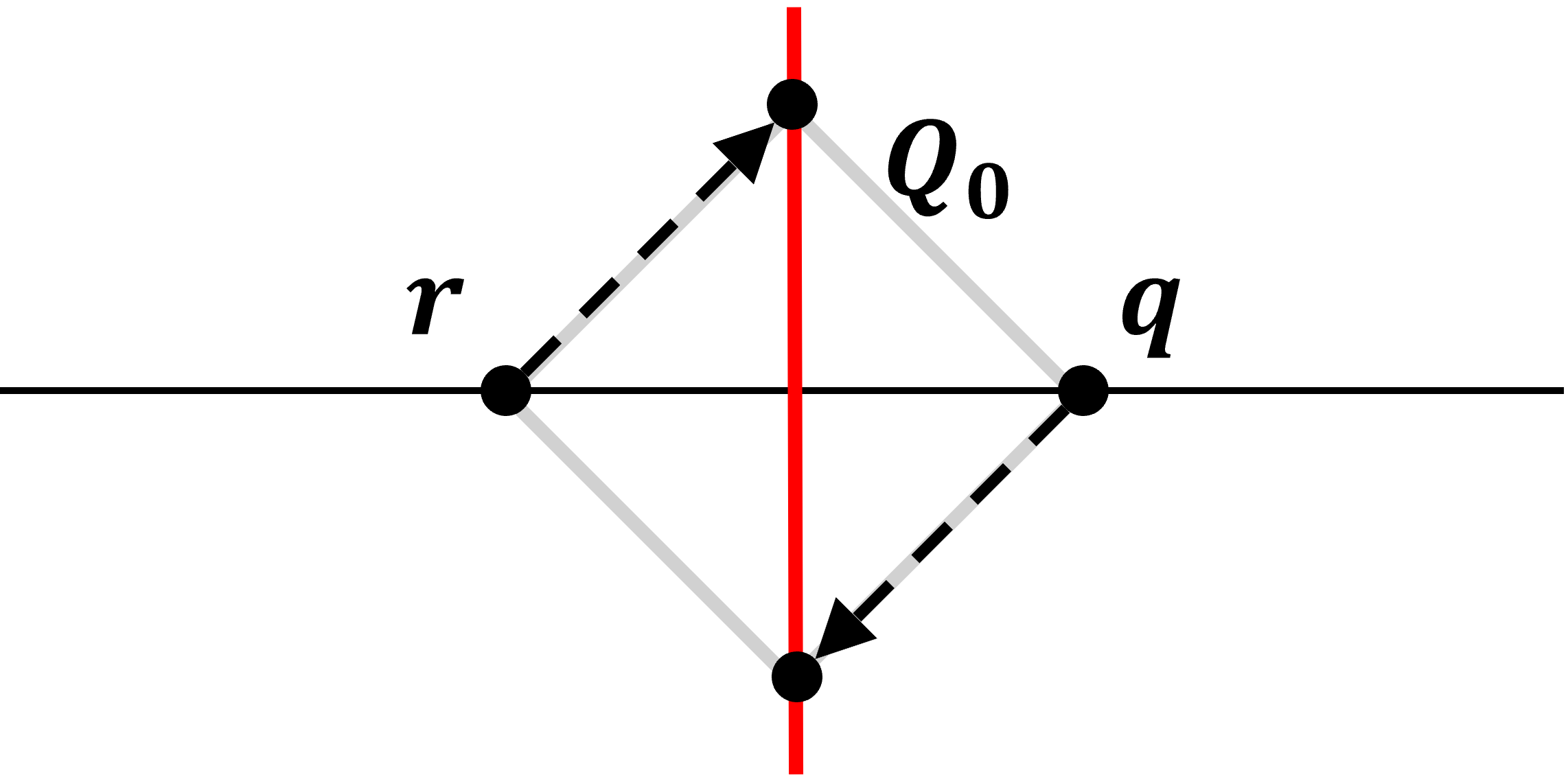}
    \caption{Movement of the \textit{VEC} problem.}
    \label{fig:vec}
\end{figure}

\begin{lemma}
    $\textit{VEC} \notin \FC^S$.
    \label{lem:20}
\end{lemma}
\begin{proof}
    Suppose there exists an algorithm $\calA$ that solves \textit{VEC} under any \SSY\ scheduler. Consider a case where an adversary continuously activates only robot $r$ from the initial configuration. Since we consider the $\FC$ model, robots determine their state based on the color of the other robot, not their own. Therefore, each time $r$ is activated, it may interpret its current observation as the initial configuration. In this case, the adversary can continue to activate only $r$ until condition 5 is violated. This behavior contradicts the specification of \textit{VEC} and the assumption that $\calA$ solves \textit{VEC} under any \SSY\ scheduler. Consequently, there exists no valid algorithm that solves \textit{VEC} under $\FC^S$.
\end{proof}
\begin{figure}[H]
        \centering
        \includegraphics[width=8cm]{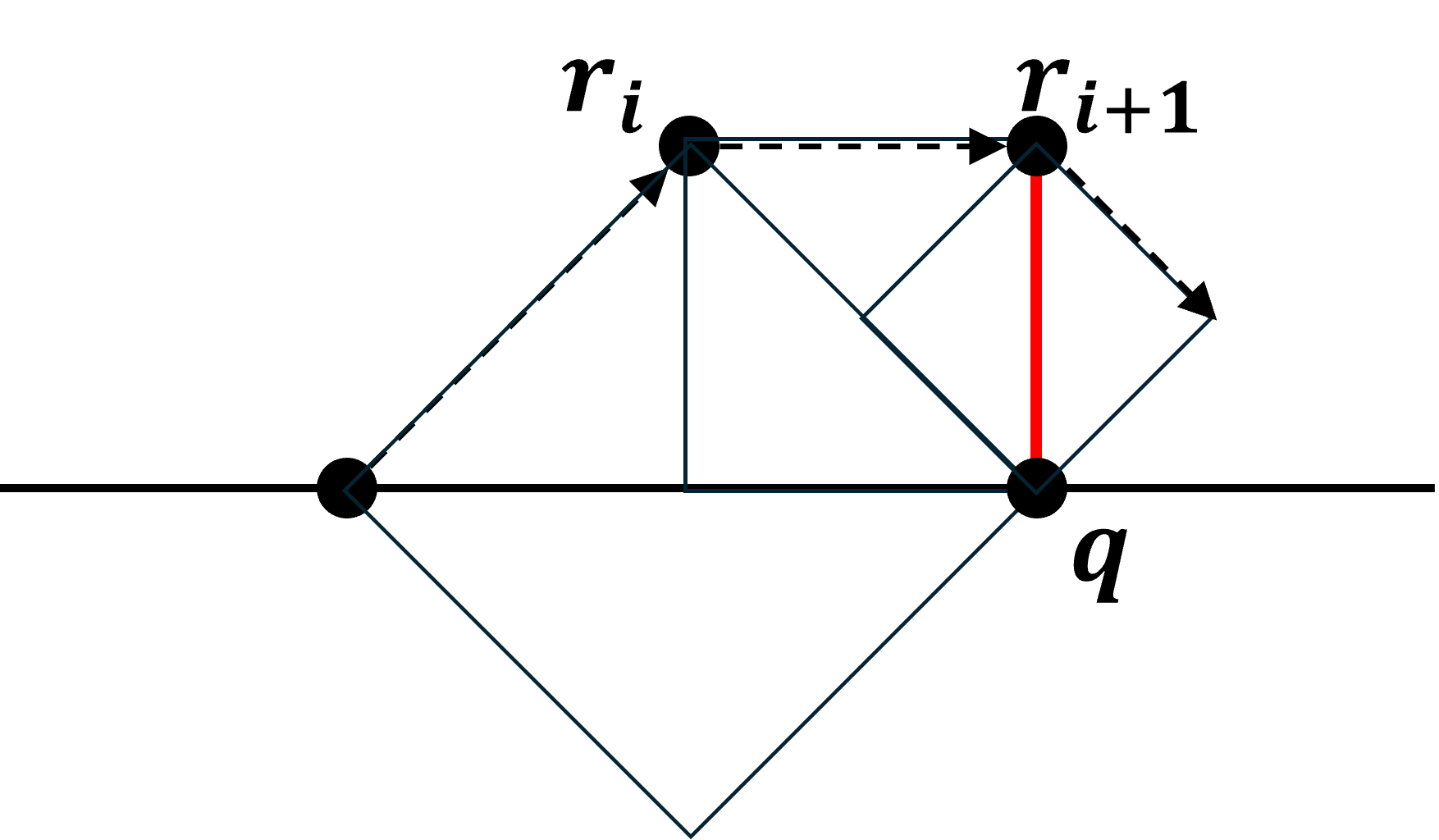}
        \caption{Movement of $r$ by $\calA$.}
        \label{fig:VEC_FCOM}
\end{figure}
    
\begin{lemma}
    $\textit{VEC} \notin \FS^{A_{LC}}$.
    \label{lem:21}
\end{lemma}
\begin{proof}
    By contradiction, assume there exists an algorithm $\calA$ that always solves \textit{VEC} for two robots under the $A_{LC}$ scheduler. Consider an execution of $\calA$ where robot $r$ is activated and moves at time $t$. While $r$ is moving, the remaining robot $q$ becomes active at time $t' > t$. Under the $\FS$ model, $q$ cannot access the internal state of $r$, nor can it remember previously observed angles or distances. Therefore, $q$ cannot detect that the observed configuration is not the initial configuration. As a result, $q$ may move in a square of a different size than $\mathcal{Q}_0$, violating Condition 3. This contradicts the assumption that algorithm $\calA$ is correct.
\end{proof}
\begin{figure}[H]
        \centering
        \includegraphics[width=5cm]{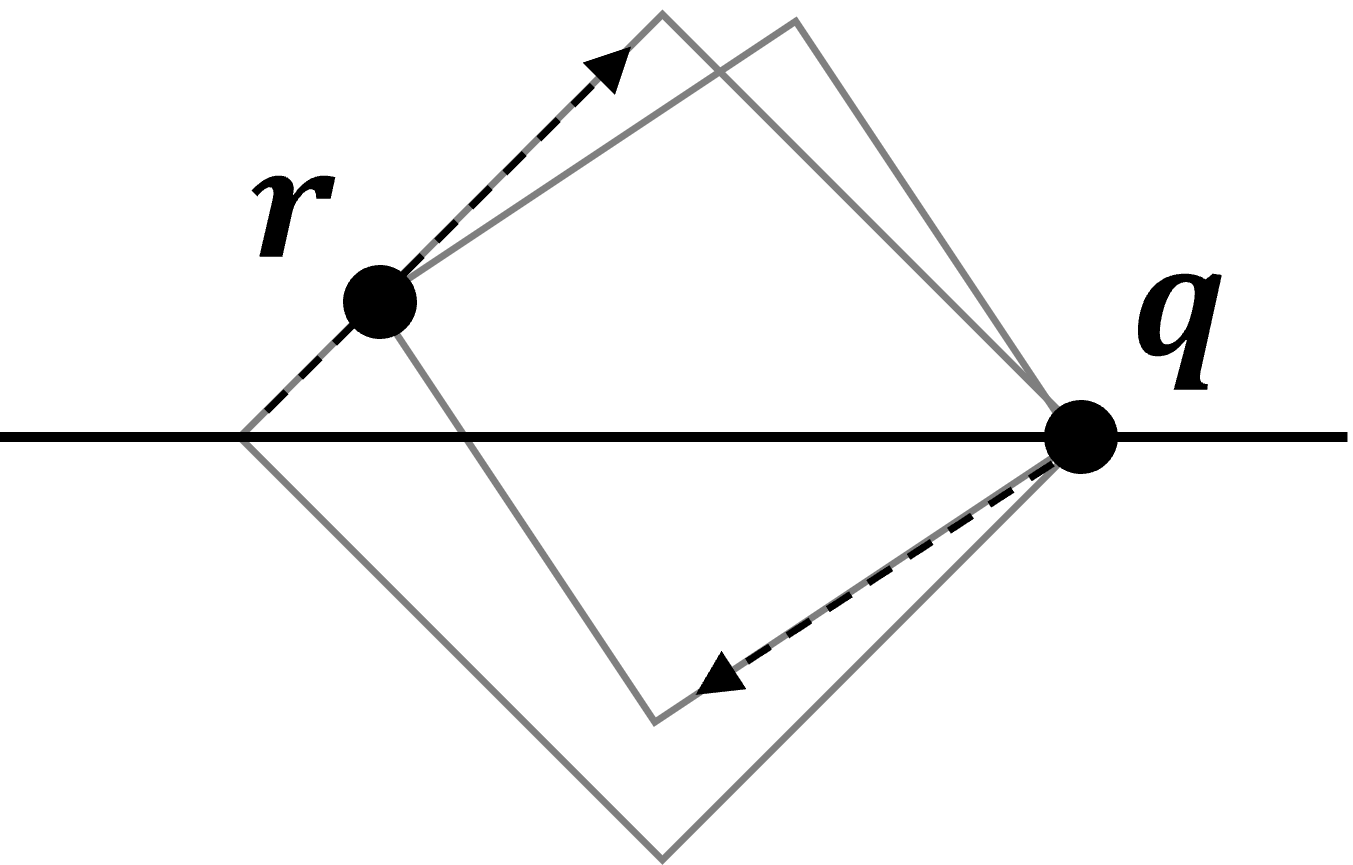}
        \caption{Two possible squares.}
        \label{fig:VEC_FSTA}
\end{figure}
    
\begin{lemma}
    $\textit{VEC} \in \FS^{A_{M}}$.
    \label{lem:22}
\end{lemma}
\begin{proof}
    The problem is solved by \textbf{Algorithm \ref{alg:VEC_FSTA_AM}}. Under the $M$-{\bf atomic}-\ASY\ scheduler, the square observed by an activated robot is always $Q_i$. Furthermore, due to the $\FS$ model, the robot can remain stationary when Condition 4 is satisfied.
\end{proof}   
\clearpage
\begin{algorithm}
    \caption{Alg \textit{VEC} for robot $r$}
    \label{alg:VEC_FSTA_AM}
    \textbf{Assumptions:} $\FS$, $M$-{\bf atomic}-\ASY\\
    \textbf{Light:} $r.\text{state} \in \{1, 2\}$, with the initial value set to $1$.\\
    
    $Phase Look:$ $r.\text{pos},\ \text{other.pos} \gets$ the positions of robot $r$ and the other robot, respectively.\\
    
    $Phase Compute:$
   \begin{algorithmic}[1]
        \State \textbf{case} $r.\text{state}$ \textbf{of}
        \State \hspace{1em} $1$:
        \State \hspace{2em} $r.\text{state} \gets 2$
        \State \hspace{2em} $r.\text{des} \gets$ the point obtained by rotating $r.\text{pos}$ $90^\circ$ clockwise around the midpoint of $r.\text{pos}$ and $\text{other.pos}$
        \State \hspace{1em} $2$:
        \State \hspace{2em} $r.\text{des}$ \Comment{Stay at the current location}
    \end{algorithmic}

    \vspace{1em}
    $Phase Move:$\\
    Move to $r_i.des$.
\end{algorithm}

\begin{figure}[H]
    \centering
    \includegraphics[width=12cm]{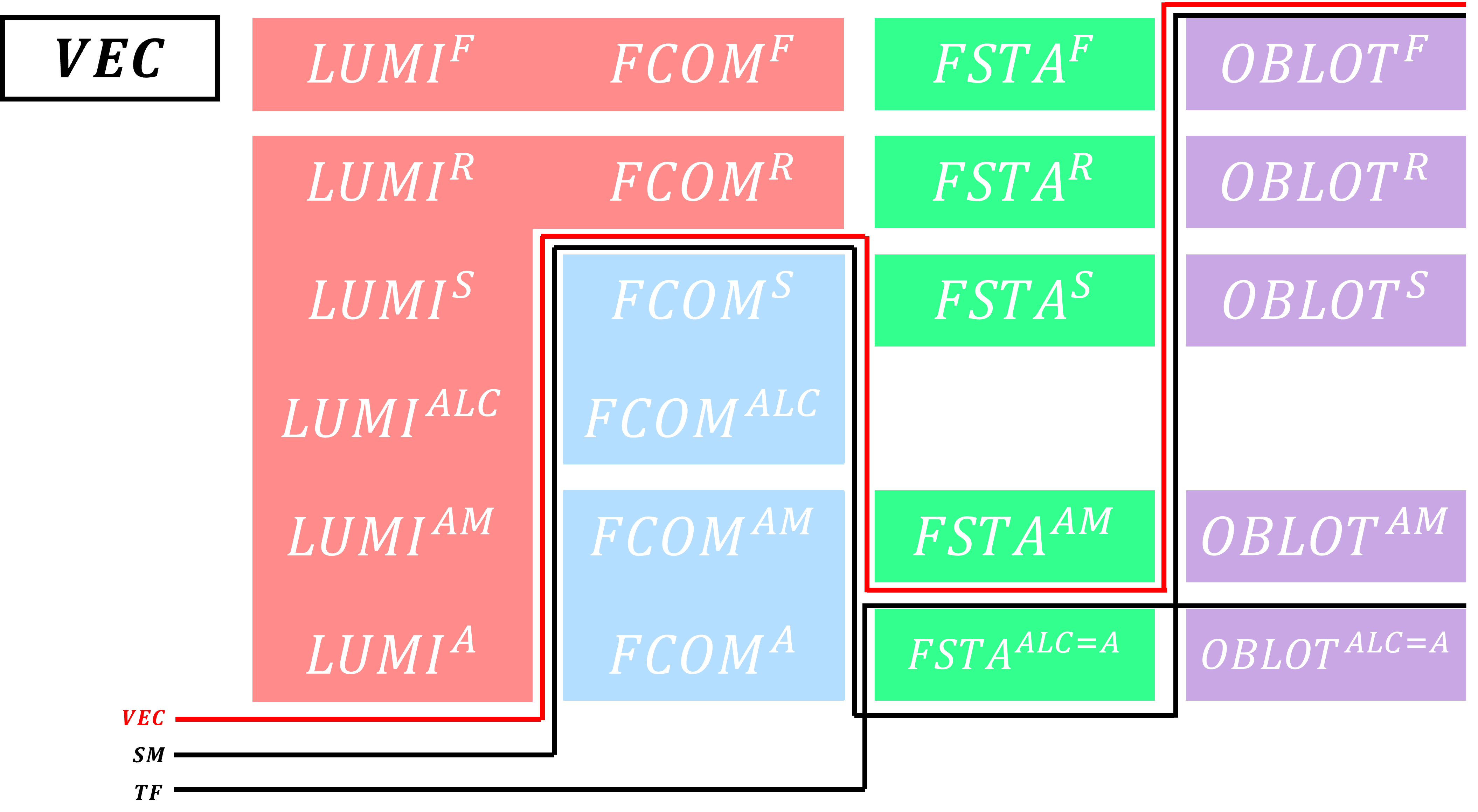}
    \caption{Boundary of VEC}
    \label{fig:VEC_boundary}
\end{figure}

\begin{theorem}
    \textit{VEC} is solved in $\FS^{A_{M}}$, but not in any of $\OB^F$ or $\FC^S$ or $\FS^{A_{LC}}$.
    \label{th:vect}
\end{theorem}
  \label{arc:ZCC}
\begin{Definition}
    \textit{ZCC} (Zig-Zag Collision-free Convergence)
    \\Two robots $r$ and $q$ are placed arbitrarily at distinct locations such that they do not overlap.  
    Each robot must solve the collision-free linear convergence problem while oscillating along the straight line connecting $r$ and $q$.  
    There exists an infinite, strictly increasing sequence of time points $t_0, t_1, t_2, \ldots$, and the problem \textit{ZCC} is defined by the following logical formula:
    \begin{align*}
        \textit{ZCC} \equiv
        \big[ &\{\exists l \in \mathbb{R}^2, \forall \epsilon \geq 0, \exists T \geq 0, \forall t \geq T : 
            |r(t) - l| + |q(t) - l| \leq \epsilon \}, \\
            &\text{and } \{\forall t \geq 0 : \text{dis}(r(0), r(t)) \leq \text{dis}(r(0), q(t)), 
            \text{dis}(q(0), q(t)) \leq \text{dis}(q(0), r(t)) \}, \\
            &\text{and } \{\text{dis}(r(0), r(t_{2i+1})) < \text{dis}(r(0), r(t_{2i})), \\
            &\quad \forall h', h'' \in [t_{2i}, t_{2i+1}], h' < h'' : 
            \text{dis}(r(t_{2i+1}), r(h'')) \leq \text{dis}(r(t_{2i+1}), r(h')) \}, \\
            &\text{and } \{\text{dis}(r(0), r(t_{2i})) > \text{dis}(r(0), r(t_{2i-1})), \\
            &\quad \forall h', h'' \in [t_{2i-1}, t_{2i}], h' < h'' : 
            \text{dis}(r(t_{2i-1}), r(h'')) \geq \text{dis}(r(t_{2i-1}), r(h')) \}, \\
            &\text{and } \{\text{dis}(r(0), r(t_{2i+1})) > \text{dis}(r(0), r(t_{2i-1}))\}
        \big]
    \end{align*}
\end{Definition}
    
\begin{figure}[H]
    \centering
    \includegraphics[width=5cm]{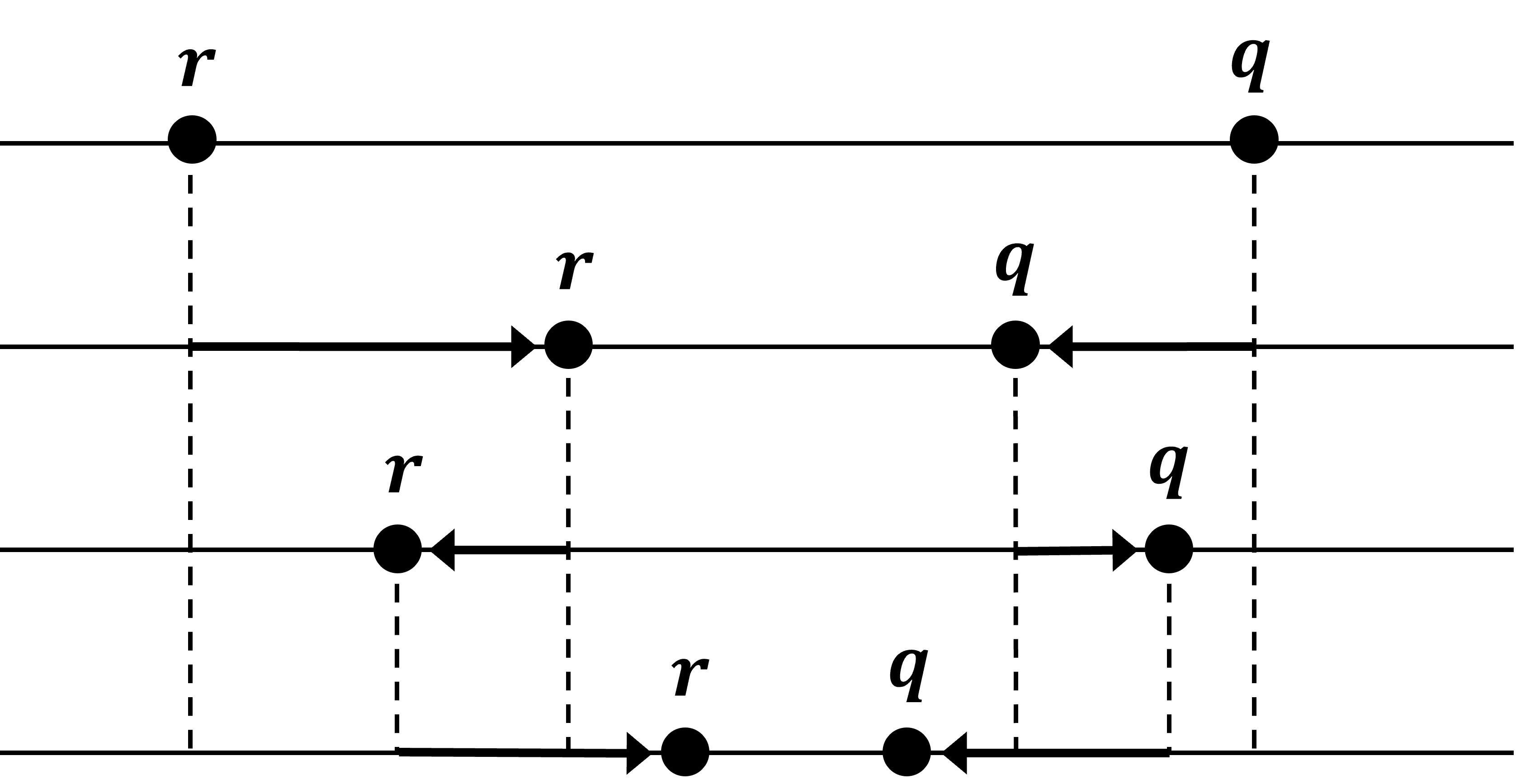}
    \caption{Movement of the \textit{ZCC} problem.}
    \label{fig:os}
\end{figure}

\begin{lemma}
    $\textit{ZCC}\notin \FS^{A_{M}}$.
    \label{lem:23}
\end{lemma}
\begin{proof}
    Assume, for the sake of contradiction, that there exists an algorithm $\calA$ that solves the \textit{ZCC} problem under the $\FS^{A_{M}}$ model. Consider two robots $r$ and $q$ that share the same unit distance and initially observe each other along the positive direction of their local $X$-axes. We assume that their local coordinate systems do not change during the execution of $\calA$. The following three observations are made:

    \begin{enumerate}
        \item By the predicate that defines the \textit{ZCC} problem, a robot must move toward the other whenever it moves. In this particular setup, we observe that the robots must remain on their respective $X$-axes.
        
        \item Every time a robot is activated and executes $\calA$, it must move. Conversely, if $\calA$ specifies that a robot activated at a certain distance $d$ from the other robot must not move, then when $\calA$ is executed in perfect synchrony and the robots are initially at distance $d$, neither robot moves. Therefore, convergence does not occur.
        
        \item When robot $r$ observes $q$ at distance $d$ and moves toward $q$ along the $X$-axis, the computed move length $f(d)$ is the same as the length that $q$ would compute when it observes $r$ at distance $d$ and moves toward $r$.
    \end{enumerate}
    
    Now consider an execution $\epsilon$ under the $A_{M}$ scheduler: initially, both robots are activated simultaneously and are at distance $d$ from each other.  
    Robot $r$ completes its computation and moves instantaneously (note that it is operating under the $A_{M}$ scheduler), while robot $q$ continues its activation and executes $\calA$ during its $\LK$ and $\CP$ phases.
    
    Clearly, each move by $r$ reduces the distance between the two robots. More precisely, from observation (3), after $k = 2l \geq 0$ moves, the distance has decreased from $d$ to $d_k$, where:
    
    \[
    d_0 = d, \quad 
    d_{k>0} = d_{k-1} - f(d_{k-1}) = (d_{k-2} + f(d_{k-2})) - f(d_{k-1}) 
    = d - \sum_{0 \leq i < l} \left(f(d_{2i}) - f(d_{2i+1})\right)
    \]
    
    By the definition of the problem, we know that $f(d_i) - f(d_{i+1}) > 0$ always holds.
\end{proof}

\begin{lemma}
    After a finite number of movements of $r$, the distance between the two robots becomes smaller than $f(d)$.
    \label{lem:24}
\end{lemma}
\begin{proof}
    Assume, for the sake of contradiction, that $r$ never comes within a distance less than $f(d)$ from $q$.
    That is, for all $k = 4n + 2 > 0$, we have $d_k > f(d)$.
    
    Now consider the execution $\hat{\epsilon}$ of $\calA$ under a Round-Robin synchronous scheduler.  
    Initially, the two robots are at distance $d$, and in each round, only one robot is activated in an alternating fashion.
    
    If $\calA$ is correct under the $A_{M}$ scheduler, it must also be correct under the Round-Robin scheduler.
    This means that, starting from distance $d$, the two robots must eventually become closer than any arbitrary positive constant $d' > 0$.
    Let $m(d')$ be the number of rounds required for this to happen.
    
    Suppose $m(f(d)) = 4n + 2$ rounds are sufficient to reduce the distance to less than $f(d)$.
    Furthermore, after round $i = 4j\ (j > 0)$, the distance $d_i$ is reduced by $f(d_i)$. Then we have:
    
    \[
    d_{m(f(d))} = d - f(d_0) - f(d_1) - \sum_{0 < i < n} \left( f(d_{4i}) + f(d_{4i+1}) - f(d_{4i-2}) - f(d_{4i-1}) \right) < f(d)
    \]
    
    Since, by the definition of the problem, each $f(d_{4i}) + f(d_{4i+1}) - f(d_{4i-2}) - f(d_{4i-1})$ is positive, this contradicts the assumption that $d_k > f(d)$ for all $k > 0$. $\square$
    
    Now consider the execution $\epsilon$ when the distance between the two robots becomes less than $f(d)$.  
    If robot $q$ completes its computation at this point and moves a distance $f(d)$ toward $r$, a collision will occur, contradicting the correctness of $\calA$.
\end{proof}

\begin{lemma}
    $\textit{ZCC} \notin \FC^{S}$.
    \label{lem:25}
\end{lemma}
\begin{proof}
    Assume that there exists an algorithm $A$ that solves the \textit{ZCC} problem under an arbitrary \SSY\ scheduler.  
    In particular, suppose that it solves the problem under a specific scheduler $S$.  
    Now, consider another scheduler $S'$ under which \textit{ZCC} cannot be solved.
    In order for algorithm $A$ to solve the problem under scheduler $S$, robot $r$ must approach robot $q$ in round $i$.  
    Let scheduler $S'$ behave identically to $S$ up to round $i$, but in round $i+1$, let $S'$ activate only robot $r$.  
    In this case, robot $q$ is not activated and thus cannot communicate to $r$ that $r$ had moved in the previous round.                
        Due to the assumption of variable disorientation, an adversary can force robot $r$ to obtain the same observation in round $i+1$ as in round $i$.  
    As a result, $r$ repeats the same move toward $q$ in consecutive rounds.                
    This leads to an invalid configuration. Hence, algorithm $A$ cannot solve \textit{ZCC} under scheduler $S'$.  
    Since $S'$ is a valid \SSY\ scheduler, and $A$ does not solve the problem under it, this contradicts the assumption that $A$ solves \textit{ZCC} under **any** \SSY\ scheduler.               
    Therefore, we conclude that no valid algorithm can solve \textit{ZCC} under the $\FC^{S}$ model with a \SSY\ scheduler.
\end{proof}

\begin{lemma}
    $\textit{ZCC} \in \FS^{S}$.
    \label{lem:26}
\end{lemma}
\begin{proof}
    The problem is solved by \textbf{Algorithm~\ref{alg:ZCC_FSTA_S}}.  
    Each robot is equipped with two light colors, and the initial color is assumed to be $A$.  
    When a robot is activated, if its own light is $A$, it switches the light to $B$ and moves toward the other robot.  
    If the light is $B$, it switches back to $A$ and moves away from the other robot.  
    This behavior enables the robot to oscillate.
    Furthermore, let $d$ be the distance between robots $r$ and $q$ at the time of activation.  
    If the robot moves toward the other robot by $d/2$ and away by $d/4$, the two robots can converge.  
    If both robots are activated simultaneously, they will meet at the midpoint of distance $d$.
    Even when one robot is activated continuously, it will not move further away from the other robot than its position two activations ago. This is also the case where chirality is undefined.
\end{proof}   
\clearpage
\begin{algorithm}
    \caption{Alg \textit{ZCC} for robot $r$}
    \label{alg:ZCC_FSTA_S}
    \textbf{Assumptions:} $\FS$, \SSY \\
    \textbf{Light:} $r.\text{color} \in \{A, B\}$, initial color is $A$\\
    
    $Phase Look:$$\text{other.pos} \gets$ the position of the other robot in $r$'s local coordinate system.\\
    
    $Phase Compute:$
    \begin{algorithmic}[1]
        \State \textbf{case} $r.\text{color}$ \textbf{of}
        \State \hspace{1em} $A$:
        \State \hspace{2em} $r.\text{color} \gets B$
        \State \hspace{2em} \parbox[t]{0.8\linewidth}{$r.\text{des} \gets$ the point in the direction of $other.\text{pos}$ at distance $|\text{other.pos}|/2$ from $r$}
        \State \hspace{1em} $B$:
        \State \hspace{2em} $r.\text{color} \gets A$
        \State \hspace{2em} \parbox[t]{0.8\linewidth}{$r.\text{des} \gets$ the point in the opposite direction from $other.\text{pos}$ at distance $|\text{other.pos}|/4$ from $r$}
    \end{algorithmic}

    \vspace{1em}
    $Phase Move:$\\
    Move to $r_i.des$.
\end{algorithm}

\begin{figure}[H]
    \centering
    \includegraphics[width=12cm]{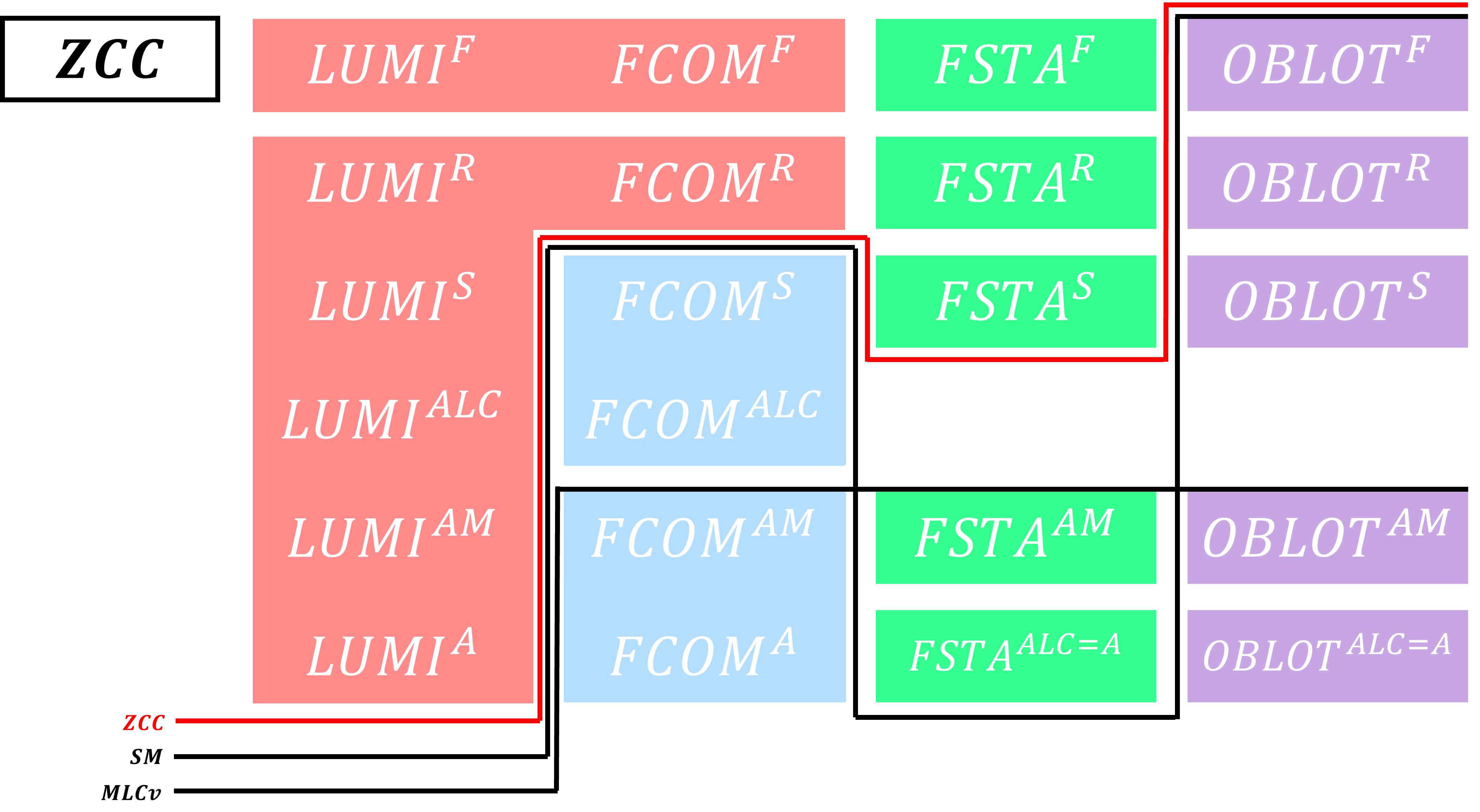}
    \caption{Boundary of ZCC}
    \label{fig:ZCC_boundary}
\end{figure}

\begin{theorem}
    \textit{ZCC} is solved in $\FS^S$, but not in any of $\FS^{A_{M}}$ or $\FC^S$.
    \label{th:zcct}
\end{theorem}

  \subsection{Synchrony Is Necessary for FSTA (LP-Cv, VTR)}\label{sec:5-2}

  \label{arc:LCM-M}
\begin{Definition}
    \textit{LP--Cv} (Leave the Place -- Convergence)
    \\We place two robots, $r$ and $q$, at arbitrary, non-overlapping positions. Each robot then moves along the line connecting $r$ and $q$, away from the other robot (this phase is referred to as $LP$). After LP, the robots must solve the convergence problem (referred to as $Cv$). In this context, when $Cv$ is executed after completing $LP$, both robots must have moved from their initial positions. Assuming robot $r$ is the one located in the negative direction and robot $q$ is in the positive direction, the \textit{LP--Cv} behavior is defined by the following logical expression.
    \begin{align*}
        \textit{LP--Cv} \equiv 
        \big[ &\exists T \geq 0 : \big[
            \{ \forall t \leq T : r(0), q(0) \in \overline{r(t)q(t)} \} \\
            &\land \{ r(T) \neq r(0), \, q(T) \neq q(0) \} \\
            &\land \{ \forall t \leq T : r(t) - r(0) \leq 0, \, q(t) - q(0) \geq 0 \}
        \big] \\
        &\land \big[
            \exists l \in \mathbb{R}^2, \, \epsilon \geq 0, \, \exists T' \geq T, \, \forall t \geq T' : 
            |r(t) - l| + |q(t) - l| \leq \epsilon
        \big] \big]
    \end{align*}
\end{Definition}

\begin{lemma}
    $\textit{LP--Cv} \notin \OB^F$.
    \label{lem:27}
\end{lemma}
\begin{proof}
    Regardless of full-synchrony, $\OB$ robots cannot determine whether they have moved from their initial positions based solely on observation.
\end{proof}

\begin{lemma}
    $\textit{LP--Cv} \notin \FS^S$.
    \label{lem:28}
\end{lemma}
\begin{proof}
    Suppose that there exists an algorithm $\calA$ that solves \textit{LP--Cv} under the $\FS^S$ model. Consider a scenario in the \textit{LP} problem where a robot $r$, having moved away from robot $q$, is activated in round $i$. In this round, robot $r$ has two possible actions: either move toward $q$ to execute \textit{Cv}, or remain stationary until $q$ completes \textit{LP}.

    Assume that round $j$ is the first round in which $r$ performs \textit{MLCv}. By this round, both $r$ and $q$ must have completed the LP phase. Now, consider a round $k$ before $q$ has performed its LP phase. Due to variable disorientation, the adversary can activate only $r$ from round $k$ until round $j$ to bring the system to an identical state from $r$'s perspective. In this scenario, $r$ would initiate \textit{MLCv} without $q$ having performed the LP phase, which means algorithm $\calA$ is incorrect.

    Hence, in the $\FS$ model, robot $r$ cannot determine, with its finite memory, whether $q$ is still at its initial position or has already moved away. Therefore, if algorithm $\calA$ dictates that $r$ should move toward $q$, it may do so even when $q$ remains at its initial position.

    On the other hand, if $\calA$ instructs $r$ to stay put, then $r$ will never observe that $q$ has moved away from its initial position, and will continue to remain stationary even if activated infinitely many times after round $i$. Both of these behaviors contradict the problem specification and violate the correctness of algorithm $\calA$.
\end{proof}

\begin{lemma}
    $\textit{LP--Cv} \in \FS^R$.
    \label{lem:29}
\end{lemma}
\begin{proof}
   The problem is solved by \textbf{Algorithm~\ref{alg:LPCv_FSTA_R}}.  
    Due to the property of the \RSY\ scheduler, when any robot is activated for the second time, the other robot must have already been activated at least once.  
    This implies that the other robot has completed the movement away from the initial position (\textit{LP}).  
    Therefore, even without being able to recognize the state of the other robot via lights, each robot can solve the problem by performing \textit{LP} on its first activation and executing \textit{Cv} on subsequent activations.
\end{proof}

\begin{algorithm}
    \caption{Alg \textit{LP--Cv} for robot $r$}
    \label{alg:LPCv_FSTA_R}
    \textbf{Assumptions:} $\FS$, \text{R\textsc{synch}}\\
    \textbf{Light:} $r.\text{state} \in \{A, B\}$, initial value is $A$\\
    
    $Phase Look:$ $other.pos \gets$ the position of the other robot in $r$'s local coordinate system\\
    
    $Phase Compute:$
    \begin{algorithmic}[1]
        \State \textbf{case} $r.\text{state}$ \textbf{of}
        \State \hspace{1em} $A$:
        \State \hspace{2em} $r.\text{state} \gets B$
        \State \hspace{2em} $r.\text{des} \gets$ the point at distance $|\texttt{other.pos}|/2$ from $r$ in the direction opposite to \texttt{other.pos}
        \State \hspace{1em} $B$:
        \State \hspace{2em} $r.\text{des} \gets$ the point at distance $|\texttt{other.pos}|/2$ from $r$ toward \texttt{other.pos}
    \end{algorithmic}

    \vspace{1em}
    $Phase Move:$\\
    Move to $r_i.des$.
\end{algorithm}

\begin{lemma}
    $\textit{LP--Cv} \in \FC^{A_{M}}$.
    \label{lem:30}
\end{lemma}
\begin{proof}
    By using three colors of lights, it is possible to determine whether the other robot has moved away from the initial position in the \textit{LP} phase.  
    Specifically, each robot makes the following decisions:
    
    \begin{enumerate}
        \item If the other robot's light is $A$, then the other robot has not been activated even once, and this robot itself may have been activated either once before or not at all.  
        Therefore, this robot changes its own light to $B$ and moves away from the other robot.
        
        \item If the other robot's light is $B$, then the other robot has been activated once before, and this robot itself may have been activated either once before or not at all.  
        Therefore, this robot changes its own light to $C$ and moves away from the other robot.
        
        \item If the other robot's light is $C$, then both robots have been activated at least once before.  
        Therefore, this robot changes its own light to $C$ and executes the \textit{Cv} movement.
    \end{enumerate}
    
    In the \textit{Cv} phase, the process of (3) is repeated, and previous research has shown that \textit{Cv} $\in \FC^{A_{M}}$.
\end{proof}
\clearpage
\begin{algorithm}
    \caption{Alg \textit{LP--Cv} for robot $r$}
    \label{alg:LPCv_FCOM_AM}
    \textbf{Assumptions:} $\FC$, $M$-{\bf atomic}-\ASY\\
    \textbf{Light:} $r.\text{state}, \text{other.state} \in \{A, B, C\}$, initial value is $A$\\
    
    $Phase Look:$ $other.pos \gets$ position of the other robot in $r$'s local coordinate system.\\
    
    $Phase Compute:$
    \begin{algorithmic}[1]
        \State \textbf{case} other.state \textbf{of}
        \State \hspace{1em} $A$:
        \State \hspace{2em} $r.\text{state} \gets B$
        \State \hspace{2em} $r.\text{des} \gets$ the point at distance $|$other.pos$|/2$ from $r$ in the opposite direction to other.pos
        \State \hspace{1em} $B$:
        \State \hspace{2em} $r.\text{state} \gets C$
        \State \hspace{2em} $r.\text{des} \gets$ the point at distance $|$other.pos$|/2$ from $r$ in the opposite direction to other.pos
        \State \hspace{1em} $C$:
        \State \hspace{2em} $r.\text{state} \gets C$
        \State \hspace{2em} $r.\text{des} \gets$ the point at distance $|$other.pos$|/2$ from $r$ toward other.pos
    \end{algorithmic}

    \vspace{1em}
    $Phase Move:$\\
    Move to $r_i.des$.
\end{algorithm}

\begin{figure}[H]
    \centering
    \includegraphics[width=11cm]{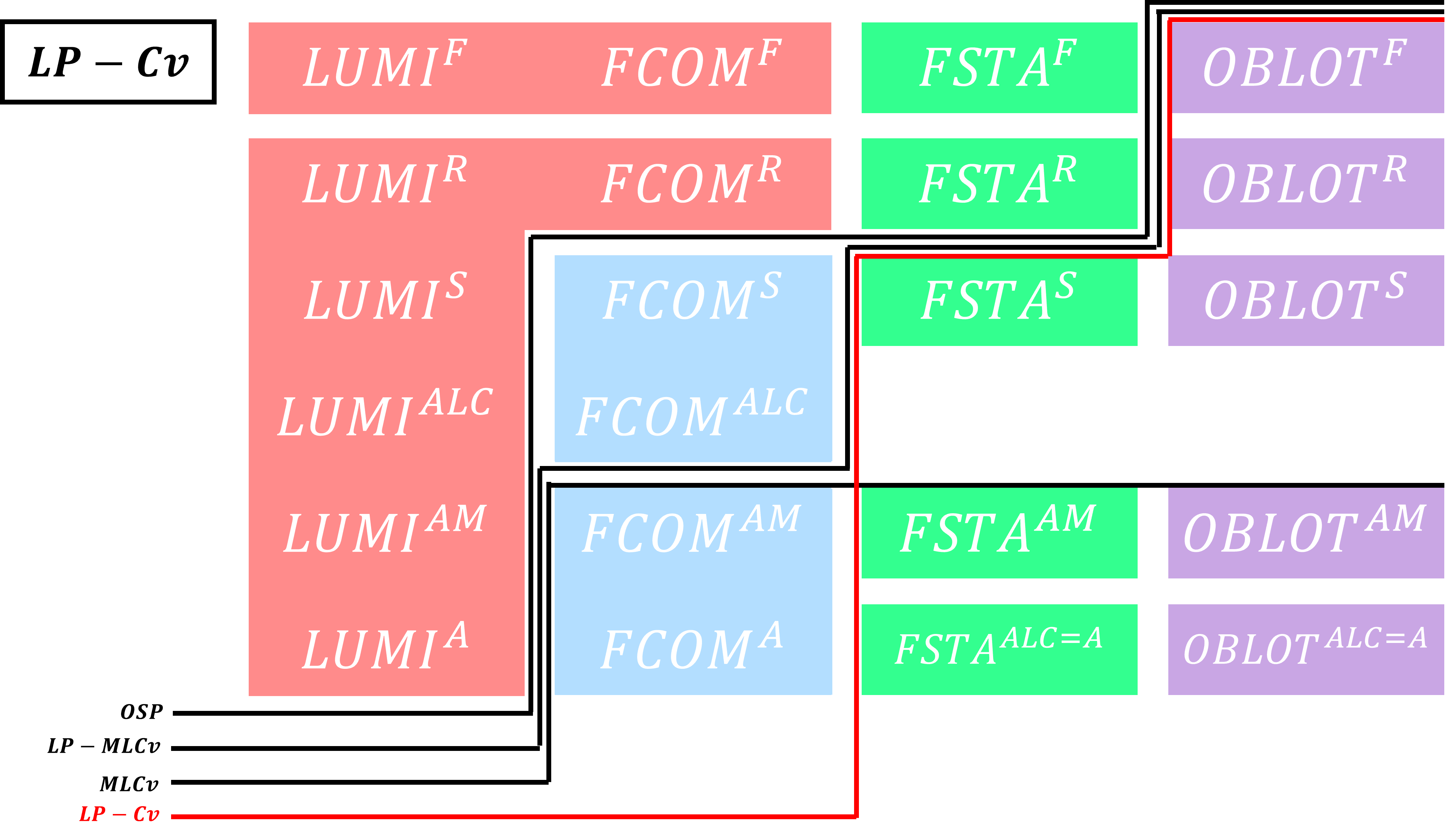}
    \caption{Boundary of LP--Cv}
    \label{fig:LP-Cv_boundary}
\end{figure}

\begin{theorem}
    \textit{LP--Cv} is solved in both of $\FS^R$ and $\FC^{A_{M}}$, but not in any of $\OB^F$ or $\FS^S$.
    \label{th:lpcvt}
\end{theorem}
  \label{arc:VTR}
\begin{Definition}
    \textit{VTR} (Vertex Traversal)
    \\Three robots $r$, $q$, and $s$ are placed at arbitrary vertices of a regular hexagon, arranged so that the vertices on the diagonals remain unoccupied. In this problem, each robot is required to perform the operation of moving to a vertex on a diagonal and then returning to its initial position exactly once.
    Let the vertices of the hexagon be denoted as $v_0, v_1, v_2, v_3, v_4, v_5$, and suppose that the initial position of each robot is $v_i$, the \textit{VTR} behavior is defined by the following logical expression.

    \begin{align*}
        \textit{VTR} \equiv 
        \quad & (0 \leq \forall t \leq t_1 : r(t) = v_i) \land (t_2 \leq \forall t \leq t_3 : r(t) = v_{i+3 \bmod 6}) \\
        & \land (\exists t,t' : t_1 \leq t < t'\leq t_2 \rightarrow r(t),r(t') \in \overline{v_i v_{i+3 \bmod 6}}, r(t) \ne r(t'), \\
        & \text{dis}(r(t), v_i) < \text{dis}(r(t'), v_i), \text{dis}(r(t), v_{i+3 \bmod 6}) > \text{dis}(r(t'), v_{i+3 \bmod 6}))\\
        & \land (t_4 \leq \forall t \leq t_5 : r(t) = v_i) \\
        & \land (\exists t,t' : t_3 \leq t < t'\leq t_4 \rightarrow r(t),r(t') \in \overline{v_{i+3 \bmod 6} v_i}, r(t) \ne r(t')), \\
        & \text{dis}(r(t), v_i) > \text{dis}(r(t'), v_i), \text{dis}(r(t), v_{i+3 \bmod 6}) < \text{dis}(r(t'), v_{i+3 \bmod 6}))\\
    \end{align*}
\end{Definition}

\begin{figure}[H]
    \centering
    \includegraphics[width=3cm]{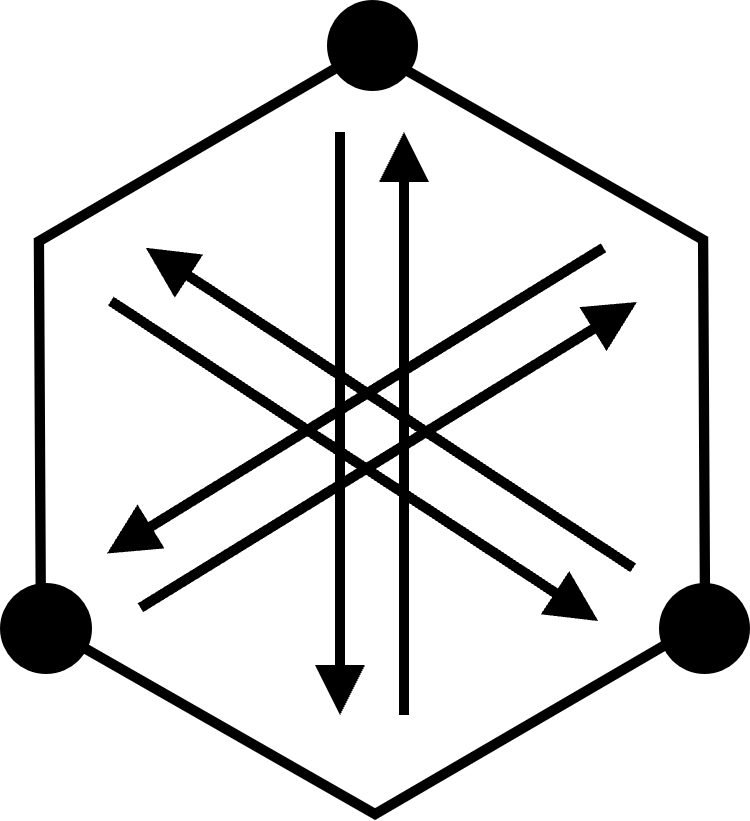}
    \caption{Movement of the \textit{VTR} problem.}
    \label{fig:VTRapd}
\end{figure}

\begin{lemma}
    $\textit{VTR} \notin \OB^F$.
    \label{lem:31}
\end{lemma}
\begin{proof}
    Even if perfect synchrony is guaranteed, $\OB$ robots cannot determine from their observations whether they should move toward a vertex on the diagonal or remain at that vertex. Therefore, it is impossible to solve the \textit{VTR} problem.
\end{proof}

\begin{lemma}
    $\textit{VTR} \notin \FS^{A_{LC}}$.
    \label{lem:32}
\end{lemma}
\begin{proof}
    By contradiction, suppose that under the $\mathcal{A_{LC}}$ scheduler there exists an algorithm $\calA$ by which three robots always solve the \textit{VTR} problem.  
    Consider the initial configuration where the three robots form an equilateral triangle.  
    At time $t$, robots $r$ and $q$ are activated and execute $\calA$ to move toward vertices on the diagonals.  
    
    Suppose that robot $s$ is activated at time $t' > t$. The configuration observed by $s$ may still be the same equilateral triangle as the initial configuration.  
    However, in the $\FS$ model, $s$ cannot access the internal states of $r$ and $q$, nor can it retain observation history, so it cannot detect that the configuration is an equilateral triangle of different size from the initial one. As a result, $s$ moves to a vertex of a different hexagon, violating the specification of \textit{VTR}.
\end{proof}
\begin{figure}[H]
    \centering
    \includegraphics[width=4cm]{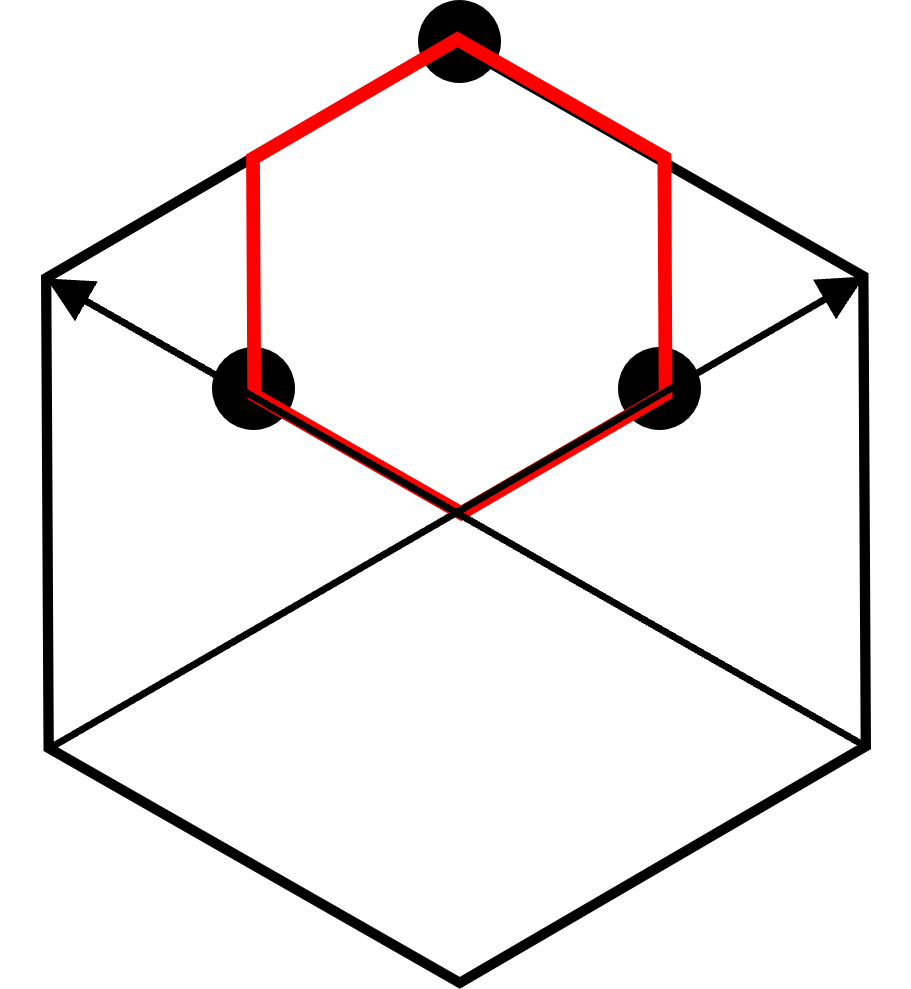}
    \caption{Another possible regular hexagon.}
    \label{fig:vtrp}
\end{figure}

\begin{lemma}
    $\textit{VTR} \in \FS^{A_{M}}$.
    \label{lem:33}
\end{lemma}
\begin{proof}
    The problem is solved by \textbf{Algorithm \ref{alg:VTR_FSTA_AM}}.  
    Under the $M$-{\bf atomic}-\ASY\ scheduler, when a robot is activated and performs its observation, the other robots are also positioned at vertices of a regular hexagon. Therefore, the circumscribed circle corresponding to the robots' configuration is uniquely determined, and the regular hexagon based on it is also uniquely determined. Thus, movement to the appropriate vertices is possible, enabling the \textit{VTR} problem to be solved.
\end{proof}

\begin{algorithm}
    \caption{Alg \textit{VTR} for robot $r$}
    \label{alg:VTR_FSTA_AM}
    \textbf{Assumptions:} $\FS$, $M$-{\bf atomic}-\ASY\\
    \textbf{Light:} $r.\text{state} \in \{A, B, C\}$, initial value is $A$\\
    
    $Phase Look:$ Observe the positions of itself and other robots (other robots' lights are not visible)\\
    
    $Phase Compute:$
    \begin{algorithmic}[1]
        \State \textbf{case} $r.\text{state}$ \textbf{of}
        \State \hspace{1em} $A$:
        \State \hspace{2em} $r.\text{state} \gets B$
        \State \hspace{2em} $r.\text{des} \gets$ the vertex on the diagonal
        \State \hspace{1em} $B$:
        \State \hspace{2em} $r.\text{state} \gets C$
        \State \hspace{2em} $r.\text{des} \gets$ the vertex on the diagonal
    \end{algorithmic}

    \vspace{1em}
    $Phase Move:$\\
    Move to $r_i.des$.
\end{algorithm}

\begin{lemma}
    $\textit{VTR} \in \FC^S$.
    \label{lem:34}
\end{lemma}
\begin{proof}
    The problem is solved by \textbf{Algorithm \ref{alg:VTR_FCOM_S}}. The possible initial configurations are limited to two types: the regular triangle configuration $\calC_1$, and the configuration in which three robots are placed on three consecutive vertices of a regular hexagon, denoted by $\calC_2$. In the $\FC$ model, robots must observe the lights of other robots in order to indirectly recognize their own number of activations. However, in order to ensure the reliability of light-based communication, it is necessary to strictly correspond to situations of simultaneous activation. Under such conditions, if the robots correctly update their states, the \textit{VTR} problem can be solved.
\end{proof}
\begin{algorithm}
    \caption{Alg \textit{VTR} for robot $r$}
    \label{alg:VTR_FCOM_S}
    \textbf{Assumptions:} $\FC$, \SSY \\
    \textbf{Light:}$r.\text{state} \in \{0, 1, 2\}$ \\
    
    $Phase Look:$ Observe one's own position, and the positions and lights of the other robots $q$ and $s$, where $q$ is the robot located at the vertex adjacent to $r$ in the clockwise direction, and $s$ is the robot located at the vertex adjacent to $r$ in the counterclockwise direction.

    $Phase Compute:$
    \begin{algorithmic}[1]
        \If{the configuration is $\calC_1$}
            \If{$q.\text{state} \neq 2$ or $s.\text{state} \neq 2$}
                \State $r.\text{state} \gets r.\text{state} + 1$
                \State $r.\text{des} \gets$ the vertex on the diagonal
            \EndIf
        \ElsIf{the configuration is $\calC_2$}
            \If{$\overline{rq} > \overline{rs}$ and 
            ($q.\text{state} \neq 2$ or $s.\text{state} \neq 1$ or 
            ($q.\text{state}(r) \neq 2$ and $q.\text{state}(r) = s.\text{state}(r)$))}
                \State $r.\text{state} \gets r.\text{state} + 1$
                \State $r.\text{des} \gets$ the vertex on the diagonal
            \ElsIf{$\overline{rq} = \overline{rs}$ and $q.\text{state} = s.\text{state} = 2$ and 
            $q.\text{state}(r) \neq 2$ and $s.\text{state}(r) \neq 2$}
                \State $r.\text{state} \gets r.\text{state} + 1$
                \State $r.\text{des} \gets$ the vertex on the diagonal
            \EndIf
        \EndIf
        \State $r.\text{state}(q) \gets q.\text{state}$
        \State $r.\text{state}(s) \gets s.\text{state}$
    \end{algorithmic}

    \vspace{1em}
    $Phase Move:$\\
    Move to $r_i.des$.
\end{algorithm}

\clearpage
\begin{figure}[H]
    \centering
    \includegraphics[width=12cm]{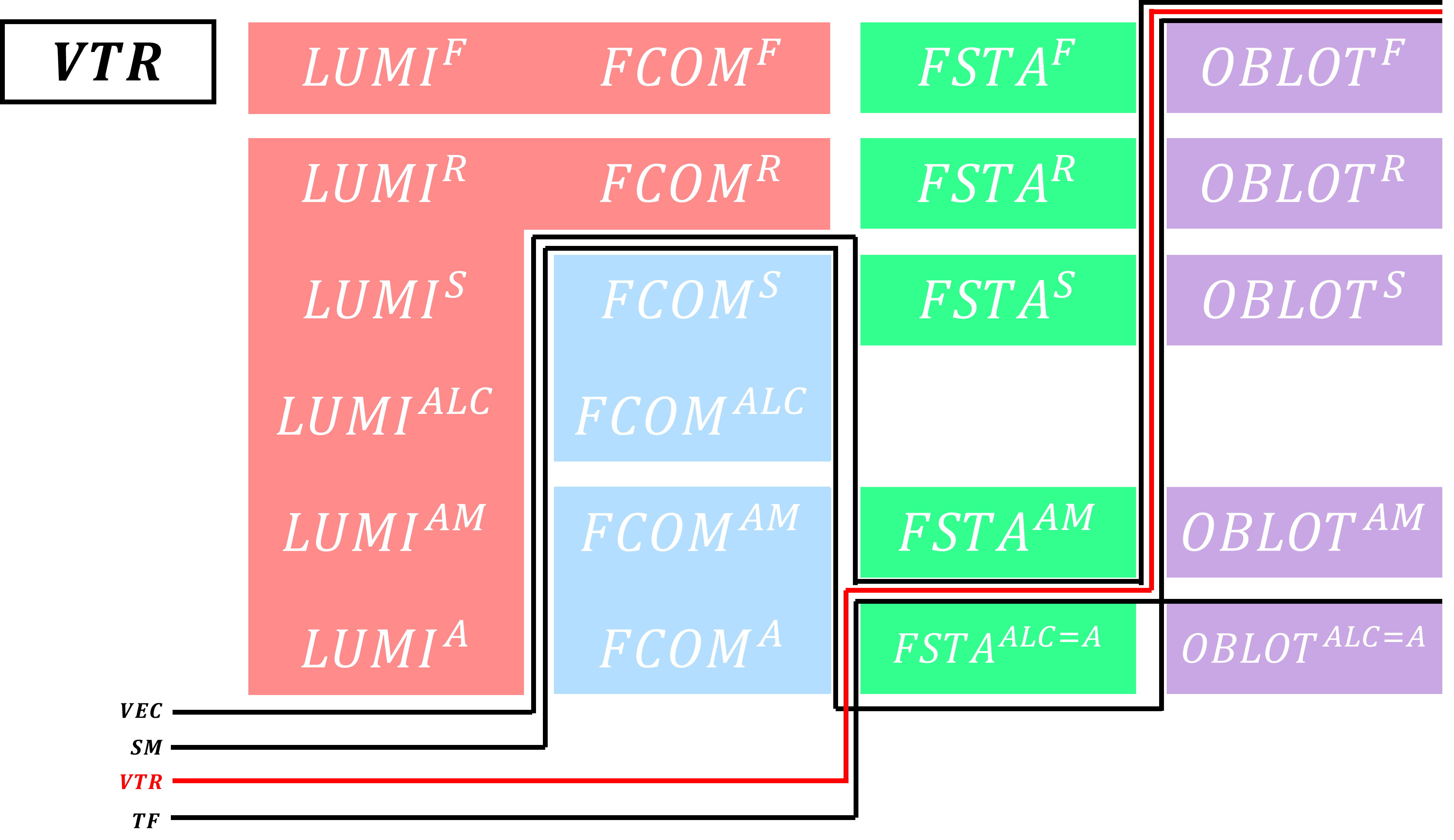}
    \caption{Boundary of VTR}
    \label{fig:VTR_boundary}
\end{figure}
        
\begin{theorem}
    \textit{VTR} is solved in both of $\FS^{A_{M}}$ and $\FC^S$, but not in any of $\OB^F$ or $\FS^{A_{ALC}}$.
    \label{th:vtrt}
\end{theorem}

  \subsection{Full Synchrony or Light Is Necessary (LCM-M)}\label{sec:5-3}
  \label{arc:LCM-M}
\begin{Definition}
    \textit{LCM--M} (Least Common Multiple -- Movement)
    \\Place three robots $A, B, C$ in this order on a straight line. In the initial configuration, the ratio of the distances between the robots is $AB:BC = 2:1$. In this problem, the following behavior is required of robots $B$ and $C$.:
    \begin{itemize}
        \item Robot $B$ shall perform the action "move away from $A$ by the initial distance $d_{AB}$ between $A$ and $B$" two times.
        \item Robot $C$ shall perform the action "move away from $A$ by the initial distance $d_{AC}$ between $A$ and $C$" once.
    \end{itemize}
    Here, let $0 < t_b, 0 < t'_b < t_c$. Let $t_b$ be the time when robot B arrives at a location $2d_{AB}$ away from robot A, and $t'_b$ be the time when robot B arrives at a location $3d_{AB}$ away from robot A. Let $t_c$ be the time when robot C arrives at a location $2d_{AC}$ away from robot A. At this time, \textit{LCM-M} is defined by the following logical expression.
    \begin{align*}
        \textit{LCM--M} \equiv
        \big[ &\{ \forall t \ge 0 : |b(0) - a(0)| \le |b(t) - a(t)| \} \\
        &\land \{ \forall t \ge 0 : |c(0) - a(0)| \le |c(t) - a(t)| \} \\
        &\land \{ \forall t''_b \ge t'_b, \forall t'_c \ge t_c : b(t''_b) = c(t'_c) \} \\
        &\land \{ \forall t''_b \ge t'_b, \forall t \ge 0 : a(t), b(t), c(t) \in \overline{a(0)b(t''_b)} \} \\
        &\land \{ \mathrm{dis}(b(t_b), a(t_b)) = 2 \, \mathrm{dis}(b(0), a(0)) \} \\
        &\land \{ \mathrm{dis}(b(t'_b), a(t'_b)) = 3 \, \mathrm{dis}(b(0), a(0)) \} \\
        &\land \{ \mathrm{dis}(c(t_c), a(t_c)) = 2 \, \mathrm{dis}(c(0), a(0)) \}
        \big]
    \end{align*}
\end{Definition}
    
\begin{figure}[H]
    \centering
    \includegraphics[width=12cm]{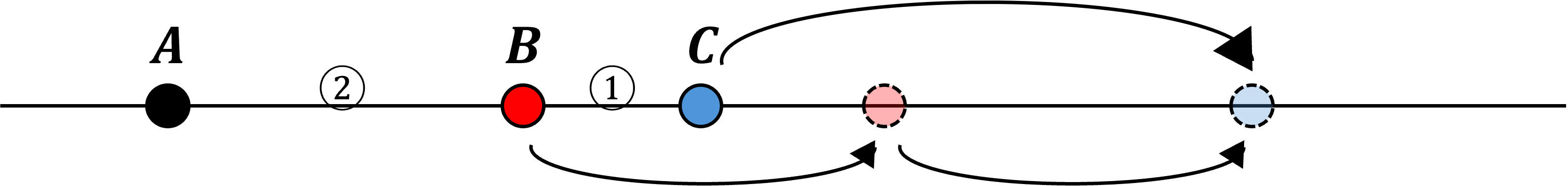}
    \caption{Movement of the \textit{LCM--M} problem.}
    \label{fig:lcmm}
\end{figure}

\begin{lemma}
    $\textit{LCM--M} \notin \OB^R$.
    \label{lem:35}
\end{lemma}
\begin{proof}
    Assume for the sake of contradiction that there exists an algorithm $\mathcal{A}$ that correctly solves the problem.
    Under a round-robin scheduler, a configuration $C'$ may appear in which each of robots $B$ and $C$ has been activated exactly once, regardless of the activation order.
    In this configuration, the ratio of distances between the robots remains $AB:BC=2:1$.
    In the $\OB$ model without lights, this configuration $C'$ cannot be distinguished from the initial configuration.
    For algorithm $\mathcal{A}$ to be correct, the robots' actions in configuration $C'$ must differ from their actions in the initial configuration (since it is not their first activation).
    However, since the robots cannot distinguish $C'$ from the initial configuration, they are forced to perform the same actions as they would on their first activation. This contradicts the correctness of algorithm $\mathcal{A}$.
    Therefore, the initial assumption is false, and no such algorithm $\mathcal{A}$ can exist.
\end{proof}

\begin{lemma}
    $\textit{LCM--M} \in \OB^F$.
    \label{lem:36}
\end{lemma}
\begin{proof}
    The problem is solved by \textbf{Algorithm \ref{alg:LCMM_OBLOT_F}}.
    In the $\OB$ model, robots do not possess lights, and therefore cannot record or recognize the number of times they have been activated. However, under the \FSY\ scheduler, all robots are activated simultaneously and can observe the ratio of distances between robots in the current configuration.  
    By utilizing this distance ratio, each robot can determine whether the current configuration corresponds to the initial state or to one resulting from prior movement, and can act accordingly.  
    Therefore, under full synchrony, robots can execute the specified number of movements correctly.
\end{proof}

\begin{algorithm}
    \caption{Alg \textit{LCM--M} for robot $r$}
    \label{alg:LCMM_OBLOT_F}
    \textbf{Assumptions:} $\OB$, \FSY \\
    
    $Phase Look:$ $A.\text{pos} \gets$ the position of the farthest robot in $r$'s local coordinate system.\\$other.\text{pos} \gets$ the position of the nearest robot in $r$'s local coordinate system.\\
    
    $Phase Compute:$
    \begin{algorithmic}[1]
       \If{$(A\,r) : (r\,other) = 2 : 1$}
            \State $r.\text{des} \gets$ the point at distance $Ar$ away from $A$
        \ElsIf{$(A\,other) : (other\,r) = 3 : 1$}
            \State $r.\text{des} \gets$ the point at distance $Ar/2$ away from $A$
        \ElsIf{$(A\,r) : (r\,other) = 3 : 1$}
            \State $r.\text{des} \gets$ the point at distance $Ar$ away from $A$
        \EndIf 
    \end{algorithmic}

    \vspace{1em}
    $Phase Move:$\\
    Move to $r_i.des$.
\end{algorithm}

\begin{lemma}
    $\textit{LCM--M} \in \FS^{A_{LC}}$.
    \label{lem:37}
\end{lemma}
\begin{proof}
    The problem is solved by \textbf{Algorithm \ref{alg:LCMM_FSTA_ALC}}.
    In the $\FS$ model, robots can recognize the number of times they have been activated using their lights, and by utilizing the ratio of distances between robots, they can execute the specified actions.
\end{proof}

\begin{algorithm}
    \caption{Alg \textit{LCM--M} for robot $r$}
    \label{alg:LCMM_FSTA_ALC}
    \textbf{Assumptions:} $\FS$, $LC$-{\bf atomic}-\ASY \\
    \textbf{Light:} $r.\text{state} \in \{0, 1\} $, initially 0\\
    
    $Phase Look:$ $A.\text{pos} \gets$ the position of the farthest robot in $r$'s local coordinate system.\\
    $other.\text{pos} \gets$ the position of the nearest robot in $r$'s local coordinate system.\\
    
    $Phase Compute:$
    \begin{algorithmic}[1]
       \If{$(A\,r):(r\,other) = 2:1 \land r.\text{state} = 0 $}
            \State $r.\text{state} \gets 1 $
            \State $r.\text{des} \gets $ the point at distance $Ar$ away from $A$
        \ElsIf{$(A\,r):(r\,other) = 2:1 \land r.\text{state} = 1 $}
            \State $r.\text{des} \gets $ the point at distance $Ar/2$ away from $A$
        \ElsIf{$(A\,r):(r\,other) = 3:1 \land r.\text{state} = 0 $}
            \State $r.\text{des} \gets $ the point at distance $Ar$ away from $A$
        \EndIf
    \end{algorithmic}

    \vspace{1em}
    $Phase Move:$\\
    Move to $r_i.des$.
\end{algorithm}

\begin{lemma}
    $\textit{LCM--M} \in \FC^{A_{M}}$.
    \label{lem:38}
\end{lemma}
\begin{proof}
    The problem is solved by \textbf{Algorithm \ref{alg:LCMM_FCOM_AM}}.
    In the $\FC$ model, each robot can recognize the activation counts of other robots via lights, and can perform the designated actions by utilizing the ratio of distances between robots.
\end{proof}

\begin{algorithm}
    \caption{Alg LCM--M for robot $r$}
    \label{alg:LCMM_FCOM_AM}
    \textbf{Assumptions:} $\FC$, $M$-{\bf atomic}-\ASY \\
    \textbf{Light:} $r.\text{state},\ \text{other.state} \in \{0, 1\}$, initially 0\\
    
    $Phase Look:$ $A.\text{pos} \gets$ the position of the farthest robot in $r$'s local coordinate system.\\
    $other.\text{pos} \gets$ the position of the nearest robot in $r$'s local coordinate system.\\
    
    $Phase Compute:$
    \begin{algorithmic}[1]
        \If{$(A r): (r\,other) = 2:1$ and $\text{other.state} = 0$}
            \State $r.\text{state} \gets 1$
            \State $r.\text{des} \gets$ the point at distance $(Ar)$ away from $A$
        \ElsIf{$(A r): (r\,other) = 2:1$ and $\text{other.state} = 1$}
            \State $r.\text{des} \gets$ the point at distance $(Ar)/2$ away from $A$
        \ElsIf{$(A r): (r\,other) = 3:1$ and $\text{other.state} = 1$}
            \State $r.\text{des} \gets$ the point at distance $(Ar)$ away from $A$
        \EndIf
    \end{algorithmic}

    \vspace{1em}
    $Phase Move:$\\
    Move to $r_i.des$.
\end{algorithm}

\clearpage
\begin{figure}[H]
    \centering
    \includegraphics[width=12cm]{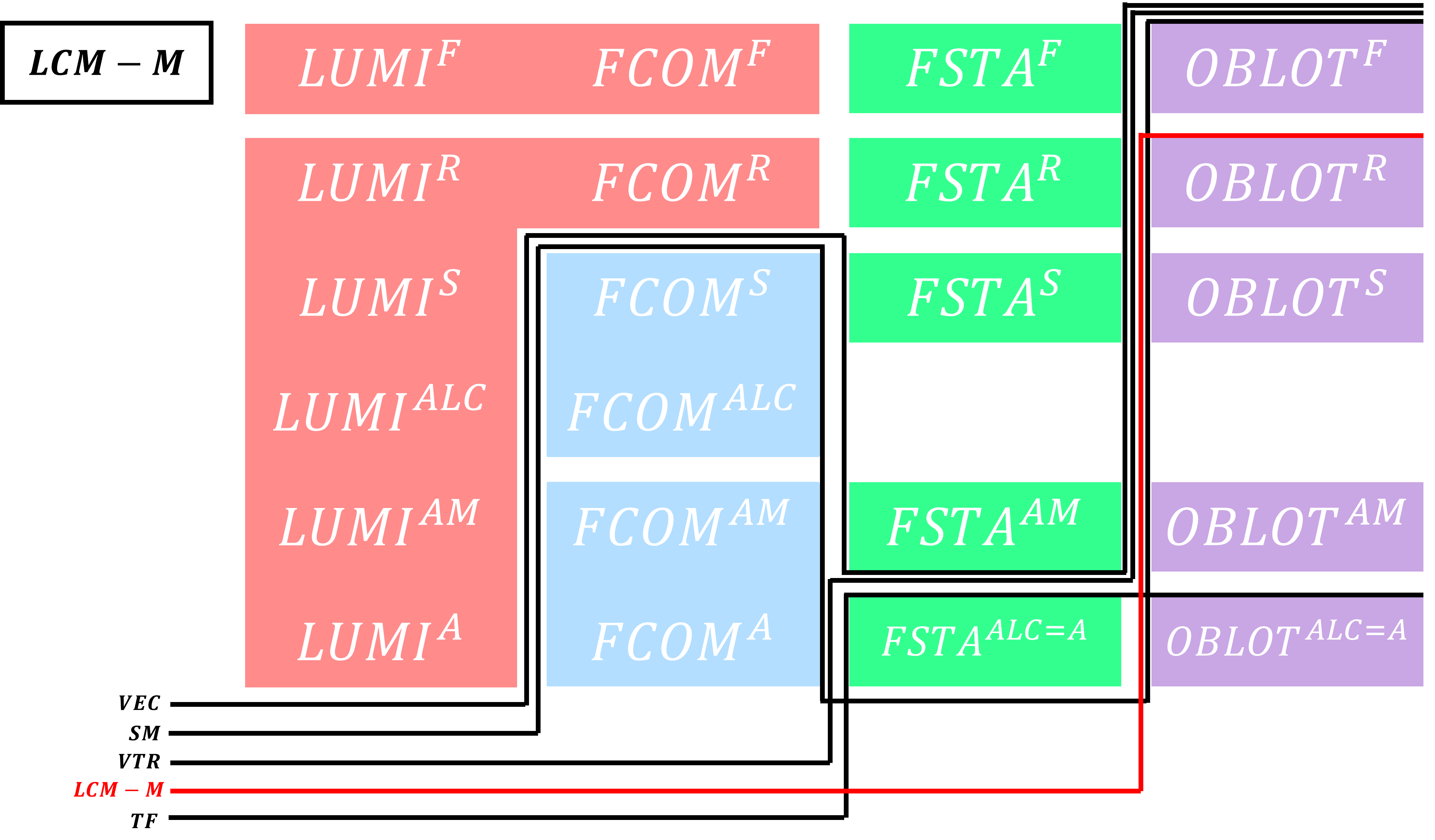}
    \caption{Boundary of LCM--M}
    \label{fig:LCM-M_boundary}
\end{figure}

\begin{theorem}
    \textit{LCM--M} is solved in $\OB^F$, $\FS^{A_{LC}}$ and $\FC^{A_{M}}$, but not in $\OB^R$.
    \label{th:lcmmt}
\end{theorem}

\section{Concluding Remarks}
This paper has extended the known landscape of mobile robot computation by identifying novel triadic separations and structural phenomena that emerge only when comparing three or more robot models simultaneously. 

Fig.~\ref{fig:separation-map} summarizes illustration of the problem landscape for the 14 major robot models, each defined by a pair of internal capability ($\OB$, $\FS$, $\FC$, $\LU$) and scheduler class (\FSY, \RSY, \SSY, \ASY\ or its variants).
Each problem is positioned within the minimal model class that can solve it. 
The grey-shaded region represents the union of two areas: (1) the regions for which separation results were already known in previous work, and (2) the additional regions where new problems were introduced in this paper to complete the landscape.
Problems highlighted in yellow (e.g., \textit{ETE}, \textit{TAR($d$)*}, \textit{LP--MLCv}) are newly defined in this work and serve to establish previously unknown separations. 
White regions remain unexplored or unresolved, and represent potential directions for future work.




\clearpage
\bibliographystyle{plainurl}
\bibliography{referenceorg}

\begin{thebibliography}{10}

\bibitem{AP}
N.~Agmon and D.~Peleg.
\newblock Fault-tolerant gathering algorithms for autonomous mobile robots.
\newblock {\em SIAM Journal on Computing}, 36(1):56--82, 2006.

\bibitem{AOSY}
H.~Ando, Y.~Osawa, I.~Suzuki, and M.~Yamashita.
\newblock A distributed memoryless point convergence algorithm for mobile robots with limited visivility.
\newblock {\em IEEE Transactions on Robotics and Automation}, 15(5):818--828, 1999.

\bibitem{BDT}
Z.~Bouzid, S.~Das, and S.~Tixeuil.
\newblock Gathering of mobile robots tolerating multiple crash faults.
\newblock In {\em the 33rd Int. Conf. on Distributed Computing Systems}, pages 334--346, 2013.

\bibitem{apdcm}
K.~Buchin, P.~Flocchini, I.~Kostitsyna, T.~Peters, N.~Santoro, and K.~Wada.
\newblock Autonomous mobile robots: Refining the computational landscape.
\newblock In {\em APDCM 2021}, pages 576--585, 2021.

\bibitem{BFKPSW-IandC}
K.~Buchin, P.~Flocchini, I.~Kostitsyna, T.~Peters, N.~Santoro, and K.~Wada.
\newblock On the computational power of energy-constrained mobile robots: Algorithms and cross-model analysis.
\newblock {\em Information and Computation}, 303(105280), 2025.

\bibitem{CG}
D.~Canepa and M.~Potop-Butucaru.
\newblock Stabilizing flocking via leader election in robot networks.
\newblock In {\em Proc. 10th Int. Symp. on Stabilization, Safety, and Security of Distributed Systems (SSS)}, pages 52--66, 2007.

\bibitem{CDN}
S.~Cicerone, Di~Stefano, and A.~Navarra.
\newblock Gathering of robots on meeting-points.
\newblock {\em Distributed Computing}, 31(1):1--50, 2018.

\bibitem{CFPS}
M.~Cieliebak, P.~Flocchini, G.~Prencipe, and N.~Santoro.
\newblock Distributed computing by mobile robots: Gathering.
\newblock {\em SIAM Journal on Computing}, 41(4):829--879, 2012.

\bibitem{CP}
R.~Cohen and D.~Peleg.
\newblock Convergence properties of the gravitational algorithms in asynchronous robot systems.
\newblock {\em {SIAM} J. on Computing}, 34(15):1516--1528, 2005.

\bibitem{DFPSY}
S.~Das, P.~Flocchini, G.~Prencipe, N.~Santoro, and M.~Yamashita.
\newblock Autonomous mobile robots with lights.
\newblock {\em Theoretical Computer Science}, 609:171--184, 2016.

\bibitem{DKKOW19}
S.~Dolev, S.~Kamei, Y.~Katayama, F.~Ooshita, and K.~Wada.
\newblock Brief announcement: Neighborhood mutual remainder and its self-stabilizing implementation of look-compute-move robots.
\newblock In {\em 33rd International Symposium on Distributed Computing}, pages 43:1--43:3, 2019.

\bibitem{FPS}
P.~Flocchini, G.~Prencipe, and N.~Santoro.
\newblock {\em Distributed Computing by Oblivious Mobile Robots}.
\newblock Morgan \& Claypool, 2012.

\bibitem{FPSW99}
P.~Flocchini, G.~Prencipe, N.~Santoro, and P.~Widmayer.
\newblock Hard tasks for weak robots: the role of common knowledge in pattern formation by autonomous mobile robots.
\newblock In {\em 10th Int. Symp. on Algorithms and Computation (ISAAC)}, pages 93--102, 1999.

\bibitem{FPSW05}
P.~Flocchini, G.~Prencipe, N.~Santoro, and P.~Widmayer.
\newblock Gathering of asynchronous robots with limited visibility.
\newblock {\em Theoretical Computer Science}, 337(1--3):147--169, 2005.

\bibitem{FPSW08}
P.~Flocchini, G.~Prencipe, N.~Santoro, and P.~Widmayer.
\newblock Arbitrary pattern formation by asynchronous oblivious robots.
\newblock {\em Theoretical Computer Science}, 407:412--447, 2008.

\bibitem{FSSW23-arxiv}
P.~Flocchini, N.~Santoro, Y.~Sudo, and K.~Wada.
\newblock On asynchrony, memory, and communication: Separations and landscapes.
\newblock CoRR abs/2311.03328, arXiv, 2023.

\bibitem{FSW19}
P.~{Flocchini}, N.~{Santoro}, and K.~{Wada}.
\newblock On memory, communication, and synchronous schedulers when moving and computing.
\newblock In {\em Proc. 23rd Int. Conference on Principles of Distributed Systems (OPODIS)}, pages 25:1--25:17, 2019.

\bibitem{FSSW23}
Paola Flocchini, Nicola Santoro, Yuichi Sudo, and Koichi Wada.
\newblock {On Asynchrony, Memory, and Communication: Separations and Landscapes}.
\newblock In {\em 27th International Conference on Principles of Distributed Systems (OPODIS 2023)}, volume 286 of {\em Leibniz International Proceedings in Informatics (LIPIcs)}, pages 28:1--28:23, 2024.

\bibitem{FYOKY}
N.~Fujinaga, Y.~Yamauchi, H.~Ono, S.~Kijima, and M.~Yamashita.
\newblock Pattern formation by oblivious asynchronous mobile robots.
\newblock {\em SIAM Journal on Computing}, 44(3):740--785, 2015.

\bibitem{GP}
V.~Gervasi and G.~Prencipe.
\newblock Coordination without communication: The case of the flocking problem.
\newblock {\em Discrete Applied Mathematics}, 144(3):324--344, 2004.

\bibitem{ISKIDWY}
T.~Izumi, S.~Souissi, Y.~Katayama, N.~Inuzuka, X.~D\'{e}fago, K.~Wada, and M.~Yamashita.
\newblock The gathering problem for two oblivious robots with unreliable compasses.
\newblock {\em SIAM Journal on Computing}, 41(1):26--46, 2012.

\bibitem{OWD}
T.~Okumura, K.~Wada, and X.~D\'{e}fago.
\newblock Optimal rendezvous $\mathcal{L}$-algorithms for asynchronous mobile robots with external-lights.
\newblock In {\em Proc. 22nd Int. Conference on Principles of Distributed Systems (OPODIS)}, pages 24:1--24:16, 2018.

\bibitem{SIW}
S.~Souissi, T.~Izumi, and K.~Wada.
\newblock Oracle-based flocking of mobile robots in crash-recovery model.
\newblock In {\em Proc. 11th Int. Symp. on Stabilization, Safety, and Security of Distributed Systems (SSS)}, pages 683--697, 2009.

\bibitem{SY}
I.~Suzuki and M.~Yamashita.
\newblock Distributed anonymous mobile robots: Formation of geometric patterns.
\newblock {\em SIAM Journal on Computing}, 28:1347--1363, 1999.

\bibitem{YS}
M.~Yamashita and I.~Suzuki.
\newblock Characterizing geometric patterns formable by oblivious anonymous mobile robots.
\newblock {\em Theoretical Computer Science}, 411(26--28):2433--2453, 2010.

\bibitem{YUKY}
Y.~Yamauchi, T.~Uehara, S.~Kijima, and M.~Yamashita.
\newblock Plane formation by synchronous mobile robots in the three-dimensional euclidean space.
\newblock {\em J. ACM}, 64:3(16):16:1--16:43, 2017.

\end{thebibliography}

\end{document}